\documentclass[ejs,preprint]{imsart}
\setattribute{journal}{name}{}
\RequirePackage[round]{natbib}
\RequirePackage[colorlinks,citecolor=blue,urlcolor=blue]{hyperref}

\usepackage{bbm}
\usepackage{comment}
\usepackage{amsmath, amsthm, amssymb,stmaryrd}
\usepackage{algpseudocode}
\usepackage{algorithm}
\usepackage{graphicx}
\usepackage{subfigure}
\usepackage{lineno}

\let\hat\widehat
\newtheorem{thm}{Theorem}
\newtheorem{lem}[thm]{Lemma}

\theoremstyle{remark}
\newtheorem{remark}{Remark}

\newcommand\R{\mathbb{R}}
\newcommand\D{\mathbb{D}}
\newcommand\K{\mathbb{K}}
\newcommand\E{\mathbb{E}}

\newcommand\G{\mathbb{G}}

\newcommand\B{\mathbb{B}}
\newcommand\X{\mathbb{X}}
\newcommand\Y{\mathbb{Y}}
\newcommand\W{\mathbb{W}}
\renewcommand\P{\mathbb{P}}

\newcommand\Haus{{\sf Haus}}

\newcommand{\var}{\mathrm{Var}}

\newcommand\cL{{\cal L}}
\newcommand\cF{{\cal F}}

\newcommand\cR{{\cal R}}
\newcommand\cG{{\cal G}}
\newcommand\cN{{\cal N}}

\catcode`@=11
\newskip\beforeproofvskip
\newskip\afterproofvskip
\beforeproofvskip=\medskipamount
\afterproofvskip=\bigskipamount

\def\prooftag{Proof}
\def\proofskip{\enspace}

\def\proof{\@ifnextchar[{\@@proof}{\@proof}}  
\def\@startproof{\par\vskip\beforeproofvskip\leavevmode}
\def\@proof{\@startproof{\scshape\prooftag.}\proofskip}
\def\@@proof[#1]{\@startproof {\scshape\prooftag #1.}\proofskip}

\catcode`@=12

\let\hat\widehat
\let\tilde\widetilde

\begin{document}

%
\begin{frontmatter}

\title{Nonparametric Inference via Bootstrapping the Debiased Estimator}
\runtitle{Bootstrapping the Debiased Estimator}

\begin{aug}
  \author{\fnms{Gang}
    \snm{Cheng}\ead[label=e1]{gangc@uw.edu}}
    \and
  \author{\fnms{Yen-Chi}
    \snm{Chen}\ead[label=e2]{yenchic@uw.edu}}
  \affiliation{Department of Statistics\\University of Washington}
  \runauthor{G. Cheng and Y.-C. Chen}
  \address{Department of Statistics\\University of Washington\\
Box 354322\\ Seattle, WA 98195 \\
          \printead{e1}}
    \address{Department of Statistics,\\ University of Washington\\
Box 354322\\ Seattle, WA 98195 \\
          \printead{e2}}
\end{aug}

\begin{abstract}
In this paper, we propose to construct confidence bands 
by bootstrapping the debiased kernel density estimator (for density estimation) 
and the debiased local polynomial regression estimator (for regression analysis).
The idea of using a debiased estimator was 
{
recently employed 
by \cite{calonico2018effect}}
to construct a confidence interval of the density function 
(and regression function) at a given point
by explicitly estimating stochastic variations.
We extend their ideas of using the debiased estimator and further propose a bootstrap approach for
constructing simultaneous confidence bands.
This modified method has an advantage that we can easily choose the smoothing bandwidth
from conventional bandwidth selectors and the confidence band will be asymptotically valid. 
We prove the validity of the bootstrap confidence band and generalize it to
density level sets and inverse regression problems.
Simulation studies confirm the validity of the proposed confidence bands/sets.
We apply our approach to an Astronomy dataset to show its applicability.
\end{abstract}



\end{frontmatter}

\section{Introduction}
In nonparametric statistics, how to construct a confidence band 
has been a central research topic for several decades.
However, this problem has not yet been fully resolved because
of its intrinsic difficulty.
The main issue is that 
the nonparametric estimation error generally contains a bias part and a stochastic variation part.
Stochastic variation can be captured using a limiting distribution or a resampling approach, such as the bootstrap
\citep{Efron1979}.
However, the bias is not easy to handle because it often involves higher-order derivatives of
the underlying function and cannot be easily captured by resampling methods (see, e.g., page 89 in \citealt{wasserman2006all}). 

To construct a confidence band, two main approaches are proposed in the literature.
The first one is to undersmooth the data so the bias converges faster than the stochastic variation 
\citep{bjerve1985uniform, hall1992bootstrap, hall1993empirical, chen1996empirical,wasserman2006all}. 
Namely, we choose the tuning parameter (e.g., the smoothing bandwidth in the kernel estimator) 
in a way such that the bias shrinks faster than the stochastic variation. 
Because the bias term is negligible compared 
to the stochastic variation, the resulting confidence band is (asymptotically) valid.
However, the conventional bandwidth selector (e.g., the ones described in \citealt{sheather2004density})
does not give an undersmoothing bandwidth so it is unclear how to practically implement this method.
The other approach estimates the bias and then constructs a confidence band after correcting the bias
{\citep{hardle1988bootstrapping,hardle1991bootstrap,hall1992effect, eubank1993confidence, sun1994simultaneous, hardle1995better, neumann1995automatic, xia1998bias, hardle2004bootstrap}.}
The second approach is sometimes called a \emph{debiased}, or bias-corrected, approach.
Because the bias term often involves higher-order derivative of the targeted function, 
we need to introduce another estimator of the derivatives to correct the bias
and obtain a consistent bias estimator.
Estimating the derivatives involves a non-conventional smoothing bandwidth (often we have to oversmooth the data)
so it is not easy to choose it in practice (there are some methods discussed in \citealt{Chacon2011}). 



In this paper, we introduce a simple approach to constructing confidence bands
for both density and regression functions by bootstrapping a debiased estimator,
which can be viewed as a synthesis of both the debiased and the undersmoothing methods. 
Our method is featured with the fact that one can use a conventional smoothing bandwidth selector, which does not
involve an explicitly undersmoothing nor oversmoothing. 
We use the kernel density estimator (KDE) to estimate the density function
and local polynomial regression for inferring the regression function. 
Our method is based on the debiased estimator proposed in \cite{calonico2018effect},
where the authors propose a confidence interval of a fixed point using 
an explicit estimation of the errors.
However, they consider univariate density and their approach is only valid for a given point, which limits the applicability.
We generalize their idea to multivariate densities and propose using the bootstrap
to construct a confidence band that is uniform for every point in the support. 
Thus, our method could be viewed as a debiased approach.
A feature of this debiased estimator 
is that we
are able to construct a confidence band even without a consistent bias estimator.  
Thus, our approach requires only one single tuning parameter-the smoothing bandwidth-and this tuning parameter is compatible with
most off-the-shelf bandwidth selectors, such as the rule of thumb in
the KDE or cross-validation in regression \citep{fan1996local,wasserman2006all,scott2015multivariate}. 
Further, we prove that after correcting for the bias in the usual KDE, the bias of the debiased KDE is now on a higher order than the usual KDE, while the stochastic variation for the debiased KDE is still on the same order as the usual KDE. Thus, choosing bandwidth by balancing bias and stochastic variation for the usual KDE turns out to be undersmoothing for the debiased KDE.  
This leads to a simple but elegant approach of constructing a valid confidence band with a uniform coverage over the entire support.
Note that \cite{bartalotti2017bootstrap} also used a bootstrap approach with the debiased estimator
to construct a CI. But their focus is on inferring the regression function of a given point under the regression discontinuity design problem.


As an illustration, consider
Figure~\ref{fig::ex01}, where we apply the nonparametric bootstrap with $L_\infty$ metric to construct
confidence bands.
We consider one example for density estimation and one example for regression.
In the first example (top row of Figure~\ref{fig::ex01}),
we have 
a size $2000$ random sample from a Gaussian mixture, such that 
with a probability of $0.6$, a data point is generated from the standard normal and with a probability of $0.4$, a data point is
from a normal centered at $4$.  We want to compare the coverage performance of usual KDE and the debiased KDE.  
We choose the smoothing bandwidth using the rule of thumb \citep{Silverman1986} { of the usual KDE, then use this bandwidth to}
estimate the density using both usual KDE and the debiased KDE, and then use the bootstrap to construct a $95\%$ confidence band. 
In the left two panels, we display one example of the confidence band for the population density function (black curve)
with a confidence band from bootstrapping the usual KDE (red band) and that from bootstrapping the debiased KDE (blue band).
The right panel shows the coverage of the bootstrap confidence band under various nominal levels.
For the second example (bottom row of Figure~\ref{fig::ex01}), we 
consider estimating the regression function of $Y = \sin(\pi\cdot X)+\epsilon$, where $\epsilon\sim N(0,0.1^2)$
and $X$ is from a uniform distribution on $[0,1]$. 
We generate $500$ points and apply the local linear smoother to estimate the regression function.
We select the smoothing bandwidth by repeating a $5$-fold cross validation of the local linear smoother.
Then we estimate the regression function using both the local linear smoother (red) and the debiased local linear smoother (blue)
and apply the empirical bootstrap to construct $95\%$ confidence bands. 
In both cases, we see that bootstrapping the usual estimator does not yield an asymptotically valid confidence band,
but bootstrapping the debiased estimator gives us a valid confidence band with nominal coverages. 
It is worth mentioning that in both density estimation and regression analysis case, the debiased method only requires one bandwidth which is the same bandwidth as the original method. This illustrates the obvious convenience of our method. 


\begin{figure}
\centering
\includegraphics[height=1.5in]{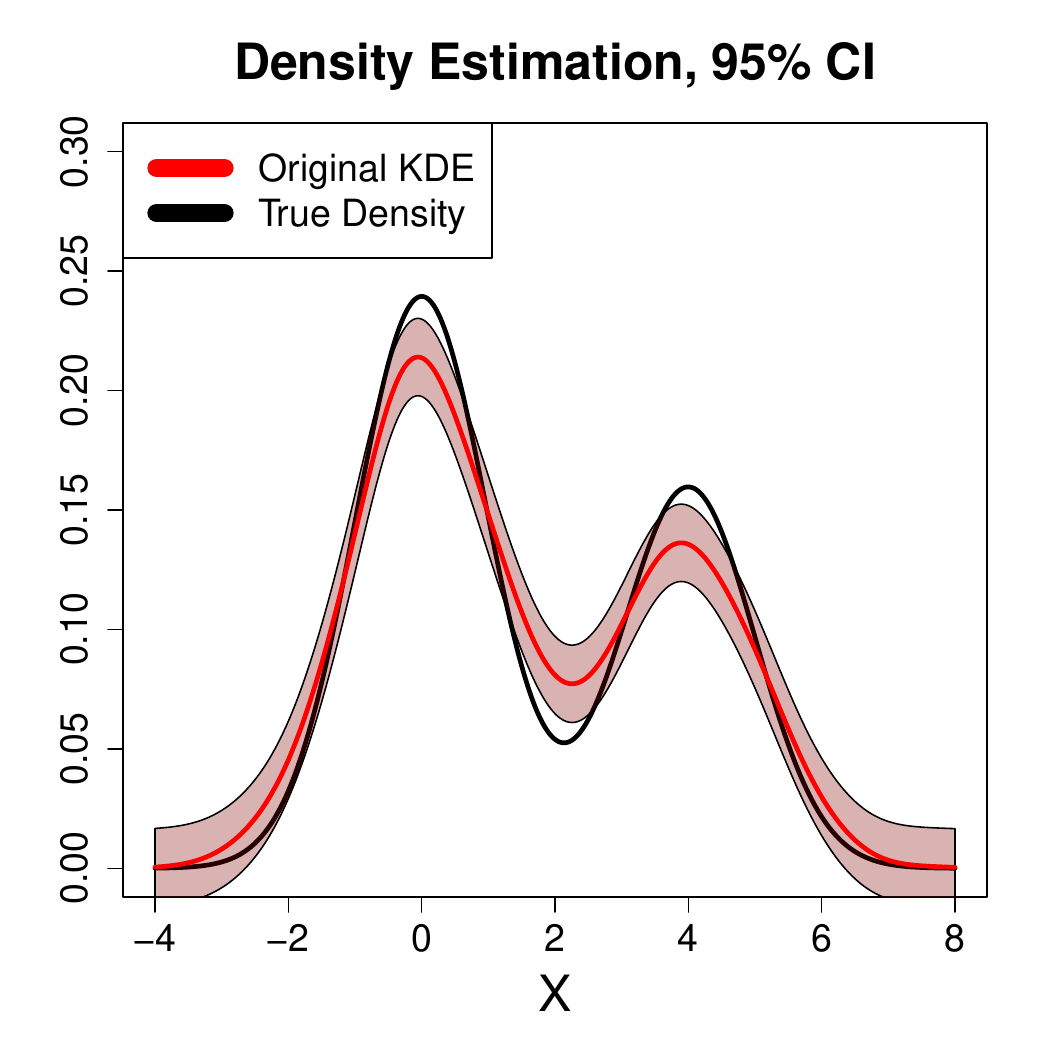}
\includegraphics[height=1.5in]{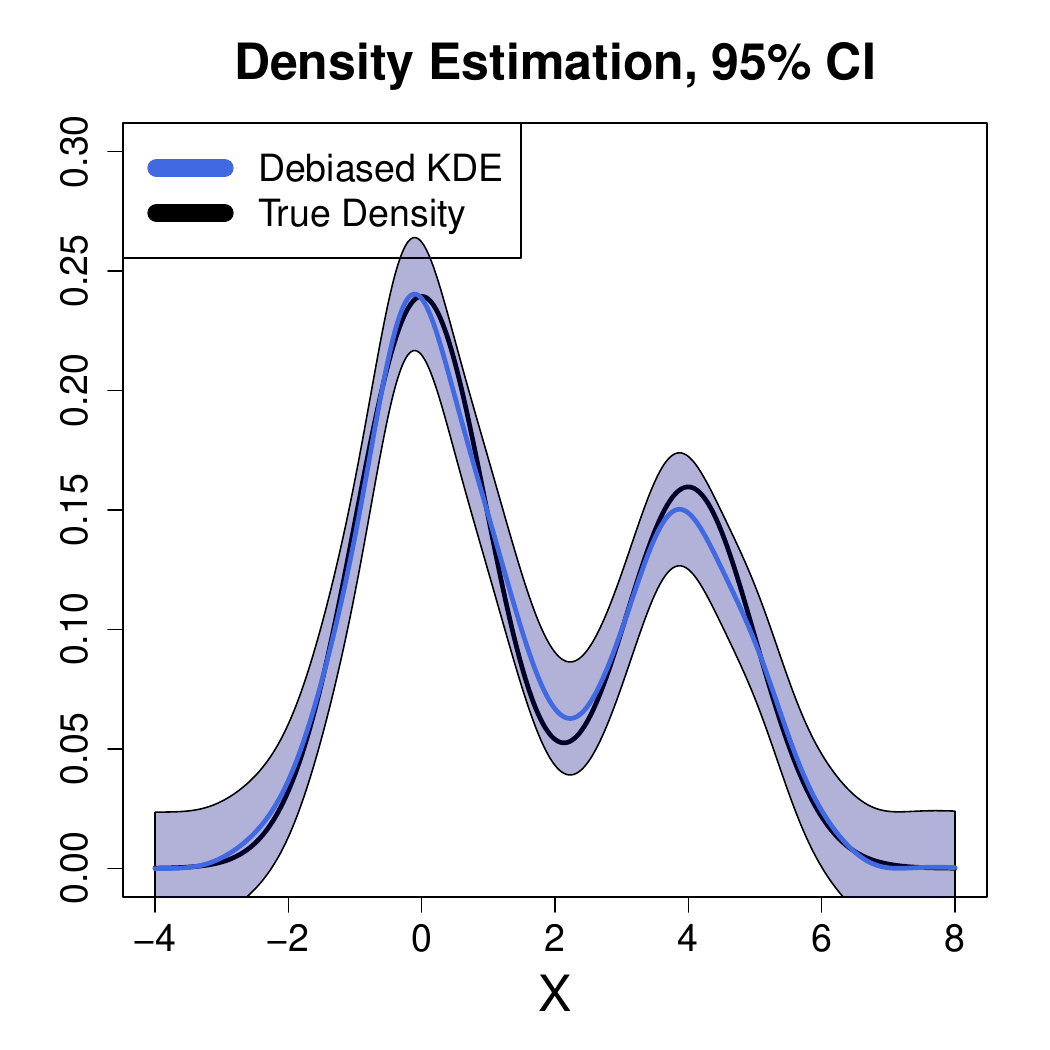}
\includegraphics[height=1.5in]{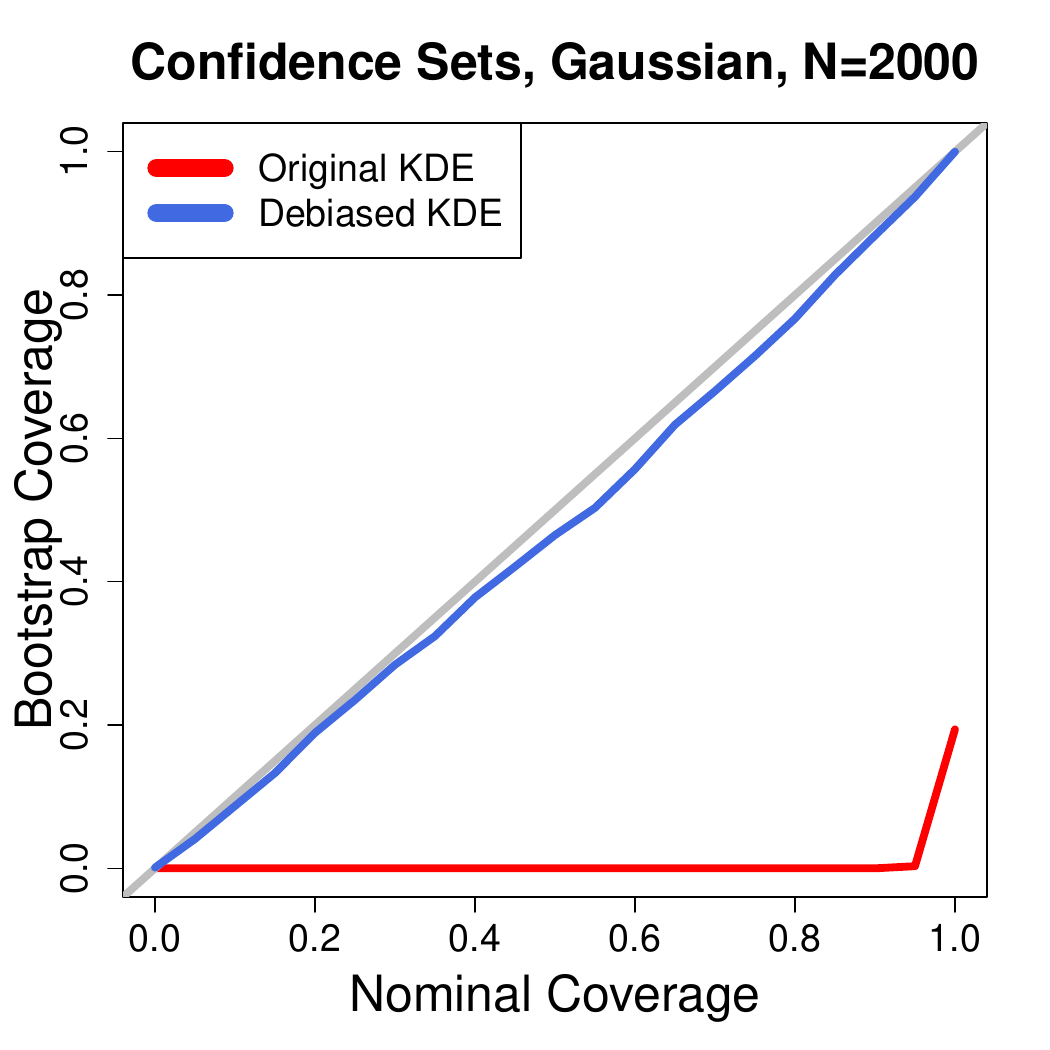}\\
\includegraphics[height=1.5in]{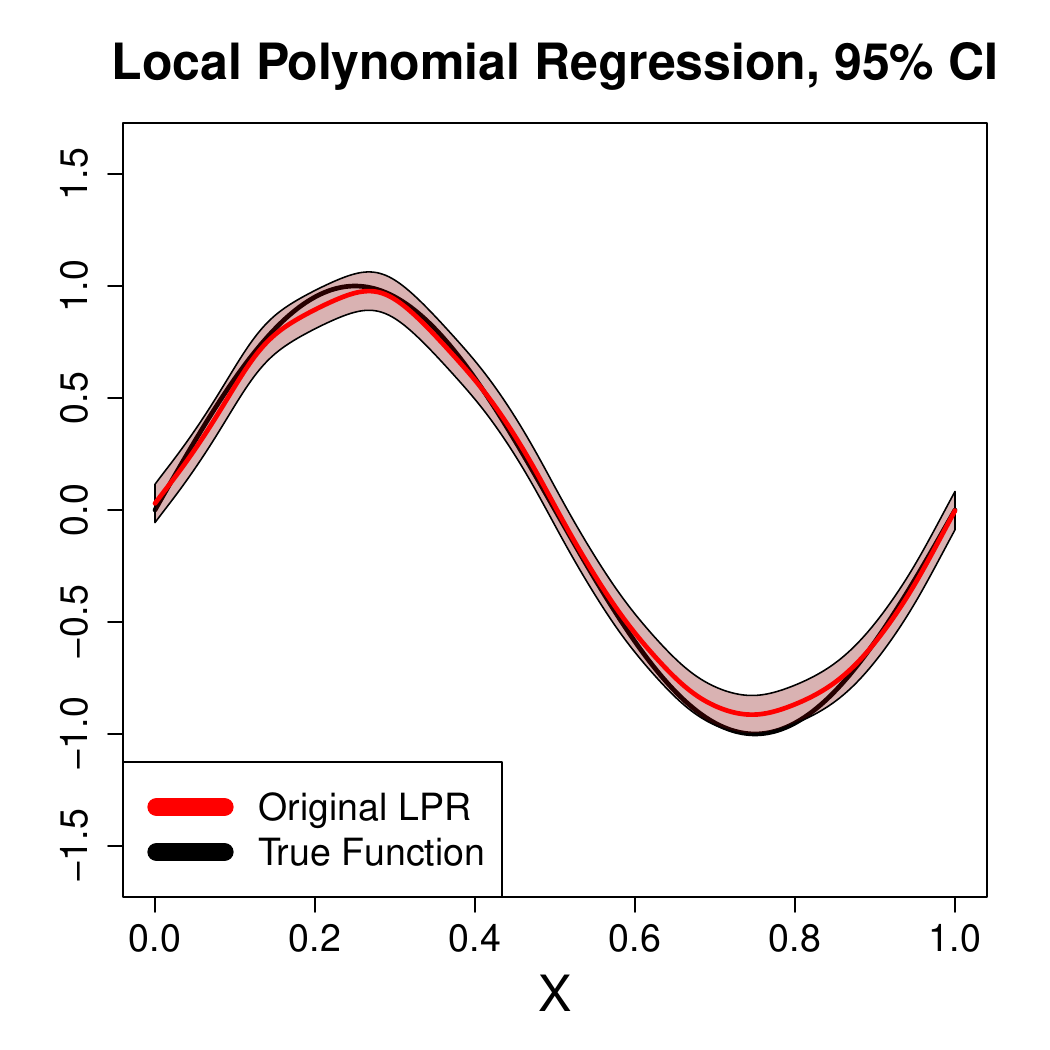}
\includegraphics[height=1.5in]{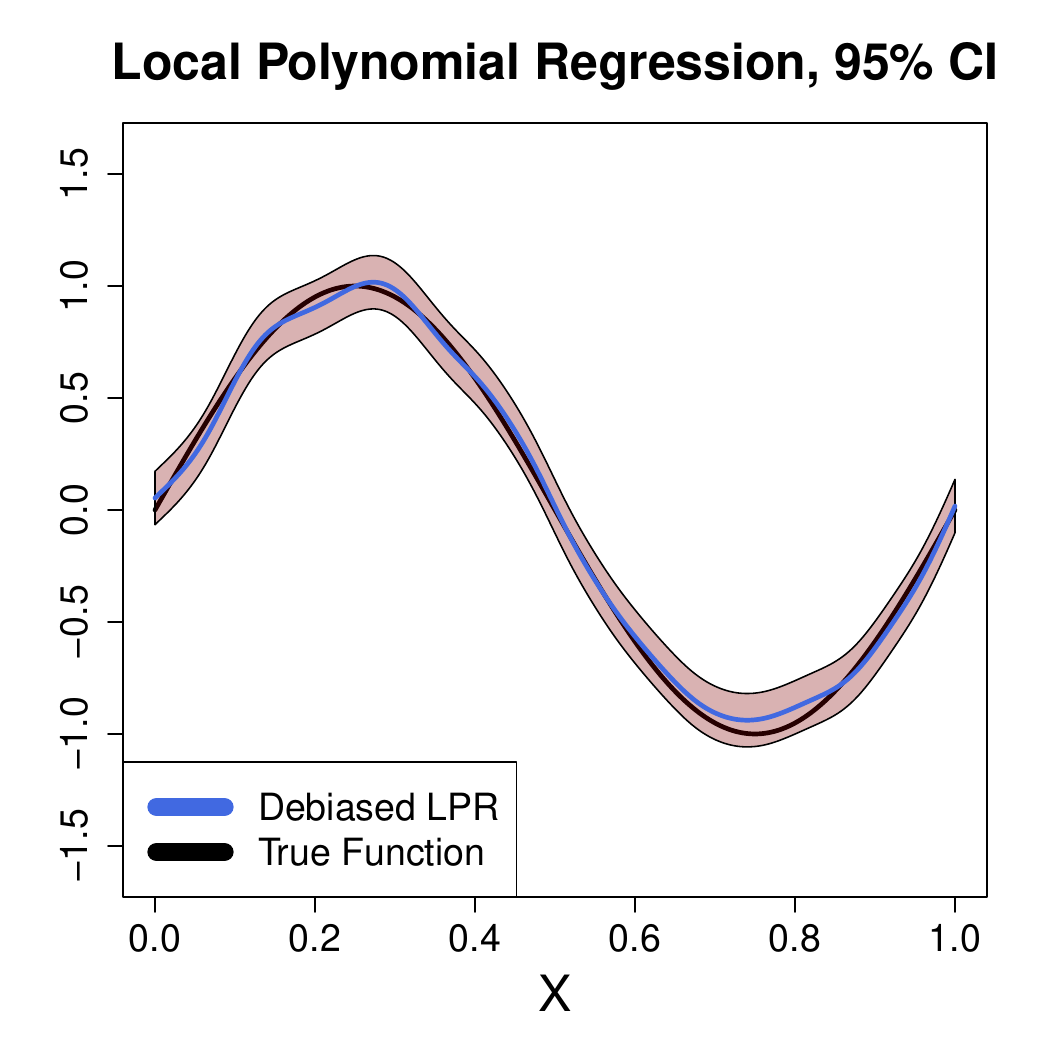}
\includegraphics[height=1.5in]{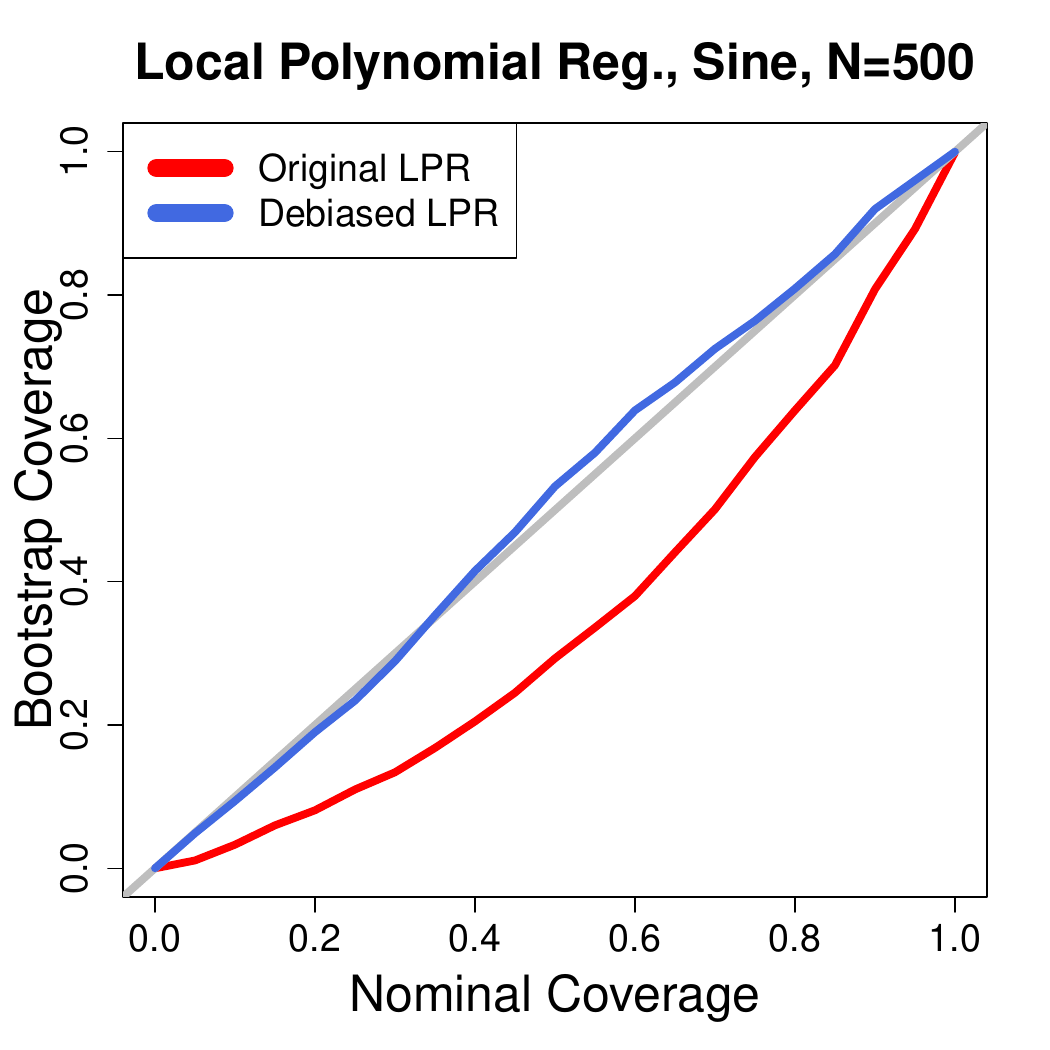}

\caption{Confidence bands from bootstrapping the usual estimator versus bootstrapping the debiased estimator.
In the top row, we consider estimating the density function of a Gaussian mixture. 
And in the bottom row, we consider estimating the regression function of a sine structure. 
One each row, the left two panels displayed one instance of $95\%$ bootstrap confidence band for both original and the debiased estimators, the right panel
shows the coverage of bootstrap confidence band under different nominal levels.}
\label{fig::ex01}
\end{figure}

\emph{Main Contributions.}
\begin{itemize}
\item 
We propose our confidence bands for both density estimation and regression problems
(Section~\ref{sec::CI::KDE} and \ref{sec::LS}). 
\item We generalize these confidence bands to both density level set and inverse regression problems (Section~\ref{sec::LV} and \ref{sec::inv}). 
\item We derive the convergence rate of the debiased estimators under uniform loss (Lemma~\ref{lem::KDE_infty} and \ref{lem::LS_empirical}).
\item We derive the asymptotic theory of the debiased estimators and prove the consistency of confidence bands
(Theorem~\ref{thm::KDE_gaussian}, \ref{thm::KDE_CI}, \ref{thm::LS_gaussian}, and \ref{thm::LS_CI}).
\item We use simulations to show that our confidence bands/sets are indeed asymptotically valid
and apply our approach to an Astronomy dataset to demonstrate the applicability (Section~\ref{sec::simulation}). 
\end{itemize}

\emph{Related Work.}
Our method is 
{inspired by} the pilot work in \cite{calonico2018effect}.
Our confidence band is a bias correction (debiasing) method, which is a common method for constructing confidence bands of nonparametric
estimators.
The confidence sets about level sets and inverse regression are related to 
\cite{lavagnini2007statistical,bissantz2009asymptotic,birke2010confidence,tang2011two,mammen2013confidence,chen2017density}.

\emph{Outline.}
In Section~\ref{sec::debiased}, we give a brief review of the debiased estimator proposed in \cite{calonico2018effect}. 
In Section~\ref{sec::CI}, we propose our approaches for constructing confidence bands
of density and regression functions and generalize these approaches to density level sets and inverse regression problems.
In Section~\ref{sec::theory}, we derive a convergence 
rate for the debiased estimator and prove the consistency of confidence bands.
In Section~\ref{sec::simulation}, we use simulations to demonstrate that our proposed confidence bands/sets
are indeed asymptotically valid. 
Finally, we conclude this paper and discuss some possible future directions in Section~\ref{sec::discuss}.

\section{Debiased Estimator}	\label{sec::debiased}
Here we briefly review the debiased estimator
of the KDE and local polynomial regression proposed in 
\cite{calonico2018effect}.

\subsection{Kernel Density Estimator}
Let $X_1,\cdots,X_n$ be IID from an unknown density function $p$ with a 
support $\K \subset \mathbb{R}^d$, $p$ is at least second-order continuously differentiable. 
The (original) KDE is
$$
\hat{p}_h(x) = \frac{1}{nh^d}\sum_{i=1}^n K\left(\frac{x-X_i}{h}\right),
$$
where $K(x)$ is a smooth function known as the kernel function 
and $h>0$ is the smoothing bandwidth.
Here we will assume $K(x)$ to be a second-order kernel function such as Gaussian
because this is a common scenario that practitioners are using. 
One can extend the idea to higher-order kernel functions. 

The bias of $\hat{p}_h$ often involves the Laplacian of the density, $\nabla^2 p(x)$, 
we define an estimator of it using another smoothing bandwidth $b>0$ as
$$
\hat{p}^{(2)}_b(x) = \frac{1}{nb^{d+2}}\sum_{i=1}^n K^{(2)}\left(\frac{x-X_i}{b}\right),
$$
where $K^{(2)}(x) = \nabla^2 K(x)$ is the Laplacian (second derivative) of the kernel function $K(x)$.

Let $\tau = \frac{h}{b}$.
Formally, the \emph{debiased KDE} is 
\begin{equation}
\begin{aligned}
\hat{p}_{\tau,h}(x) &= \hat{p}_h(x) - \frac{1}{2}c_K\cdot h^2\cdot \hat{p}^{(2)}_b(x)\\
& = \frac{1}{nh^d}\sum_{i=1}^n K\left(\frac{x-X_i}{h}\right) - \frac{1}{2}\cdot c_K\cdot h^2\cdot \frac{1}{nb^{d+2}}\sum_{i=1}^n K^{(2)}\left(\frac{x-X_i}{b}\right)\\
& = \frac{1}{nh^d}\sum_{i=1}^n M_\tau\left(\frac{x-X_i}{h}\right),
\end{aligned}
\label{eq::dKDE}
\end{equation}
where 
\begin{equation}
M_\tau(x) = K(x) -\frac{1}{2} c_K \cdot \tau^{d+2} \cdot K^{(2)}(\tau\cdot x),
\label{eq::M}
\end{equation}
and $c_K = \int x^2 K(x)dx$. 
Note that when we use the Gaussian kernel, $c_K=1$.
The function $M_\tau(x)$ can be viewed as a new kernel function, which we called the
\emph{debiased kernel function}.
Actually, this kernel function is a higher-order kernel function \citep{scott2015multivariate,calonico2018effect}.
Note that the second quantity $\frac{1}{2}c_K\cdot h^2\cdot \hat{p}^{(2)}_b(x)$
is an estimate for the asymptotic bias in the KDE. 
An important remark is that \emph{we allow $\tau\in (0,\infty)$} to be a fixed number and still have a valid confidence band.  
In practice, we often choose $h=b$ ($\tau=1$) for simplicity and it works well in our experiements.
{
Because the estimator in equation \eqref{eq::dKDE} 
uses the same smoothing bandwidth for both density and bias estimations,
it does not provide a consistent estimate of the second derivative (bias)
so it is not a traditional debiased estimator. 
}

With a fixed $\tau $, we only need one bandwidth for the debiased estimator, {which is designed for the original KDE. } 
Note that when using the MISE-optimal bandwidth for the usual KDE $h = O(n^{-1/(d+4)})$,
$\hat{p}^{(2)}_b(x)$ may \emph{not} be a consistent estimator of $p^{(2)}_b(x)$ since the variance of $\hat{p}^{(2)}(x)$ is at the order of $O(1)$.  
Although it is not a consistent estimator, it is unbiased in the limit. 
Thus, adding this term to the original KDE trades the bias of $\hat{p}_h(x)$ into the stochastic variability of $\hat{p}_{\tau,h}(x)$
and knock the bias into the next order, which is an important property that allows us to choose $h$ and $b$ to be of the same order.
For statistical inference, as long as we can use resampling methods to
capture the variability of the estimator, we are able to construct a valid confidence band.


%

\subsection{Local Polynomial Regression}
Now we introduce the debiased estimator for the local polynomial regression \citep{fan1996local,wasserman2006all}.
For simplicity, we consider the local linear smoother (local polynomial regression with degree $1$) and assume that 
the covariate has dimension $1$.
One can generalize this method into a higher-order local polynomial regression and
multivariate covariates.

Let $(X_1,Y_1),\cdots, (X_n,Y_n)$ be the observed random sample for the covariate $X_i\in\D\subset \R$
and the response $Y_i\in\R$.
The parameter of interest is
the regression function $r(x) = \E(Y_i|X_i=x)$.

The local linear smoother estimates $r(x)$ by
\begin{equation}
\hat{r}_h(x) = \sum_{i=1}^n \ell_{i,h}(x) Y_i,
\label{eq::LS}
\end{equation}
with 
\begin{align*}
\ell_{i,h}(x) &= \frac{\omega_{i,h}(x)}{\sum_{j=1}^n\omega_{j,h}(x)} \\
\omega_{i,h}(x)&= K\left(\frac{x-X_i}{h}\right)(S_{n,h,2}(x)-(X_i-x)S_{n,h,1}(x))\\
S_{n,h,j}(x)&= \sum_{i=1}^n(X_i-x)^{j}K\left(\frac{x-X_i}{h}\right), \quad j=1,2,
\end{align*}
where $K(x)$ is the kernel function and $h>0$ is the smoothing bandwidth.

To debias $\hat{r}_h(x)$, we use the local polynomial regression
for estimating the second derivative $ r''(x)$.
We consider the third-order local polynomial regression estimator of $r''(x)$ \citep{fan1996local, xia1998bias}, 
which is given by
\begin{equation}
\hat{r}^{(2)}_b(x) = \sum_{i=1}^n \ell_{i,b}(x,2)Y_i
\label{eq::2LS}
\end{equation}
with 
\begin{align*}
\ell_b(x,2)^T &= \left(\ell_{1,b}(x,2),\cdots,\ell_{n,b}(x,2)\right)\in\R^n\\
&= 2! e_3^T (X^T_x W_{b,x} X_x)^{-1}X^T_xW_x,
\end{align*}
where
\begin{align*}
e_3^T& = (0,0,1,0),\\
X_x &=\begin{pmatrix}
  1 & X_1-x & \cdots & (X_1-x)^3\\
  1 & X_2-x & \cdots & (X_2-x)^3 \\
  \vdots  & \vdots  & \ddots & \vdots  \\
  1 & X_n-x & \cdots &(X_n-x)^3
 \end{pmatrix}\in \R^{n\times 4},\\
 W_{b,x}&= {\sf Diag}\left(K\left(\frac{x-X_1}{b}\right),\cdots, K\left(\frac{x-X_n}{b}\right)\right)\in\R^{n\times n}.
\end{align*}
Namely, $\hat{r}^{(2)}_b(x)$ is the local polynomial regression estimator of second derivative $r^{(2)}(x)$
using smoothing bandwidth $b>0$.

By defining $\tau= h/b$, the \emph{debiased local linear smoother} is 
\begin{equation}
\hat{r}_{\tau,h}(x) = \hat{r}_h(x) - \frac{1}{2}\cdot c_K\cdot h^2\cdot \hat{r}^{(2)}_{h/\tau}(x),
\label{eq::dLS}
\end{equation}
where $c_K = \int x^2 K(x)dx$ is the same as the constant used in the debiased KDE.
Note that in practice, we often choose $h=b (\tau=1)$. 
Essentially, the debiased local linear smoother uses $\hat{r}^{(2)}_{h/\tau}(x)$
to correct the bias of the local linear smoother $\hat{r}_h(x)$.

\begin{remark}
One can also construct a debiased estimator using the kernel regression (Nadaraya-Watson estimator; \citealt{nadaraya1964estimating}).
However, because the bias of the kernel regression has an extra design bias term
$$
\frac{1}{2}c_K\cdot h^2\cdot \frac{r'(x)p'(x)}{p(x)},
$$ 
the debiased estimator will be more complicated.
We need to estimate $r'(x), p'(x),$ and $p(x)$ to correct the bias.

\end{remark}

\section{Confidence Bands}	\label{sec::CI}

\subsection{Inference for Density Function}	\label{sec::CI::KDE}

\begin{figure}
\fbox{\parbox{\textwidth}{
\begin{center}
{\sc Confidence Bands of Density Function}
\end{center}
\begin{center}
\begin{enumerate}
\item 
Choose the smoothing bandwidth $h_{RT}$ by a standard approach such as the rule of thumb or cross-validation \citep{Silverman1986,sheather1991reliable,sheather2004density}. 
\item 
Compute the debiased KDE $\hat{p}_{\tau,h_{RT}}$ with a fixed value $\tau$ (in general, we choose $\tau=1$).
\item
Bootstrap the original sample for $B$ times and compute the bootstrap debiased KDE
$$
\hat{p}^{*(1)}_{\tau,h_{RT}}, \cdots, \hat{p}^{*(B)}_{\tau,h_{RT}}.
$$
\item
Compute the quantile 
$$
\hat{t}_{1-\alpha} = \hat{F}^{-1}(1-\alpha), \quad \hat{F}(t) = \frac{1}{B}\sum_{j=1}^BI\left(\|\hat{p}^{*(j)}_{\tau,h_{RT}}-\hat{p}_{\tau, h_{RT}}\|_\infty<t\right).
$$
\item 
Output the confidence band 
$$
\hat{C}_{1-\alpha}(x) = \left[\hat{p}_{\tau, h}(x) - \hat{t}_{1-\alpha},\,\, \hat{p}_{\tau, h}(x) + \hat{t}_{1-\alpha}\right].
$$
\end{enumerate}
\end{center}
}}
\caption{Confidence bands of the density function.}
\label{fig::alg::KDE}
\end{figure}

Here is how we construct our confidence bands of density function.
Given the original sample $X_1,\cdots, X_n$, we apply the empirical bootstrap \citep{Efron1979} to generate
the bootstrap sample
$X_1^*,\cdots,X_n^*$. 
Then we apply the debiased KDE \eqref{eq::dKDE} with the bootstrap sample to obtain
the bootstrap debiased KDE. 
\begin{equation}
\begin{aligned}
\hat{p}^*_{\tau,h}(x) 
 = \frac{1}{nh^d}\sum_{i=1}^n M_\tau\left(\frac{x-X^*_i}{h}\right),
\end{aligned}
\label{eq::dKDE_bt}
\end{equation}
where $M_\tau$ is the debiased kernel defined in equation \eqref{eq::M}.
Finally, we compute the bootstrap $L_\infty$ metric $\left\|\hat{p}^*_{\tau,h}-\hat{p}_{\tau,h}\right\|_\infty$,
where $\|f\|_\infty = \sup_{x}|f(x)|$.

Let $\hat{F}(t) =P\left(\left\|\hat{p}^{*}_{\tau,h}-\hat{p}_{\tau,h}\right\|_\infty\leq t|X_1,\cdots,X_n\right)$ 
be the distribution of the bootstrap $L_\infty$ metric
and let $\hat{t}_{1-\alpha}$ be the $(1-\alpha)$ quantile of $\hat{F}(t)$.
Then a $(1-\alpha)$ confidence band of $p$ is 
$$
\hat{C}_{1-\alpha}(x) = \left[\hat{p}_{\tau, h}(x) - \hat{t}_{1-\alpha}, \,\,\hat{p}_{\tau, h}(x) + \hat{t}_{1-\alpha}\right].
$$ 

In Theorem~\ref{thm::KDE_CI}, we prove that this is an asymptotic valid confidence band of $p$
when $\frac{ nh^{d+4}}{\log n}\rightarrow c_0 \geq 0$ for some $c_0 < \infty$ and some other regularity conditions for bandwidth $h$ hold. 
Namely, we will prove
$$
P\left( p(x) \in \hat{C}_{1-\alpha}(x)\,\, \forall x\in\K \right) = 1-\alpha+o(1).
$$
The constraint on the smoothing bandwidth allows us to choose $h = O(n^{-1/(d+4)})$, which is the rate of
most bandwidth selectors in the KDE literature \citep{Silverman1986,sheather1991reliable,sheather2004density, hall1983}.  
Thus, we can choose the tuning parameter using one of these standard methods and bootstrap the debiased estimators
to construct a confidence band.
{ Note for our purpose of inference, the bandwidth was chosen to optimize the original KDE}. 
Though the construction of a confidence band is simple,
it leads to a band with a simultaneous coverage.
Figure~\ref{fig::alg::KDE} provides a summary of the proposed procedure. 

Note that one can replace the KDE using the local polynomial density estimator
and the resulting confidence band is still valid.
The validity of the confidence band follows from the validity of 
the confidence band of the local linear smoother (Theorem~\ref{thm::LS_CI}).

\begin{remark}	\label{rm::honest0}
{
An alternative approach to constructing confidence band is via bootstrapping
a weighted $L_\infty$ statistic such that the difference $\hat{p}_{\tau,h}-p$
is inversely weighted according to an estimate of its variance.
This leads to a variable bandwidth confidence band.
For one concrete example, 
we consider using $\hat{\sigma}_{rbc}$, 
the estimated variance of $\hat{p}_{\tau,h}$ in 
\cite{calonico2018effect}
to construct a variable-width confidence band.
Specifically, we bootstrap
$$
\left\|\frac{\hat{p}_{\tau,h}-p}{\hat{\sigma}_{rbc}}\right\|_\infty,
$$
where $\sigma^2_{rbc} = (nh^d) \var(\hat{p}_{\tau, h})  = \frac{1}{h^d} \left[ \E[M_{\tau} \left(\frac{x - X_i}{h}\right)^2] - \E^2[M_{\tau} \left(\frac{x - X_i}{h}\right)]\right]$ and naturally $\hat{\sigma}_{rbc}^2 = \frac{1}{h^d} \left[ \frac{1}{n} \sum_{i=1}^n M_{\tau}^2 \left( \frac{x - X_i}{h}\right) - \left(\frac{1}{n} \sum_{i=1}^n M_{\tau} \left( \frac{x - X_i}{h} \right)  \right)^2 \right]$.  
$\hat{\sigma}^2_{rbc}$ is non-asymptotic and the above statistic is exactly the studentization quantity proposed in 
\cite{calonico2018effect} to take into account the additional variability introduced by bias term.
We choose $\tilde{t}_{1-\alpha}$ as the $1-\alpha$ quantile of 
$$
\left\|\frac{\hat{p}^*_{\tau,h}-\hat{p}_{\tau,h}}{\hat{\sigma}^*_{rbc}}\right\|_\infty
$$
and construct a confidence band using
\begin{equation}
\tilde{C}_{1-\alpha}(x) = \left[\hat{p}_{\tau, h}(x) - \tilde{t}_{1-\alpha}\hat{\sigma}_{rbc}(x), \,\,\hat{p}_{\tau, h}(x) +\tilde{t}_{1-\alpha}\hat{\sigma}_{rbc}(x)\right].
\label{eq::honest}
\end{equation}
A feature of this confidence band is that the width of the resulting confidence band depends on $x$
and by a similar derivation as Theorem~\ref{thm::KDE_CI},
it is also an asymptotically valid confidence band (more details are given in Appendix \ref{sec::RM2}).   
}
\end{remark}


\begin{remark}
In a sense, the debiased estimator is similar to the debiased lasso
\citep{javanmard2014confidence,
van2014asymptotically,
zhang2014confidence}
where we add an extra term to the original estimator to correct the bias
so that the stochastic variation dominates the estimation error. 
Then the stochastic variation can be estimated using either a limiting distribution
or a bootstrap, which leads to a (asymptotically) valid confidence band. 

\end{remark}

\subsubsection{Inference for Density Level Sets}	\label{sec::LV}
In addition to the confidence band of $p$, bootstrapping the debiased KDE gives us
a confidence set of the \emph{level set} of $p $.
Let $\lambda$ be a given level.
We define
$$
D=\{x: p(x)= \lambda\}
$$
as the $\lambda$-level set of $p $ \citep{Polonik1995,Tsybakov1997}.

A simple estimator for $D$ is the plug-in estimator based on the debiased KDE:
$$
\hat{D}_{\tau, h} = \{x: \hat{p}_{\tau,h}(x)=\lambda\}. 
$$
Under regularity conditions, a consistent density estimator leads to a consistent
level set estimator \citep{Polonik1995,Tsybakov1997,Cuevas2006,Rinaldo2010b,qiao2017asymptotics}. 

Now we propose a confidence set of $D$ based on bootstrapping the debiased KDE. 
We will use the method proposed in \cite{chen2017density}.
To construct a confidence set for $D$, we introduce the \emph{Hausdorff distance} which is defined as
$$
\Haus(A,B) = \max\left\{\sup_{x\in A}d(x,B), \sup_{x\in B}d(x,A)\right\}. 
$$
The Hausdorff distance is like an $L_\infty$ metric for sets.

Recall that $\hat{p}^*_{\tau,h}$ is the bootstrap debiased KDE. 
Let $\hat{D}^*_{\tau,h} = \{x: \hat{p}^*_{\tau,h}(x)=\lambda\}$ be the plug-in estimator of $D$
using the bootstrap debiased KDE.
Now define $\hat{t}^{LV}_{1-\alpha}$ to be the $1-\alpha$ quantile of the distribution of the bootstrap Hausdorff distance
$$
\hat{F}^{LV}(t) =P\left(\Haus(\hat{D}^{*}_{\tau,h}, \hat{D}_{\tau,h})<t|X_1,\cdots,X_n\right).
$$
Then a $(1-\alpha)$ confidence set of $D$ is 
$$
\hat{D}_{\tau,h} \oplus \hat{t}^{LV}_{1-\alpha},
$$
where $A\oplus r =\{x: d(x,A)\leq r\}$ for a set $A$ and a scalar $r>0$.
In Theorem~\ref{thm::LV_CI}, we prove that this is an asymptotically valid confidence set of $D$.

\begin{remark}	\label{rm::LV::alternative}
\cite{mammen2013confidence} 
proposed an alternative way to construct
confidence sets for the level sets by inverting the confidence bands of KDE.
They proposed using
$$
\{x: |\hat{p}_h(x)-\lambda|<\epsilon_{n,\alpha}\}
$$
as a confidence set of $D$,
where $\epsilon_{n,\alpha}$ is some suitable quantity computed from the data. 
This idea also works for the debiased KDE; we 
can construct a confidence set as
$$
\left\{x: |\hat{p}_{\tau,h}(x)-\lambda|<\hat{t}_{1-\alpha}\right\},
$$
where $\hat{t}_{1-\alpha}$ is the $1-\alpha$ quantile of bootstrap $L_\infty$ metric given in Section~\ref{sec::CI::KDE}.
Moreover, Theorem~\ref{thm::KDE_CI} implies
that this is also an asymptotically valid confidence set.
\end{remark}

\subsection{Inference for Regression Function}	\label{sec::LS}

\begin{figure}
\fbox{\parbox{\textwidth}{
\begin{center}
{\sc Confidence Bands of Regression Function}
\end{center}
\begin{center}
\begin{enumerate}
\item 
Choose the smoothing bandwidth $h_{CV}$ by cross-validation ($5$-fold or $10$-fold) or other bandwidth selector 
with the usual local linear smoother.
\item 
Compute the debiased local linear smoother $\hat{r}_{\tau, h_{CV}}$ with a fixed value $\tau$ (in general, we choose $\tau=1$).
\item
Bootstrap the original sample for $B$ times and compute the bootstrap debiased local linear smoother 
$$
\hat{r}^{*(1)}_{\tau,h_{CV}}, \cdots, \hat{r}^{*(B)}_{\tau,h_{CV}}.
$$
\item
Compute the quantile 
$$
\hat{s}_{1-\alpha} = \hat{G}^{-1}(1-\alpha), \quad \hat{G}(s) = \frac{1}{B}\sum_{j=1}^BI\left(\|\hat{r}^{*(j)}_{\tau,h_{CV}}-\hat{r}_{\tau, h_{CV}}\|_{\infty}<s\right).
$$
\item 
Output the confidence band 
$$
 \hat{C}^R_{1-\alpha}(x) = \left[\hat{r}_{\tau, h}(x) - \hat{s}_{1-\alpha},\,\, \hat{r}_{\tau, h}(x) + \hat{s}_{1-\alpha}\right]
$$
\end{enumerate}
\end{center}
}}
\caption{Confidence bands of the regression function.}
\label{fig::alg::reg}
\end{figure}

Now we turn to the confidence band for the regression function $r(x)$.
Again we propose using the empirical bootstrap (in the regression case it is also
known as the paired bootstrap) to estimate $r(x)$.
Other bootstrap methods, such as the multiplier bootstrap (also known as the wild bootstrap; \citealt{wu1986jackknife})
or the residual bootstrap \citep{freedman1981bootstrapping}, will also work under slightly different assumptions.
Recall that $\hat{r}_{\tau,h}(x)$ is the debiased local linear smoother.

Given the original sample $(X_1,Y_1),\cdots,(X_n,Y_n)$, we generate a bootstrap sample, denoted as
$(X_1^*,Y_1^*),\cdots,(X_n^*,Y_n^*)$. 
Then we compute the debiased local linear smoother using the bootstrap sample 
to get the bootstrap debiased local linear smoother $\hat{r}^*_{\tau,h}(x)$.
Let 
$\hat{s}_{1-\alpha}$ be the $(1-\alpha)$ quantile of 
the distribution
$$
\hat{G}(s) = P\left(\|\hat{r}^*_{\tau,h}-\hat{r}_{\tau,h}\|_{\infty}< s|X_1,\cdots,X_n\right). 
$$
Then a $(1-\alpha)$ confidence band of $r(x)$ is
$$
 \hat{C}^R_{1-\alpha}(x) = \left[\hat{r}_{\tau, h}(x) - \hat{s}_{1-\alpha}, \,\,\hat{r}_{\tau, h}(x) + \hat{s}_{1-\alpha}\right].
$$
That is,
the confidence band is the debiased local linear smoother plus/minus the bootstrap quantile. 
The bottom left panel of Figure~\ref{fig::ex01} shows an example of the confidence band.

In Theorem \ref{thm::LS_CI}, we prove that $\hat{r}_{\tau,h}\pm \hat{s}_{1-\alpha}$ is indeed
an asymptotic $1-\alpha$ confidence band of the regression function $r(x)$ when $h\rightarrow 0, \frac{nh^5}{\log n}\rightarrow c_0\geq 0$ for some $c_0$ bounded and some other regularity conditions for bandwidth hold.
i.e.
$$
P\left( r(x) \in \hat{C}^R_{1-\alpha}(x) \,\,\forall x\in\D \right) = 1-\alpha+o(1).
$$
The condition on smoothing bandwidth 
is compatible with the optimal rate of the usual local linear smoother ($h = O(n^{-1/5})$) \citep{li2004cross, xia2002asymptotic}.
Thus, we suggest choosing the smoothing bandwidth by cross-validating the original local linear smoother.
This leads to a simple but valid confidence band. 
We can also use other bandwidth selectors such as those introduced in
Chapter 4 of \cite{fan1996local}; these methods all yield a bandwidth at rate $O(n^{-1/5})$,
which works for our approach.
Figure~\ref{fig::alg::reg} summarizes the above procedure of constructing a confidence band.

\subsubsection{Inference for Inverse Regression}	\label{sec::inv}

The debiased local linear smoother can be used to construct confidence sets of the inverse regression problem 
\citep{lavagnini2007statistical,bissantz2009asymptotic,birke2010confidence,tang2011two}.
Let $r_0$ be a given level, the inverse regression finds the collection of points $\cR$
such that 
$$
\cR  = \{x: r(x) = r_0\}.
$$
Namely, $R$ is the region of covariates such that the regression function $r(x)$ equals $r_0$, a fixed level.
Note that the inverse regression is also known as the calibration problem 
\citep{brown1993measurement,gruet1996nonparametric,weisberg2005applied} and regression level set \citep{cavalier1997nonparametric,Laloe2012}.

A simple estimator of $R$ is the plug-in estimator from the debiased local linear smoother:
$$
\hat{\cR}_{\tau, h} = \left\{x: \hat{r}_{\tau,h}(x)=r_0\right\}. 
$$
\cite{Laloe2012} proved that $\hat{\cR}_{\tau,h}$ is a consistent estimator of $\cR$ under smoothness assumptions.

To construct a confidence set of $\cR$, we propose the following bootstrap confidence set.
Recall that $\hat{r}^*_{\tau,h}(x)$ is the bootstrap debiased local linear smoother
and let
$$
\hat{\cR}^*_{\tau, h} = \left\{x: \hat{r}^*_{\tau,h}(x)=r_0\right\}
$$
be the plug-in estimator of $\cR$. 
Let $\hat{s}^{R}_{1-\alpha}$ be the $(1-\alpha)$ quantile of the distribution
$$
\hat{G}^R(s) = P\left(\Haus(\hat{\cR}_{\tau, h},\hat{\cR}_{\tau, h})< s|X_1,\cdots,X_n\right). 
$$
Then an asymptotic confidence set of $\cR$ is
$$
\hat{\cR}_{\tau, h}\oplus \hat{s}^{R}_{1-\alpha}= \{x\in\K: d(x,\hat{\cR}_{\tau,h})\leq \hat{s}^R_{1-\alpha}\}.
$$
In Theorem~\ref{thm::IR_CI}, we prove that $\hat{\cR}^*_{\tau, h}\oplus \hat{s}^{R}_{1-\alpha}$
is indeed an asymptotically valid $(1-\alpha)$ confidence set of $\cR$.

When $\cR$ contains only one element, say $x_0$, asymptotically
the estimator $\hat{\cR}_{\tau,h}$ will contain only one element $\hat{x}_0$. 
Moreover, 
$\sqrt{nh}(\hat{x}_0-x_0)$ converges to a mean $0$ normal distribution.
Thus, we can use the bootstrap $\hat{\cR}^*_{\tau, h}$ to estimate the variance of $\sqrt{nh}(\hat{x}_0-x_0)$
and use the asymptotic normality to construct
a confidence set. 
Namely, we use
$$
[\hat{x}_0 + z_{\alpha/2} \cdot \hat{\sigma}_R ,\,\,  \hat{x}_0 + z_{1-\alpha/2} \cdot\hat{\sigma}_R]
$$
as a confidence set of $x_0$,
where $z_{\alpha}$ is the $\alpha$ quantile of a standard normal distribution
and $\hat{\sigma}_R$ is the bootstrap variance estimate. 
We will also compare the coverage of confidence sets using this approach in
Section~\ref{sec::simulation}.

Similar to Remark \ref{rm::LV::alternative}, 
an alternative method of the confidence set of the inverse regression is given by inverting
the confidence and of the regression function:
$$
\left\{x: |\hat{m}_{\tau,h}(x)-r_0| <\hat{s}_{1-\alpha}\right\},
$$
where $\hat{s}_{1-\alpha}$ is the bootstrap $L_\infty$ metric of
the debiased local linear smoother (Section~\ref{sec::LS}). 
As long as we have an asymptotically valid confidence band of $m(x)$, 
the resulting confidence set of inverse regression is also asymptotically valid.

\cite{bissantz2009asymptotic} and \cite{birke2010confidence} suggested constructing
confidence sets of $\cR$ by undersmoothing.
However, undersmoothing is not compatible with many common bandwidth selectors
for regression analysis and the size will shrink at a slower rate.
On the other hand,
our method does not require any undersmoothing and later we will prove that
the smoothing bandwidth from cross-validation $h_{CV}$ is compatible with our method
(Theorem~\ref{thm::IR_CI}).
Thus, we can simply choose $h_{CV}$ as the smoothing bandwidth and
bootstrap the estimators to construct the confidence set.

\section{Theoretical Analysis}	\label{sec::theory}

\subsection{Kernel Density Estimator}	\label{sec::KDE}

For a multi-index vector $\beta =
(\beta_1,\ldots,\beta_d)$ of non-negative integers, we define
$|\beta| = \beta_1 + \beta_2 + \cdots + \beta_d$
and the corresponding derivative operator
\begin{equation}
D^\beta = \frac{\partial^{\beta_1}}{\partial x_1^{\beta_1}} \cdots 
\frac{\partial^{\beta_d}}{\partial x_d^{\beta_d}},
\label{eq::d1}
\end{equation}
where $D^\beta f$ is often written as $f^{[\beta]}$.
For a real number $\ell$, let $\lfloor\ell\rfloor$ be the largest integer strictly less than $\ell$.
For any given $\xi,L>0$,
we define the H\"older Class $\Sigma(\xi,L)$ 
(Definition 1.2 in \citealt{Tsybakov1997})
as the collection of functions such that
$$
\Sigma(\xi,L) =\left\{f: |f^{[\beta]}(x)-f^{[\beta]}(y)|\leq L|x-y|^{\xi-|\beta|}, \,\forall \beta\,\, s.t.\,\, |\beta|=  \lfloor\xi\rfloor\right\}.
$$

To derive the consistency of confidence bands/sets, we need the following assumptions.\\
{\bf Assumptions.}
\begin{itemize}
\item[(K1)] $K(x)$ is a second order kernel function, symmetric and has at least second-order bounded derivative
and
$$
\int \|x\|^2 K^{[\beta]}(\|x\|) dx <\infty, \qquad \int  \left(K^{[\beta]}(\|x\|)\right)^2 dx <\infty,
$$
where $K^{[\beta]}$ is partial derivative of $K$ with respect to the multi-index vector $\beta=(\beta_1,\cdots,\beta_d)$
and 
for $|\beta| \leq 2$.  

\item[(K2)] Let 
\begin{align*}
\mathcal{K}_\gamma &= \left\{y\mapsto K^{[\beta]}\left(\frac{\|x-y\|}{h}\right): x\in\mathbb{R}^d, |\beta|=\gamma, h>0\right\},
\end{align*}
where $K^{[\beta]}$ is defined in equation \eqref{eq::d1}
and $\mathcal{K}^*_\ell = \bigcup_{\gamma=0}^\ell \mathcal{K}_\gamma$. 
We assume that $\mathcal{K}^*_2$ is a VC-type class. i.e., 
there exist constants $A,v$, and a constant envelope $b_0$ such that
\begin{equation}
\sup_{Q} N(\mathcal{K}^*_2, \cL^2(Q), b_0\epsilon)\leq \left(\frac{A}{\epsilon}\right)^v,
\label{eq::VC}
\end{equation}
where $N(T,d_T,\epsilon)$ is the $\epsilon$-covering number for a
semi-metric set $T$ with metric $d_T$ and $\cL^2(Q)$ is the $L_2$ norm
with respect to the probability measure $Q$.
\item[(P)] The density function $p$  is bounded and in H\"older Class $\Sigma(2+\delta_0,L_0)$ 
for some constant $L_0>0$ and $2\geq\delta_0> 2/3$ with a compact support $\mathbb{K} \subset \mathbb{R}^d$.  
Further, for any $x_0$ on the boundary of $\K$,
$p(x_0) =0$ and $\nabla p(x_0) = 0$.
\item[(D)] The gradient on the level set $D=\{x: p(x)=\lambda\}$ is bounded from zero; i.e.,
$$
\inf_{x\in D} \|\nabla p(x)\|\geq g_0>0
$$
for some $g_0$.
\end{itemize}

(K1) is a common and mild condition on kernel functions \citep{wasserman2006all,scott2015multivariate}. 
The specific form of bias estimation depends on the order of the kernel function. 
(K2) is also a weak assumption to control the complexity of kernel functions so  
we have uniform consistency on density, gradient, and Hessian estimation
\citep{Gine2002,Einmahl2005,genovese2009path,genovese2014nonparametric,chen2015asymptotic}.
Note that many common kernel functions, such as the Gaussian kernel, satisfy this assumption.
(P) involves two parts; a smoothness assumption and a boundary assumption.
We can interpret the smoothness assumption as requiring a smooth second-order derivative of the density function.
Note that the lower bound on $\delta_0$ ($\delta_0>2/3$) is to make sure the bias of a debiased estimator
is much smaller than the stochastic variation so our confidence band is valid. 
When $\delta_0>2$, our procedure is still valid but the bias of the debiased KDE 
will be at rate $O(h^4)$ and will not be of a higher order. 
The boundary conditions of (P) are needed 
to regularize the bias on the boundary. 
(D) is a common assumption in the level set estimation literature to ensure level sets are $(d-1)$ dimensional hypersurfaces; see, e.g.,
\cite{Cadre2006}, \cite{chen2017density}, and \cite{qiao2017asymptotics}.

Our first result is the pointwise bias and variance of the debiased KDE.
\begin{lem}[Pointwise bias and variance]
Assume (K1) and (P)
and $\tau \in(0,\infty)$ is fixed.
Then the bias and variance of $\hat{p}_{\tau,h}$ is at rate
\begin{align*}
\E\left(\hat{p}_{\tau,h}(x)\right) - p(x) &= O(h^{2+\delta_0})\\
{\sf Var}(\hat{p}_{\tau,h}(x)) & = O\left(\frac{1}{nh^d} \right).
\end{align*}
\label{lem::KDE_BV}
\end{lem}

Lemma~\ref{lem::KDE_BV} is consistent with \cite{calonico2018effect} and it shows an interesting result:
the bias of the debiased KDE has rate $O(h^{2+\delta_0})$
and its stochastic variation has the same rate as the usual KDE. 
This means that the debiasing operation kicks the bias of the density estimator into
the next order and keeps the stochastic variation as the same order.
Moreover, this also implies that the optimal bandwidth for the debiased KDE
is $h = O(n^{-\frac{1}{d+4+2\delta_0}})$, which corresponds to oversmoothing the usual KDE.
This is because when $\tau$ is fixed, 
the debiased KDE is actually a KDE with a fourth-order kernel function \citep{calonico2018effect}. 
Namely, the debiased kernel $M_\tau$ is a fourth-order kernel function.
Thus, the bias is pushed to the order $O(h^{2+\delta_0})$ rather than the usual rate $O(h^2) $.

Using the empirical process theory, we can further derive the convergence rate under the $L_\infty$ error.
\begin{lem}[Uniform error rate of the debiased KDE]
Assume (K1-2) and (P) holds, and $\tau \in(0,\infty)$ is fixed, and $h = n^{-\frac{1}{\varpi}}$ for some $\varpi > 0$ such that $\frac{nh^{d+4}}{\log n}\rightarrow c_0 \geq 0$ for some $c_0$ bounded and $\frac{nh^d}{\log n} \rightarrow \infty$.
Then 
\begin{align*}
\left\|\hat{p}_{\tau,h} -p\right\|_{\infty} & = O(h^{2+\delta_0}) + O_P\left(\sqrt{\frac{\log n}{nh^d}}\right).
\end{align*}
\label{lem::KDE_infty}
\end{lem}

To obtain a confidence band, we need to study the $L_\infty$ error
of the estimator $\hat{p}_{\tau,h}$.
Recall from \eqref{eq::dKDE},
\begin{align*}
\hat{p}_{\tau,h}(x)
& = \frac{1}{nh^d}\sum_{i=1}^n M_\tau\left(\frac{x-X_i}{h}\right)\\
& = \frac{1}{h^d}\int M_\tau\left(\frac{x-y}{h}\right) d\P_n(y).
\end{align*}
Lemma~\ref{lem::KDE_BV} implies 
$$
\E\left(\hat{p}_{\tau,h}(x)\right) = \frac{1}{h^d}\int M_\tau\left(\frac{x-y}{h}\right) d\P(y) = p(x) + O(h^{2+\delta_0}).
$$

Using the notation of empirical process and defining $f_x(y) = \frac{1}{\sqrt{h^d}}M_\tau\left(\frac{x-y}{h}\right)$,
we can rewrite the difference
$$
\hat{p}_{\tau,h}(x) -p(x) = \frac{1}{\sqrt{h^d}}\left(\P_n(f_x) - \P(f_x)\right) + O(h^{2+\delta_0}).
$$
Therefore,
\begin{equation}
\sqrt{nh^d}\left(\hat{p}_{\tau,h}(x) -p(x)\right) = \G_n(f_x) + O(\sqrt{nh^{d+4+2\delta_0}}) = \G_n(f_x) + o(1)
\label{eq::emp1}
\end{equation}
when $\frac{nh^{d+4}}{\log n}\rightarrow c_0$ for some $c_0\geq0$ bounded.
Based on the above derivations, we define the function class
$$
\mathcal{F}_{\tau,h} = \left\{f_x(y)=\frac{1}{\sqrt{h^d}}M_\tau\left(\frac{x-y}{h}\right): x\in\K\right\}.
$$
By using the Gaussian approximation method of
\cite{chernozhukov2014anti,chernozhukov2014gaussian},
we derive the asymptotic behavior of $\hat{p}_{\tau,h}$.

\begin{thm}[Gaussian approximation]
Assume (K1-2) and (P). 
Assume $\tau \in(0,\infty)$ is fixed, and $h = n^{-\frac{1}{\varpi}}$ for some $\varpi > 0$ such that $\frac{nh^{d+4}}{\log n}\rightarrow c_0 \geq 0$ for some $c_0$ bounded and $\frac{nh^d}{\log n} \rightarrow \infty$.
Then there exists a Gaussian process $\B_n$ defined on $\mathcal{F}_{\tau,h}$ such that 
for any $f_1,f_2\in \mathcal{F}_{\tau,h}$, $\E(\B_n(f_1)\B_n(f_2)) = {\sf Cov}\left(f_1(X_i), f_2(X_i)\right)$
and
$$
\sup_{t\in\mathbb{R}}\left|\P\left(\sqrt{nh^d}\left\|\hat{p}_{\tau,h} -p\right\|_{\infty}\leq t\right) - \P\left(\sup_{f\in\mathcal{F}_{\tau,h}}\|\B_n(f)\|\leq t\right)\right|
= O\left(\left(\frac{\log^7 n}{nh^d}\right)^{1/8}\right).
$$
\label{thm::KDE_gaussian}
\end{thm}
Theorem~\ref{thm::KDE_gaussian} shows that the $L_\infty$ metric can be approximated by
the distribution of the supremum of a Gaussian process.
The requirement on $h$, $\frac{ nh^{d+4}}{\log n}\rightarrow c_0 \geq 0$ for some $c_0$,
is very useful--it allows the case where $h=O(n^{-\frac{1}{d+4}})$, the optimal choice of smoothing bandwidth
of the usual KDE.
As a result,
we can choose the smoothing bandwidth
using standard receipts such as the 
the rule of thumb and least square cross-validation method \citep{Chacon2011,Silverman1986}.
A similar Gaussian approximation (and later the bootstrap consistency) 
also appeared in \cite{neumann1998simultaneous}.

Finally, we prove that the distribution of the bootstrap $L_\infty$ error $\|\hat{p}_{\tau,h}^*-\hat{p}_{\tau,h}\|_{\infty}$
approximates the distribution of the original $L_\infty$ error,
which leads to the validity of the bootstrap confidence band.
\begin{thm}[Confidence bands of density function]
Assume (K1-2) and (P). 
Assume $\tau \in(0,\infty)$ is fixed, and $h = n^{-\frac{1}{\varpi}}$ for some $\varpi > 0$ such that $\frac{nh^{d+4}}{\log n}\rightarrow c_0 \geq 0$ for some $c_0$ bounded and $\frac{nh^d}{\log n} \rightarrow \infty$.  
Let $\hat{t}_{1-\alpha}$ be the $1-\alpha$ quantile of the distribution of the bootstrapped $L_\infty$ metric; namely,
$$
\hat{t}_{1-\alpha} = \hat{F}^{-1}(1-\alpha),\quad
\hat{F}(t) = P\left(\left\|\hat{p}_{\tau,h}^*-\hat{p}_{\tau,h}\right\|_{\infty}<t|X_1,\cdots,X_n\right).
$$
Then define the $1 - \alpha$ confidence band $\hat{C}_{1 - \alpha}$ as 
$$
\hat{C}_{1 - \alpha}(x) = [\hat{p}_{\tau, h}(x) - \hat{t}_{1 - \alpha}, \hat{p}_{\tau, h}(x) + \hat{t}_{1 - \alpha}]
$$
we have 
$$
P\left( p(x) \in \hat{C}_{1-\alpha}(x)\,\, \forall x\in\K \right) 
=1-\alpha+ O\left(\left(\frac{\log^7 n}{nh^d}\right)^{1/8}\right).
$$
Namely, $\hat{C}_{1-\alpha}(x)$ is an asymptotically valid $1-\alpha$ confidence band of the density function $p$.

\label{thm::KDE_CI}
\end{thm}
Theorem~\ref{thm::KDE_CI} proves that bootstrapping the debiased KDE leads to an asymptotically valid confidence 
band of $p$.
Moreover, we can choose the smoothing bandwidth at rate $h=O(n^{-\frac{1}{d+4}})$,
which is compatible with most bandwidth selectors.
This shows that bootstrapping the debiased KDE
yields a confidence band with width shrinking at rate $O_P(\sqrt{\log n} \cdot n^{-\frac{2}{d+4}})$,
which is not attainable if we undersmooth the usual KDE.

{
Note that our confidence band has a coverage error $O\left(\left(\frac{\log^7 n}{nh^d}\right)^{1/8}\right)$,
which is due to the stochastic variation of the estimator. The bias of the debiased estimator
is of a smaller order so it does not appear in the coverage error. 
When the bias and the stochastic variation are of a similar order, there will be an additional term
from the bias and one may be able to choose the bandwidth by optimizing the coverage error \citep{calonico2015effect}.  
However, deriving the influence of bias is not easy since the limiting distribution does not have
a simple form like a Gaussian. 
}

\begin{remark}
The bootstrap consistency given in Theorem~\ref{thm::KDE_CI} 
shows that our method may be very useful in topological data analysis 
\citep{carlsson2009topology,edelsbrunner2012persistent,wasserman2018topological}.
Many statistical inferences of topological features of a density function
are accomplished by bootstrapping the $L_\infty$ distance \citep{fasy2014confidence,chazal2014stochastic,chen2016generalized,jisu2016statistical}.
However, most current approaches 
consider bootstrapping the original KDE so the inference is for the topological features of the `smoothed' density function
rather than the features of the original density function $p$. 
By bootstrapping the debiased KDE, we can construct confidence 
sets for the topological features of $p$. 
In addition, the assumption (P) in topological data analysis is reasonable
because
many 
topological features are related to the critical points (points where the density gradient is $0$)
and the curvature at these points (eigenvalues of the density Hessian matrix). 
To guarantee consistency when estimating these structures, we 
need to assume more smoothness of the density function, so (P) is a very mild assumption when we want to infer topological features.
\end{remark}

\begin{remark}	\label{rm::honest}
By a similar derivation as \cite{chernozhukov2014anti},
we can prove that $\hat{C}_{1-\alpha}(x)$
is a honest confidence band of the H\"older class $\Sigma(2+\delta_0,L_0)$ for some $\delta_0, L_0>0$.
i.e.,
$$
\inf_{p \in \Sigma(2+\delta_0,L_0)} P\left( p(x) \in \tilde{C}_{1-\alpha}(x)\,\, \forall x\in\K \right) 
=1-\alpha+ O\left(\left(\frac{\log^7 n}{nh^d}\right)^{1/8}\right).
$$

For a H\"older class $\Sigma(2+\delta_0,L_0)$, 
the optimal width of the confidence band will be at rate $O\left(n^{-\frac{1+\frac{\delta_0}{2}}{d+4+2\delta_0}}\right)$ \citep{Tsybakov1997}.
With $h=O(n^{-\frac{1}{d+4}})$, the width of our confidence band 
is at rate $O_P(\sqrt{\log n} \cdot n^{-\frac{2}{d+4}})$,
which is suboptimal when $\delta_0$ is large.
However, when $\delta_0$ is small, the size of our confidence band shrinks almost at the same rate
as the optimal confidence band.
\end{remark}

\begin{remark}
The correction in the bootstrap coverage, $O\left(\left(\frac{\log^7 n}{nh^d}\right)^{1/8}\right)$, 
is not optimal. 
\cite{chernozhukov2017central} introduced an induction method to obtain
a rate of $O(n^{-1/6})$ for bootstrapping high dimensional vectors.
We believe that one can apply a similar technique to obtain a coverage correction at rate 
$O\left(\left(\frac{\log^7 n}{nh^d}\right)^{1/6}\right)$.

\end{remark}

The Gaussian approximation also works for the Hausdorff error of the level set estimator $\hat{D}_{\tau,h}$
\citep{chen2017density}. 
Thus, bootstrapping the Hausdorff metric approximates the distribution of the actual Hausdorff error,
leading to the following result.

\begin{thm}[Confidence set of level sets]
Assume (K1-2), (P), (D), and
$\tau \in(0,\infty)$ is fixed, and $h = n^{-\frac{1}{\varpi}}$ for some $\varpi > 0$ such that $\frac{nh^{d+4}}{\log n}\rightarrow c_0 \geq 0$ for some $c_0$ bounded and $\frac{nh^{d+2}}{\log n} \rightarrow \infty$.
Recall that $C^{LV}_{n,1-\alpha} = \hat{D}_{\tau,h}\oplus \hat{s}_{1-\alpha}$.
Then 
$$
P\left(D\subset C^{LV}_{n,1-\alpha}\right)
=1-\alpha+ O\left(\left(\frac{\log^7 n}{nh^d}\right)^{1/8}\right).
$$
Namely, $C^{LV}_{n,1-\alpha}$ is an asymptotic confidence set of the level set $D=\{x: p(x)=\lambda\}$.
\label{thm::LV_CI}
\end{thm}
The proof of Theorem~\ref{thm::LV_CI} is similar to the proof of Theorem 4 in \cite{chen2017density},
so we ignore it.
The key element in the proof is showing that 
the supremum of an empirical process approximates
the Hausdorff distance,
so we can approximate the Hausdorff distance using the supremum of a Gaussian process. 
Finally we show that the bootstrap Hausdorff distance converges to the same Gaussian process.

Theorem~\ref{thm::LV_CI} proves that the bootstrapping confidence set of the level set
is asymptotically valid.
Thus, bootstrapping the debiased KDE leads to not only a valid confidence band
of the density function but also a valid confidence set of the density level set.
Note that \cite{chen2017density} proposed bootstrapping the original level set estimator $\hat{D}_h = \{x: \hat{p}_h(x)=\lambda\}$,
which leads to a valid confidence set of the smoothed level set $D_h=\{x: \E\left(\hat{p}_h(x)\right)=\lambda\}$. 
However, their confidence set is not valid for inferring $D$
unless we undersmooth the data.

\subsection{Local Polynomial Regression}	\label{sec::LPR}

To analyze the theoretical behavior of the local linear smoother, we consider the following assumptions.\\
{\bf Assumptions.}
\begin{itemize}
\item[(K3)] Let 
\begin{align*}
\mathcal{K}^\dagger_\ell &= \left\{y\mapsto \left(\frac{x-y}{h}\right)^\gamma K\left(\frac{x-y}{h}\right): x\in\D, \gamma=0,\cdots, \ell, h>0\right\},
\end{align*}
We assume that $\mathcal{K}^\dagger_6$ is a VC-type class (see assumption (K2) for the formal definition). 
\item[(R1)] The density of covariate $X$, $p_X$, 
has compact support $\D\subset\R$ and $p_X(x)>0$ for all $x\in\D$. $\sup_{x \in \D}\E(|Y|^4|X=x)\leq C_0<\infty$. Moreover, $p_X$ is continuous and the regression function $r$
is in H\"older Class $\Sigma(2+\delta_0,L_0)$ for some constant $L_0>0$ and $2\geq \delta_0>2/3$.
\item[(R2)] At any point of $\cR$, the gradient of $r$ is nonzero, i.e.,
$$
\inf_{x\in \cR} \|r'(x)\| \geq g_1>0,
$$
for some $g_1$.
\end{itemize}

(K3) is the local polynomial version assumption of (K2), 
which is a mild assumption that any kernel with a compact support and the Gaussian kernel satisfy this assumption.
(R1) contains two parts. The first part
is a common assumption to guarantee the convergence rate of the local polynomial regression \citep{fan1996local,wasserman2006all}. 
The latter part of (R1) is analogous to (P), which is a very mild condition.
(R2) is an analogous assumption to (D) that is needed to derive the convergence rate of the inverse regression. 

\begin{lem}[Bias and variance of the debiased local linear smoother]
Assume (K1), (R1),
and $\tau \in(0,\infty)$ is fixed,  
Then the bias and variance of $\hat{r}_{\tau,h}$ for a given point $x$ is at rate
\begin{align*}
\E\left(\hat{r}_{\tau,h}(x)\right) - r(x) &= O(h^{2+\delta_0}) + O\left(\sqrt{\frac{h^3}{n}}\right)\\
{\sf Var}(\hat{r}_{\tau,h}(x)) & = O\left(\frac{1}{nh} \right).
\end{align*}
with $h = O(n^{-1/5})$, the rate for bias would be 
$$
\E\left(\hat{r}_{\tau,h}(x)\right) - r(x) = O(h^{2+\delta_0})\\
$$
\label{lem::LS_BV}
\end{lem}

Define
$\Omega_k\in\R^{(k+1)\times (k+1)}$
whose elements $\Omega_{k,ij} = \int u^{i+j-2}K(u)du$.
and define $e_1^T = (1,0)$ and $e_3^T = (0,0,1,0)$. 
Let $\psi_x:\R^2\mapsto \R$ be a function 
defined as 
\begin{equation}
\begin{aligned}
\psi_x(z) =\frac{1}{p_X(x)h}\left( e_1^T\Omega_1^{-1} \Psi_{0,x}(z) - 
c_K\cdot\tau^3\cdot e_3^T \Omega_3^{-1}\Psi_{2,\tau x}(\tau z_1, z_2)\right),
\end{aligned}
\label{eq::LS::function}
\end{equation}
and
\begin{equation}
\begin{aligned}
\Psi_{0,x}(z_1,z_2)^T &= (\eta_0(x,z_1,z_2), \eta_1(x,z_1,z_2))\\
\Psi_{2,x}(z_1,z_2)^T &= (\eta_0(x,z_1,z_2), \eta_1(x,z_1,z_2),\eta_2(x,z_1,z_2),\eta_3(x,z_1,z_2))\\
\eta_j(x,z_1,z_2) &=  z_2\cdot\left(\frac{z_1-x}{h}\right)^j\cdot K\left(\frac{z_1-x}{h}\right).
\end{aligned}
\label{eq::z1z2}
\end{equation}

\begin{lem}[Empirical approximation]
Assume (K1,3), (R1),
and $\tau \in(0,\infty)$ is fixed, and $h = n^{-\frac{1}{\varpi}}$ for some $\varpi > 0$ such that $\frac{nh^{5}}{\log n}\rightarrow c_0 \geq 0$ for some $c_0$ bounded and $\frac{nh}{\log n} \rightarrow \infty$. 
Then the scaled difference $\sqrt{nh}(\hat{r}_{\tau,h}(x)-E(\hat{r}_{\tau,h}(x)))$ has the following approximation:
\begin{align*}
\sup_{x\in\D}\left\|\frac{\sqrt{nh}\left(\hat{r}_{\tau,h}(x)- E(\hat{r}_{\tau, h}(x))\right) - \sqrt{h}\G_n(\psi_x)}{\frac{1}{\sqrt{h}} \G_n(e_1^T \Psi_{x,0} - c_K \cdot \tau^3 e_3^T \Psi_{2,\tau x})}\right\| = O(h)+ O_P\left(\sqrt{\frac{\log n}{nh}}\right),
\end{align*}
where $\psi_x(z) $ is defined in equation \eqref{eq::LS::function}.
Moreover, 
the debiased local linear smoother $\hat{r}_{\tau,h}(x)$ has the following error rate
$$
\|\hat{r}_{\tau,h} - r\|_{\infty} = O(h^{2+\delta_0}) + O\left(\sqrt{\frac{h^3}{n}}\right) + O_P\left( \sqrt{\frac{\log n}{nh}}\right).
$$
with $h = O(n^{-1/5})$
$$
\|\hat{r}_{\tau,h} - r\|_{\infty} = O(h^{2+\delta_0})  + O_P\left( \sqrt{\frac{\log n}{nh}}\right).
$$
\label{lem::LS_empirical}
\end{lem}
Lemma~\ref{lem::LS_empirical} shows that we can approximate the
$\sqrt{nh}\left(\hat{r}_{\tau,h}(x)-\E(\hat{r}_{\tau, h}(x))\right)$
by an empirical process $\sqrt{h}\G_n(\psi_x)$. 
Based on this approximation, the second assertion (uniform bound) is an immediate result
from the empirical process theory in \cite{Einmahl2005}.
Lemma~\ref{lem::LS_empirical} is another form of the uniform Bahadur representation \citep{bahadur1966note,kong2010uniform}.

Note that Lemma~\ref{lem::LS_empirical} also works for the usual local linear smoother
or other local polynomial regression estimators (but centered at their expectations). 
Namely, the local polynomial regression
can be uniformly approximated by an empirical process. 
This implies that we can apply empirical process theory
to analyze the asymptotic behavior of the local polynomial regression.


\begin{remark}
\cite{fan1996local} have discussed the prototype of the empirical approximation.
However, they only derived a pointwise approximation rather than a uniform approximation. 
To construct a confidence band that is uniform for all $x\in\D$, we need a uniform approximation of the local linear smoother
by an empirical process.
\end{remark}

Now we define the function class
$$
\mathcal{G}_{\tau,h} = \left\{\sqrt{h}\psi_x(z): x\in\D\right\},
$$
where $\psi_x(z)$ is defined in equation \eqref{eq::LS::function}. 
The set $\cG_{\tau,h}$ is analogous to the set $\cF_{\tau,h}$ in the KDE case.
With this notation and using Lemma~\ref{lem::LS_empirical}, we conclude
$$
\sup_{x\in\D}\|\sqrt{nh}\left(\hat{r}_{\tau,h}(x)-r(x)\right)\| \approx \sup_{f\in \cG_{\tau,h}}\|\G_n(f)\|.
$$
Under assumption (K1, K3)
and applying the Gaussian approximation method of \cite{chernozhukov2014anti,chernozhukov2014gaussian},
the distribution of the right-hand-side will be approximated 
by the distribution of the maxima of a Gaussian process, which leads to the following conclusion.

\begin{thm}[Gaussian approximation of the debiased local linear smoother]
Assume (K1,3), (R1),
$\tau \in(0,\infty)$ is fixed, and $h = n^{-\frac{1}{\varpi}}$ for some $\varpi > 0$ such that $\frac{nh^{5}}{\log n}\rightarrow c_0 \geq 0$ for some $c_0$ bounded and $\frac{nh}{\log n} \rightarrow \infty$. 
Then there exists a Gaussian process $\B_n$ defined on $\mathcal{G}_{\tau,h}$ such that 
for any $f_1,f_2\in \mathcal{G}_{\tau,h}$, $\E(\B_n(f_1)\B_n(f_2)) = {\sf Cov}\left(f_1(X_i,Y_i), f_2(X_i,Y_i)\right)$
and
$$
\sup_{t\in\mathbb{R}}\left|\P\left(\sqrt{nh^d}\left\|\hat{r}_{\tau,h} -r\right\|_{\infty}\leq t\right) - \P\left(\sup_{f\in\mathcal{G}_{\tau,h}}\|\B_n(f)\|\leq t\right)\right|
= O\left(\left(\frac{\log^7 n}{nh}\right)^{1/8}\right).
$$
\label{thm::LS_gaussian}
\end{thm}
The proof of Theorem \ref{thm::LS_gaussian} follows a similar way as the proof of Theorem~\ref{thm::KDE_gaussian}
so we omit it. 

Theorem~\ref{thm::LS_gaussian} shows that the $L_\infty$ error
of the debiased linear smoother will be approximated by the maximum of a Gaussian process.
Thus, as long as we can prove that the bootstrapped $L_\infty$ error converges
to the same Gaussian process, we have bootstrap consistency of the confidence band.

\begin{thm}[Confidence band of regression function]
Assume (K1,3), (R1),
$\tau \in(0,\infty)$ is fixed, and $h = n^{-\frac{1}{\varpi}}$ for some $\varpi > 0$ such that $\frac{nh^{5}}{\log n}\rightarrow c_0 \geq 0$ for some $c_0$ bounded and $\frac{nh}{\log n} \rightarrow \infty$. 
Let 
$\hat{s}_{1-\alpha}$ be the $(1-\alpha)$ quantile of 
the distribution
$$
\hat{G}(s) = P\left(\|\hat{r}^*_{\tau,h}-\hat{r}_{\tau,h}\|_{\infty}< s|X_1,\cdots,X_n\right). 
$$
Then define the confidence band as following:
$$
\hat{C}^R_{1-\alpha}(x) = \left[\hat{r}_{\tau, h}(x) - \hat{s}_{1-\alpha}, \,\,\hat{r}_{\tau, h}(x) + \hat{s}_{1-\alpha}\right].
$$
We would have 
$$
P\left( r(x) \in \hat{C}^R_{1-\alpha}(x)\,\, \forall x\in\D \right)
= 1-\alpha+O\left(\left(\frac{\log^7 n}{nh}\right)^{1/8}\right).
$$
Namely, $\hat{C}^R_{1-\alpha}(x)$ is an asymptotically valid $1-\alpha$ confidence band of the regression function $r$.
\label{thm::LS_CI}
\end{thm} 
The proof of Theorem \ref{thm::LS_CI} follows a similar way as the proof of Theorem~\ref{thm::KDE_CI},
with Theorem \ref{thm::KDE_gaussian} being replaced by Theorem \ref{thm::LS_gaussian}.
Thus, we omit the proof.

Theorem~\ref{thm::LS_CI} proves that the confidence band from bootstrapping the debiased local linear smoother
is asymptotically valid. 
This is a very powerful result because 
Theorem~\ref{thm::LS_CI} is compatible with the smoothing bandwidth selected 
by the cross-validation of the original local linear smoother.
This implies the validity of the proposed procedure in Section~\ref{sec::LS}.

Finally, we prove that the confidence set of the inverse regression $\cR$ is also asymptotically valid.
\begin{thm}[Confidence set of inverse regression]
Assume (K1,3), (R1--2), and
$\tau \in(0,\infty)$ is fixed, and $h = n^{-\frac{1}{\varpi}}$ for some $\varpi > 0$ such that $\frac{nh^{5}}{\log n}\rightarrow c_0 \geq 0$ for some $c_0$ bounded and $\frac{nh}{\log n} \rightarrow \infty$. 
Then 
$$
P\left(\cR\subset \hat{\cR}^*_{\tau, h}\oplus \hat{s}^{R}_{1-\alpha}\right)
=1-\alpha+ O\left(\left(\frac{\log^7 n}{nh}\right)^{1/8}\right).
$$
Namely, $\hat{\cR}^*_{\tau, h}\oplus \hat{s}^{R}_{1-\alpha}$ is an asymptotically valid confidence set of the inverse regression $\cR$.
\label{thm::IR_CI}
\end{thm}
The proof of Theorem~\ref{thm::IR_CI} is basically the same as the proof of Theorem~\ref{thm::LV_CI}.
Essentially, the inverse regression is just the level set of the regression function.
Thus, as long as we have a confidence band of the regression function, we have the confidence set of the inverse regression.

A good news is that Theorem~\ref{thm::IR_CI} is compatible with
the bandwidth from the cross-validation $h_{CV}$.
Therefore, we can simply choose $h=h_{CV}$ and then construct the confidence set
by bootstrapping the inverse regression estimator.

\begin{remark}
Note that one can revise the bound on coverage correction in Theorem~\ref{thm::IR_CI}
into the rate $ O\left(\left(\frac{1}{nh}\right)^{1/6}\right)$ by using the following facts. 
First,
the original Hausdorff error is approximately the maximum of 
absolute values of a few normal random variables. 
This is because each estimated location of the inverse regression follows 
an asymptotically normal distribution centered at one population location of the inverse regression.
Then because the bootstrap will approximate this distribution, 
by the Gaussian comparison theorem 
(see, e.g., Theorem 2 in \citealt{chernozhukov2014comparison} and Lemma 3.1 in 
\citealt{chernozhukov2013gaussian}), 
the approximation error rate is $ O\left(\left(\frac{1}{nh}\right)^{1/6}\right)$.

\end{remark}

\begin{remark}
Note that the above results are assuming the $h= h_n\rightarrow0$ in a deterministic way. 
If $h$ is chosen from some conventional data-driven methods, our results still hold. 
Here we give a high-level sketch proof for the KDE case with a smoothing bandwidth chosen
by the (least square) cross-validation approach \citep{sheather2004density}, one can generalize it
to the local polynomial regression problem. 
For the cross-validation method \citep{sheather2004density}, 
denote $u_0$ and $l_0$ as two fixed positive constants,  such that $h_0 \in [l_0, u_0]n^{-\frac{1}{d+4}}$, where $h_0$ is the optimal bandwidth with respect to MISE. 
By theorem \ref{thm::KDE_CI},
\begin{align*}
\sup_{t}\Bigg|P\Bigg(\sqrt{nh^d} \|\hat{p}_{\tau,h}-p\|_\infty <t\Bigg)- & P\Bigg(\sqrt{nh^d} \|\hat{p}^*_{\tau,h}-\hat{p}_{\tau,h}\|_\infty <t\Bigg|\mathcal{X}_n\Bigg)\Bigg|\\
& \leq  O_P\left( \left(\frac{\log^7 n}{nh^d}\right)^{1/8}\right),
\end{align*}
which
leads to a uniform upper bound
\begin{align*}
\sup_{h \in [l_0, u_0] n^{-\frac{1}{d+4}}} \sup_{t}\Bigg|P\Bigg(\sqrt{nh^d} \|\hat{p}_{\tau,h}-p\|_\infty <t\Bigg)- & P\Bigg(\sqrt{nh^d} \|\hat{p}^*_{\tau,h}-\hat{p}_{\tau,h}\|_\infty <t\Bigg|\mathcal{X}_n\Bigg) \Bigg| \\
& \leq  O_P\left( \left(\frac{\log^7 n}{nh_0^d}\right)^{1/8}\right).
\end{align*}
Let $\hat{h}_{CV}$ be the bandwidth chosen by the cross-validation approach.
\cite{duong2005cross} and \cite{sain1994cross} have shown that
$\frac{\hat{h}^{CV} - h_0}{h_0} = O_P(n^{-\frac{\min(d,4)}{2d + 8}})$, so 
$$
P(\hat{h}_{CV} \in [l_0, u_0] n^{-\frac{1}{d+4}}) = 1 - O(n^{-\frac{\min(d,4)}{2d + 8}}).
$$
Combining these two observations together, we obtain
\begin{align*}
& \sup_{t}\Bigg|P\Bigg(\sqrt{n\hat{h}_{CV}^d} \|\hat{p}_{\tau,\hat{h}_{CV}}-p\|_\infty <t\Bigg)-P\Bigg(\sqrt{n\hat{h}_{CV}^d} \|\hat{p}^*_{\tau,\hat{h}_{CV}}-\hat{p}_{\tau,\hat{h}_{CV}}\|_\infty <t\Bigg|\mathcal{X}_n\Bigg)\Bigg| \\
& \leq \sup_{h \in [l_0, u_0] n^{-\frac{1}{d+4}}} \sup_{t}\Bigg|P\Bigg(\sqrt{nh^d} \|\hat{p}_{\tau,\hat{h}_{CV}}-p\|_\infty <t\Bigg)- P\Bigg(\sqrt{nh^d} \|\hat{p}^*_{\tau,\hat{h}_{CV}}-\hat{p}_{\tau,\hat{h}_{CV}}\|_\infty <t\Bigg|\mathcal{X}_n\Bigg) \Bigg| 
\end{align*}
with a probability of $1 - O(n^{-\frac{\min(d,4)}{2d + 8}})$.  This means that 
\begin{align*}
\sup_{t}\Bigg|P\Bigg(\sqrt{n\hat{h}_{CV}^d} \|\hat{p}_{\tau,\hat{h}_{CV}}-p\|_\infty <t\Bigg)- & P\Bigg(\sqrt{n\hat{h}_{CV}^d} \|\hat{p}^*_{\tau,\hat{h}_{CV}}-\hat{p}_{\tau,\hat{h}_{CV}}\|_\infty <t\Bigg|\mathcal{X}_n\Bigg)\Bigg| \\
& \leq O_P\left( \left(\frac{\log^7 n}{nh_0^d}\right)^{1/8}\right).
\end{align*}

Thus, the confidence band proposed in \ref{sec::CI::KDE} is indeed valid. 
For the case of local linear regression, \cite{li2004cross} has already established a similar rate
when the smoothing bandwidth is chosen by the cross-validation approach.
As a result, the same analysis also applies to the local linear regression. 
\end{remark}

\begin{remark}
{\cite{calonico2018coverage} studied 
the problem of optimal coverage error for a confidence interval
and applied their result to a pointwise confidence interval from the debiased local polynomial regression estimator.
In our case, the coverage error is the quantity 
$$
\sup_{P_{XY}\in \mathcal{P}_{XY}}|P\left( r(x) \in \hat{C}^R_{1-\alpha}(x)\,\, \forall x\in\D \right)
- 1-\alpha|,
$$
where $P_{XY}$ is a joint distribution function of of $X$ and $Y$
and $\mathcal{P}_{XY}$ is a class of joint distribution functions.
Theorem~\ref{thm::LS_CI} shows that this quantity will be of the rate $O\left(\left(\frac{\log^7 n}{nh}\right)^{1/8}\right)$
for the class of functions $\mathcal{P}_{XY}$ satisfying our conditions. 
However, this rate is probably suboptimal
when comparing to the rate described in Lemma 3.1 of \cite{calonico2018coverage}.
It is of great interest to study
the optimal rate of a simultaneous confidence band 
and design a procedure that can achieve this rate.
}
\end{remark}

\section{Data Analysis}	\label{sec::simulation}

\subsection{Simulation: Density Estimation}		\label{sec::sim::KDE}

In this section, we demonstrate the coverage of proposed confidence bands/sets
of the density function and level set.

\begin{figure}[h]
\centering
\includegraphics[height=1.5in]{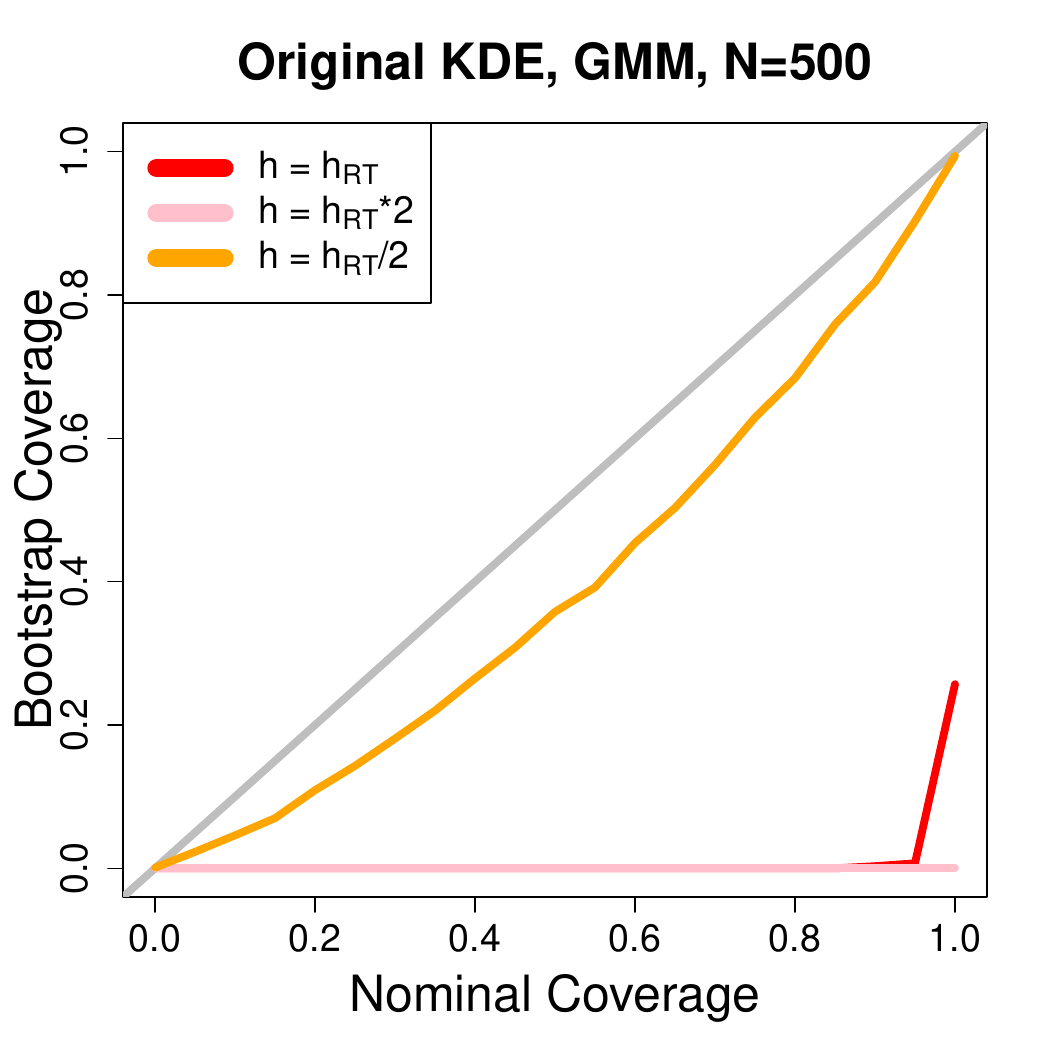}
\includegraphics[height=1.5in]{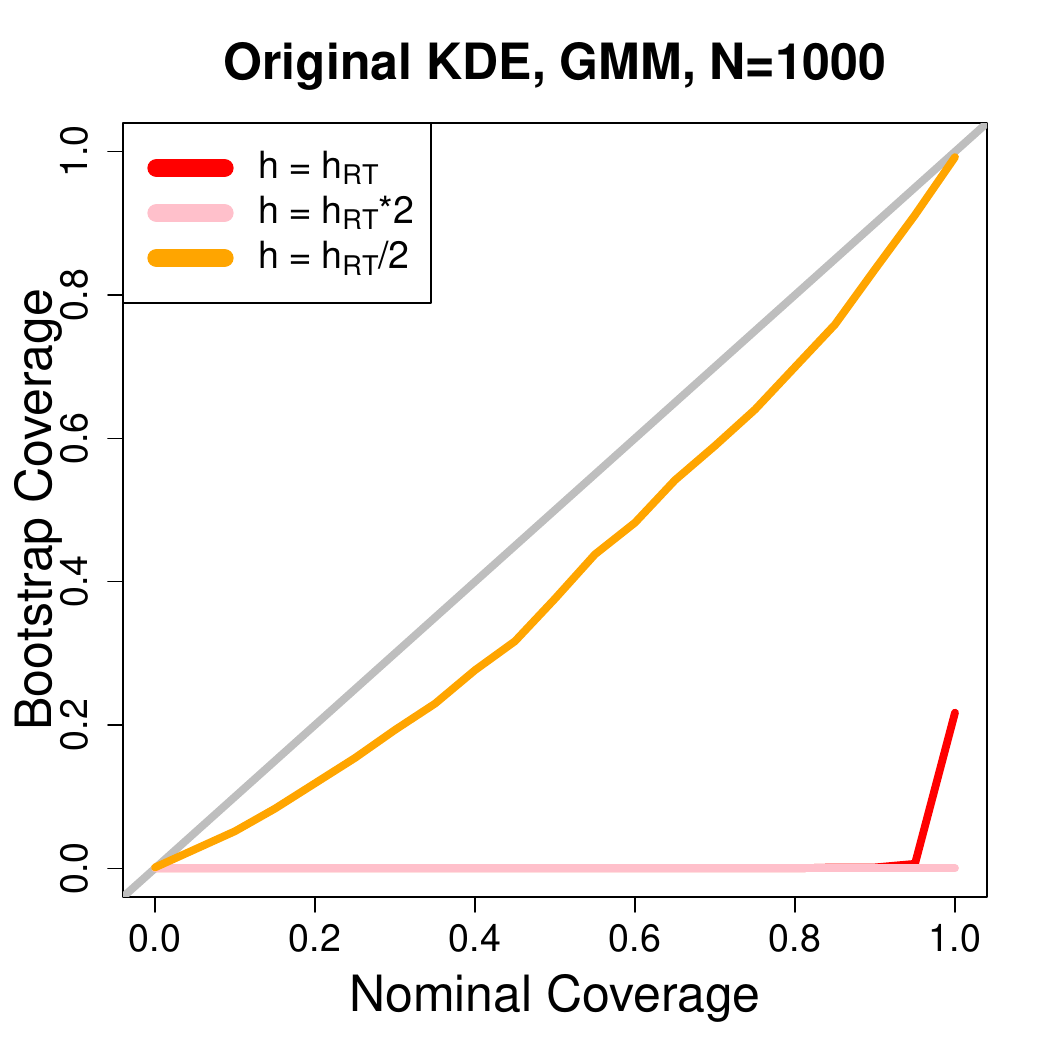}
\includegraphics[height=1.5in]{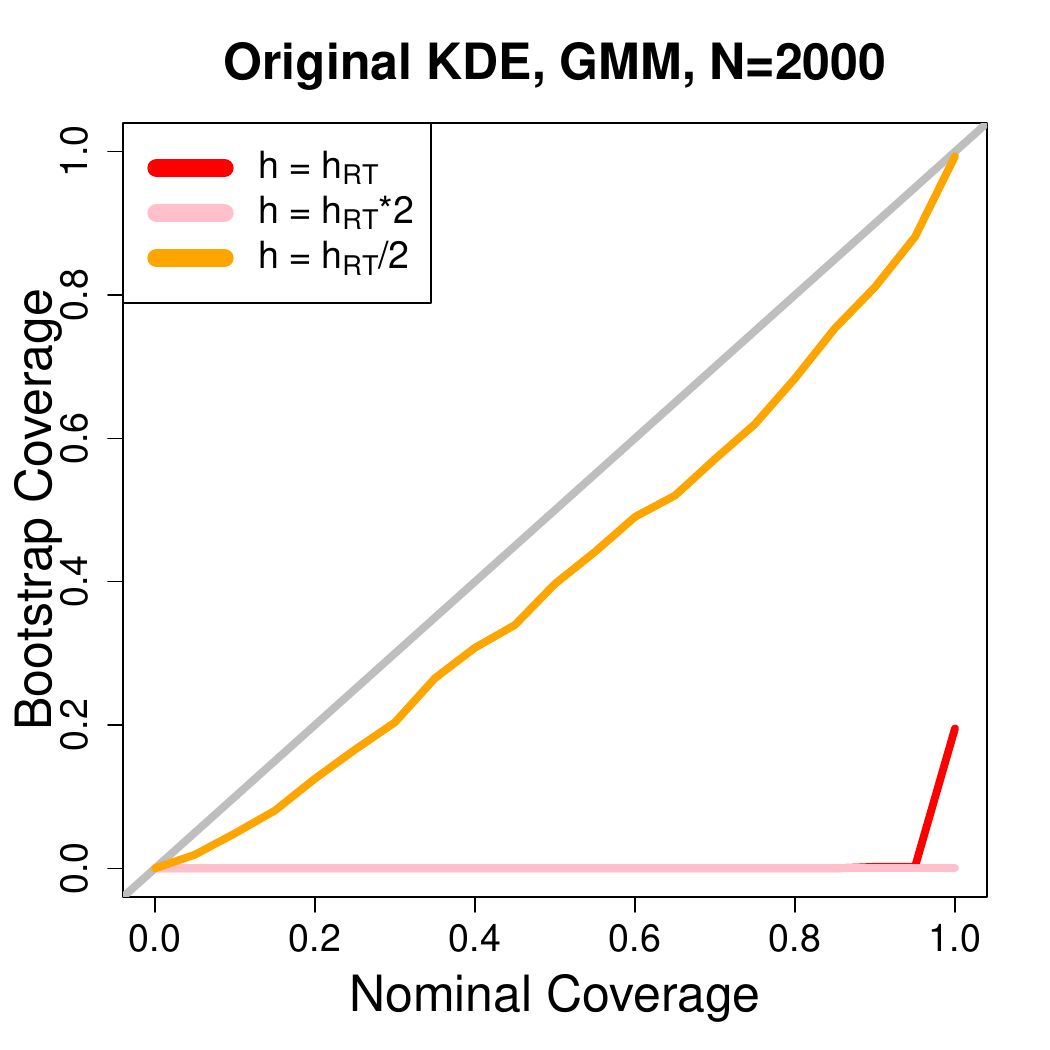}\\
\includegraphics[height=1.5in]{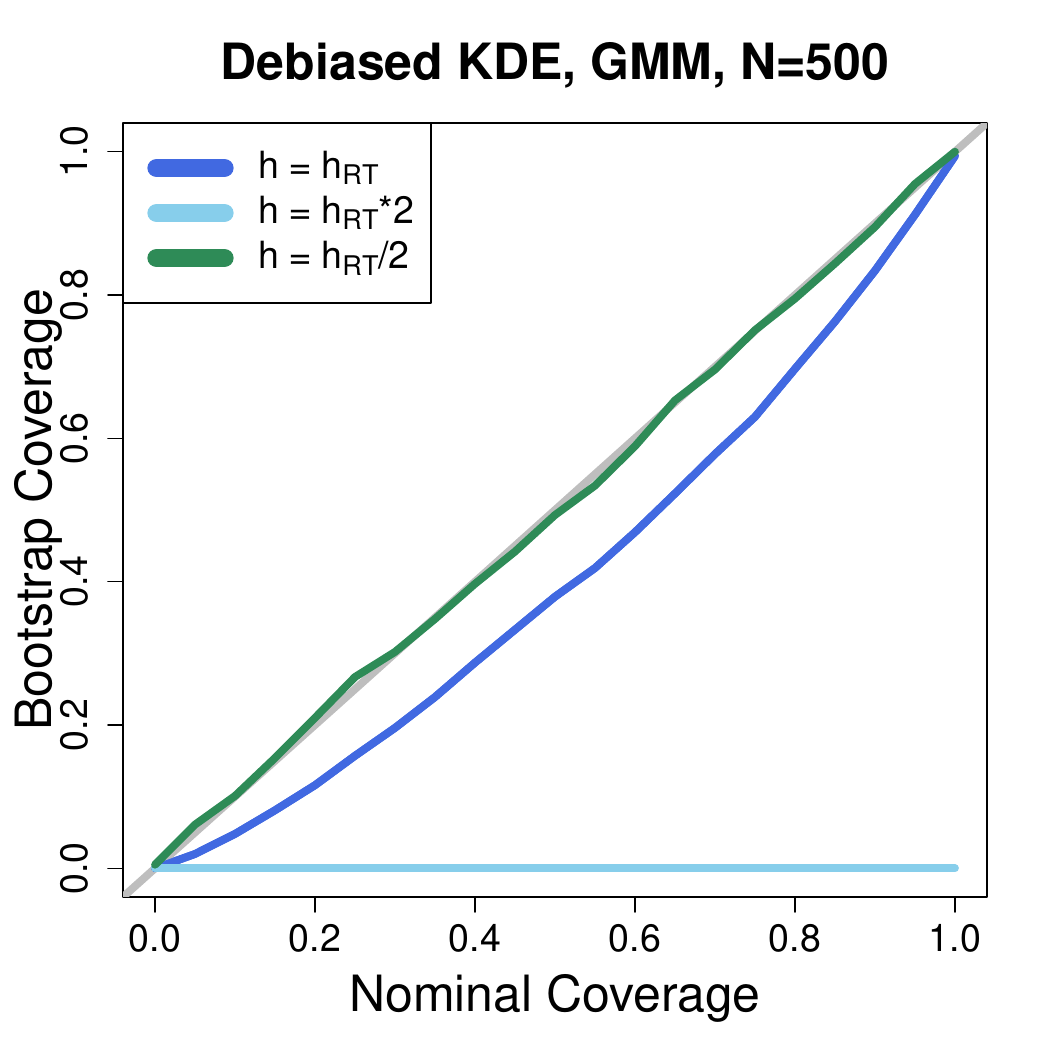}
\includegraphics[height=1.5in]{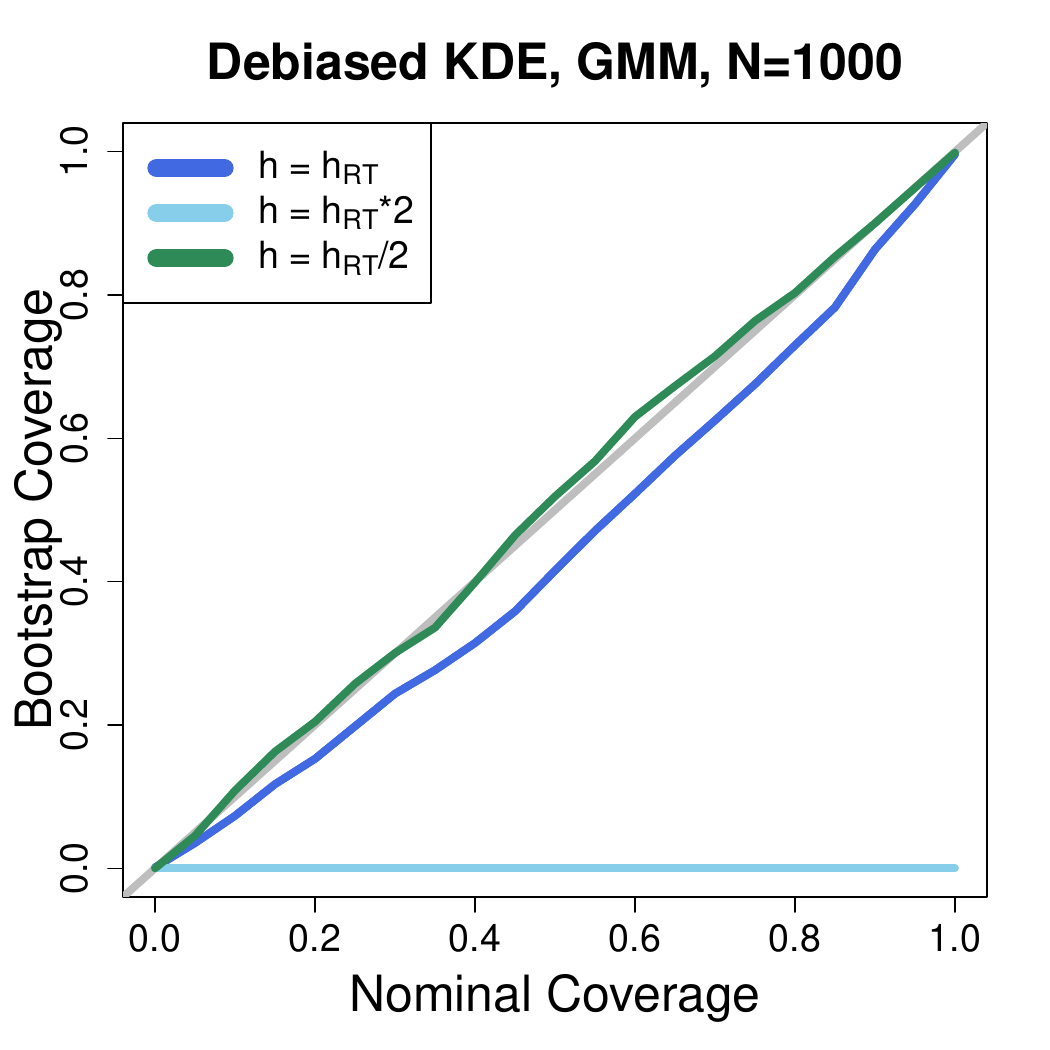}
\includegraphics[height=1.5in]{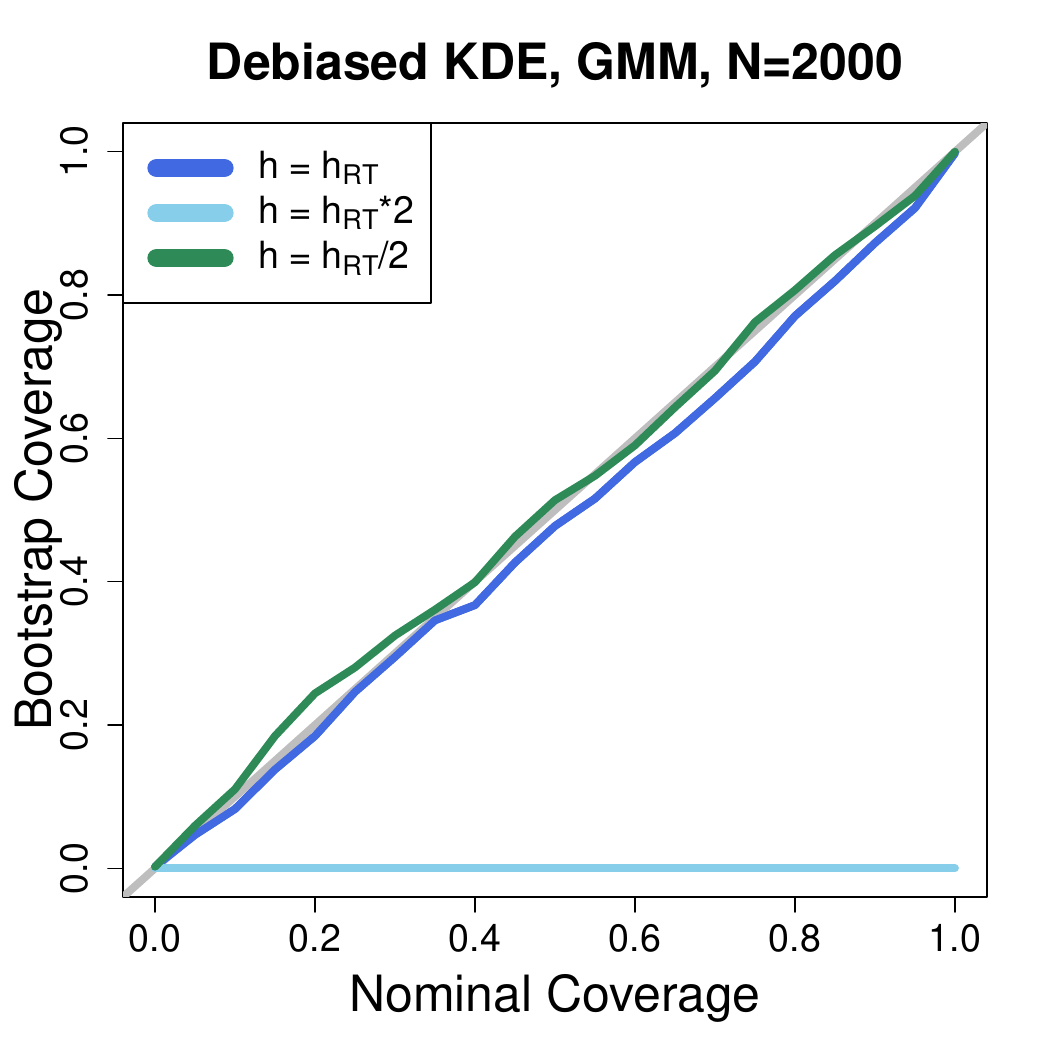}
\caption{Confidence bands of density estimation. 
The top row displays bootstrap coverage versus nominal coverage when we bootstrap the original KDE.
The bottom row shows coverage comparison via bootstrapping the debiased KDE.
It is clear that when we bootstrap the original KDE, the confidence band has undercoverage in every case.
On the other hand, when we bootstrap the debiased KDE, the confidence band achieves nominal coverage
when we undersmooth the data (green curves) or when the sample size is large enough (blue curve).}
\label{fig::de}
\end{figure}

{\bf Density functions.}
To demonstrate the validity of confidence bands for density estimation, we
consider the following Gaussian mixture model. 
We generate $n$ IID data points $X_1,\cdots, X_n$ from a Gaussian mixture such that,
with a probability of $0.6$, $X_i$ is from $N(0,1)$, a standard normal,
and with a probability of $0.4$, $X_i$ is from $N(4,1)$, a standard normal centered at $4$.
The population density of $X_i$ is shown in the black curve in the top left panel of Figure~\ref{fig::ex01}.
We consider three different sample sizes: $n=500, 1000,$ and $2000$.
We bootstrap both the original KDE and the debiased KDE for $1000$ times with three different bandwidths:
$h_{RT}$, $h_{RT}\times 2$, and $h_{RT}/2$, where $h_{RT}$ is the bandwidth from the rule of thumb \citep{Silverman1986}. 
We use these three different bandwidths to show the robustness of the bootstrapped confidence bands against
bandwidth selection.
The result is given in Figure~\ref{fig::de}.
In the top row (the case of bootstrapping the original KDE), 
except for the undersmoothing case (orange line),
confidence band coverage is far below the nominal level. 
And even in the undersmoothing case, the coverage does not achieve the nominal level.
In the bottom row, we display the result of bootstrapping the debiased KDE. 
We see that undersmoothing (green curve)
always yields a confidence band with nominal coverage.
The rule of thumb (blue curve) yields an asymptotically valid confidence band--the bootstrap coverage
achieves nominal coverage when the sample size is large enough (in this case, we need a sample size about $2000$). 
This affirms Theorem~\ref{thm::KDE_CI}.
For the case of oversmoothing, it still fails to generate a valid confidence band.

To further investigate the confidence bands from bootstrapping the debiased estimator,
we consider their width in Figure~\ref{fig::kde_band_width}. 
In each panel, we compare the width of confidence bands generated by 
bootstrapping the debiased estimator (blue) and bootstrapping the original estimator
with undersmoothing bandwidth (red; undersmoothing refers to half of the selected bandwidth by a bandwidth selector). 
In the top row, we consider the case where the smoothing bandwidth is selected by the rule of thumb.
In the bottom row, we choose the smoothing bandwidth by the cross-validation method. 
The result is based on the median width of confidence band from $1000$ simulations.
In every panel, we see that bootstrapping the debiased estimator
leads to a confidence band with a narrower width.
This suggests that bootstrapping the debiased estimator
not only guarantees the coverage but also yield a confidence band that is narrower. 
{A more comprehensive simulation study is provided in Appendix \ref{sec::DE::sim::app}.}



\begin{figure}[h!]
\centering
\includegraphics[height=1.5in]{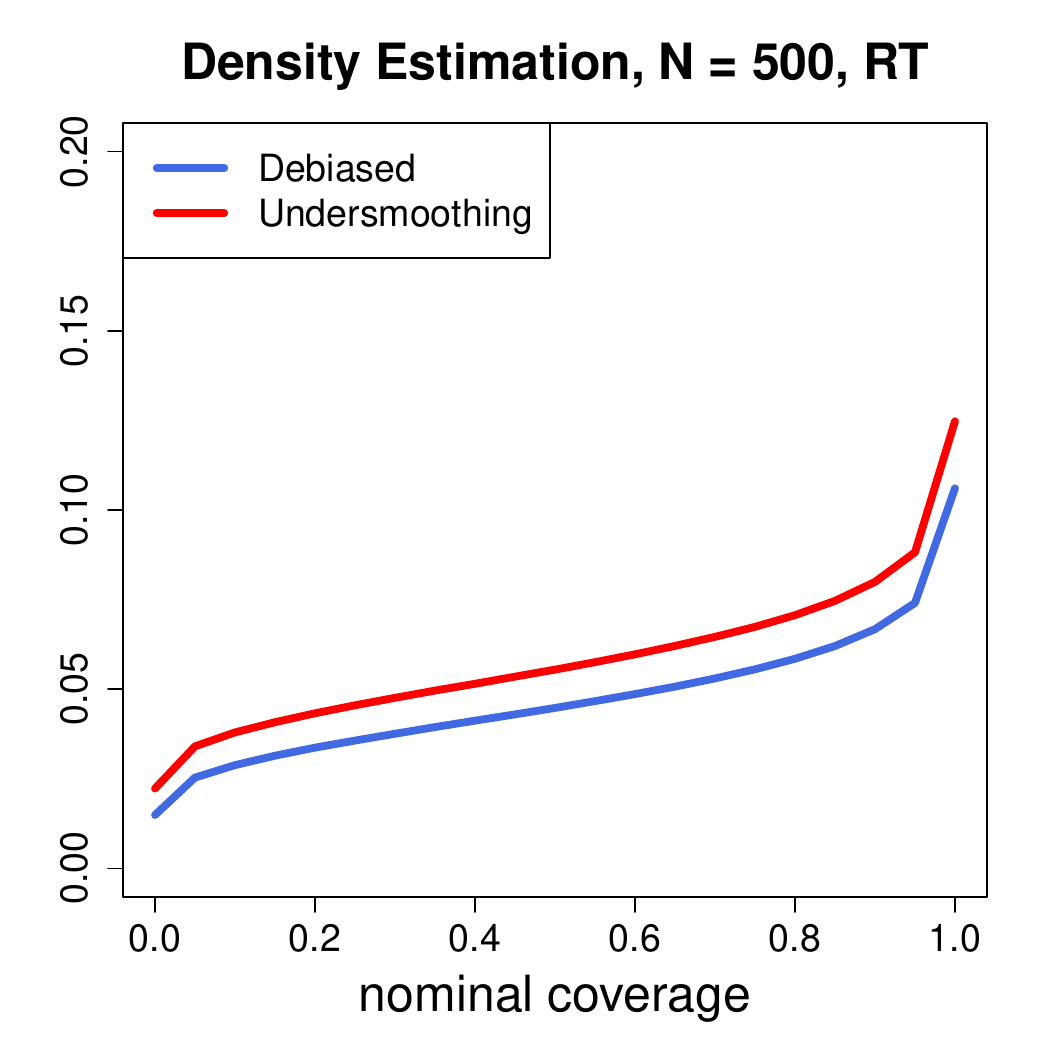}
\includegraphics[height=1.5in]{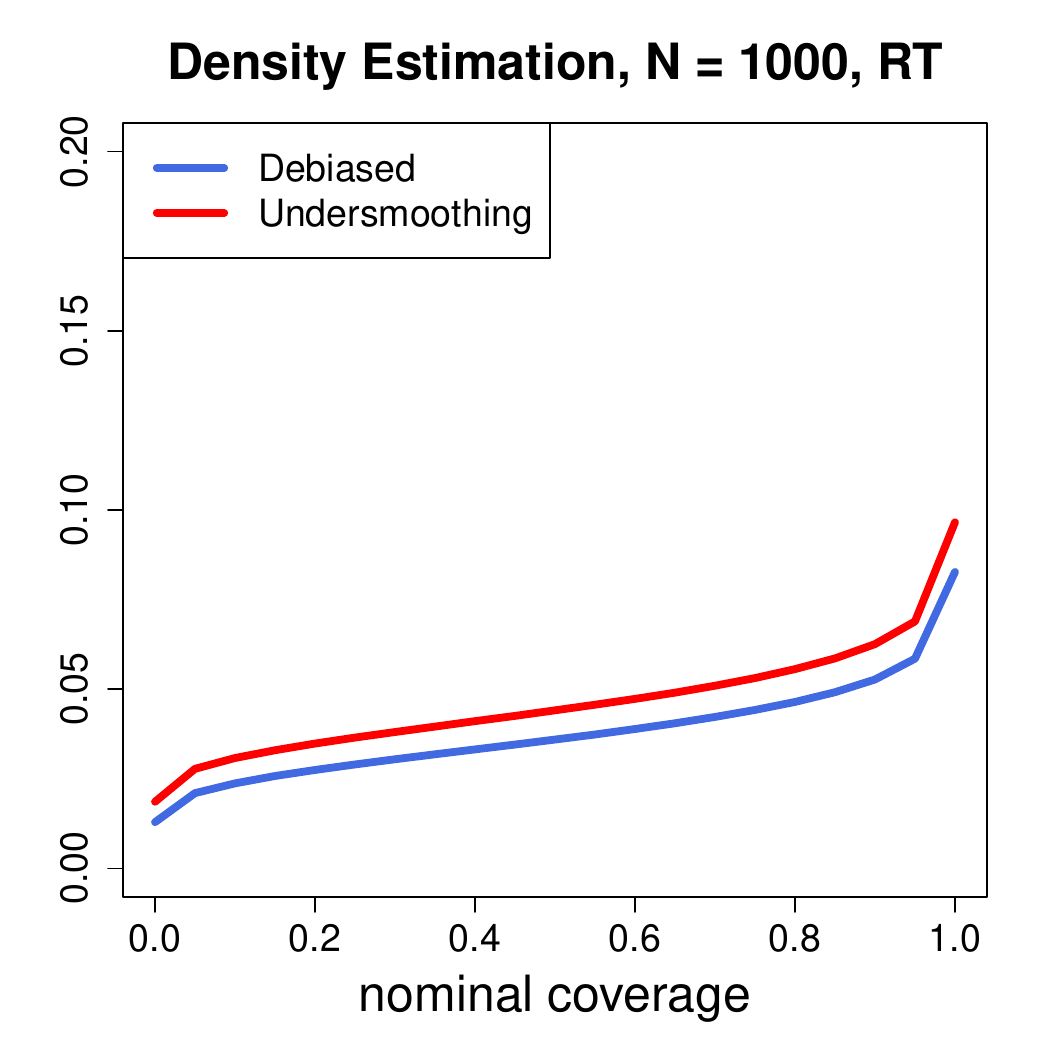}
\includegraphics[height=1.5in]{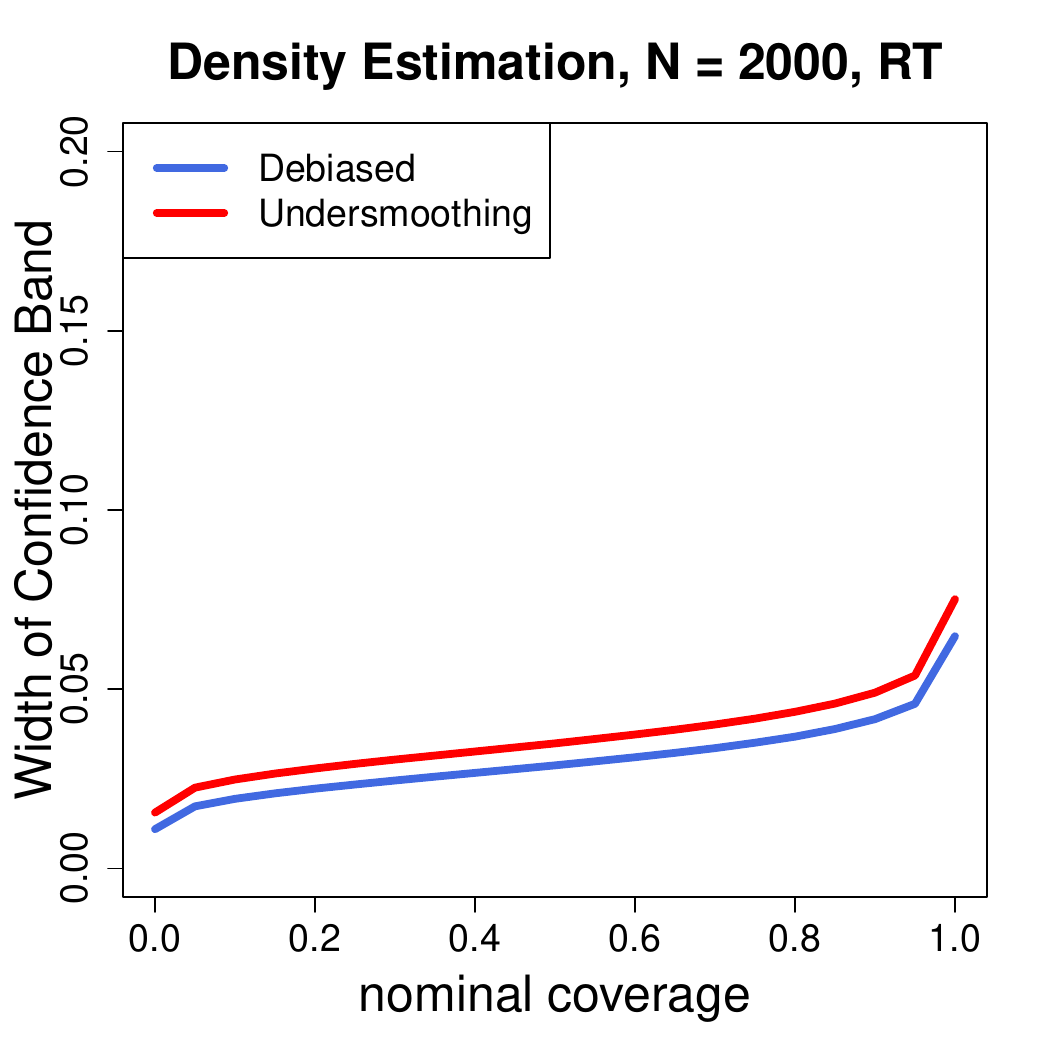}\\
\includegraphics[height=1.5in]{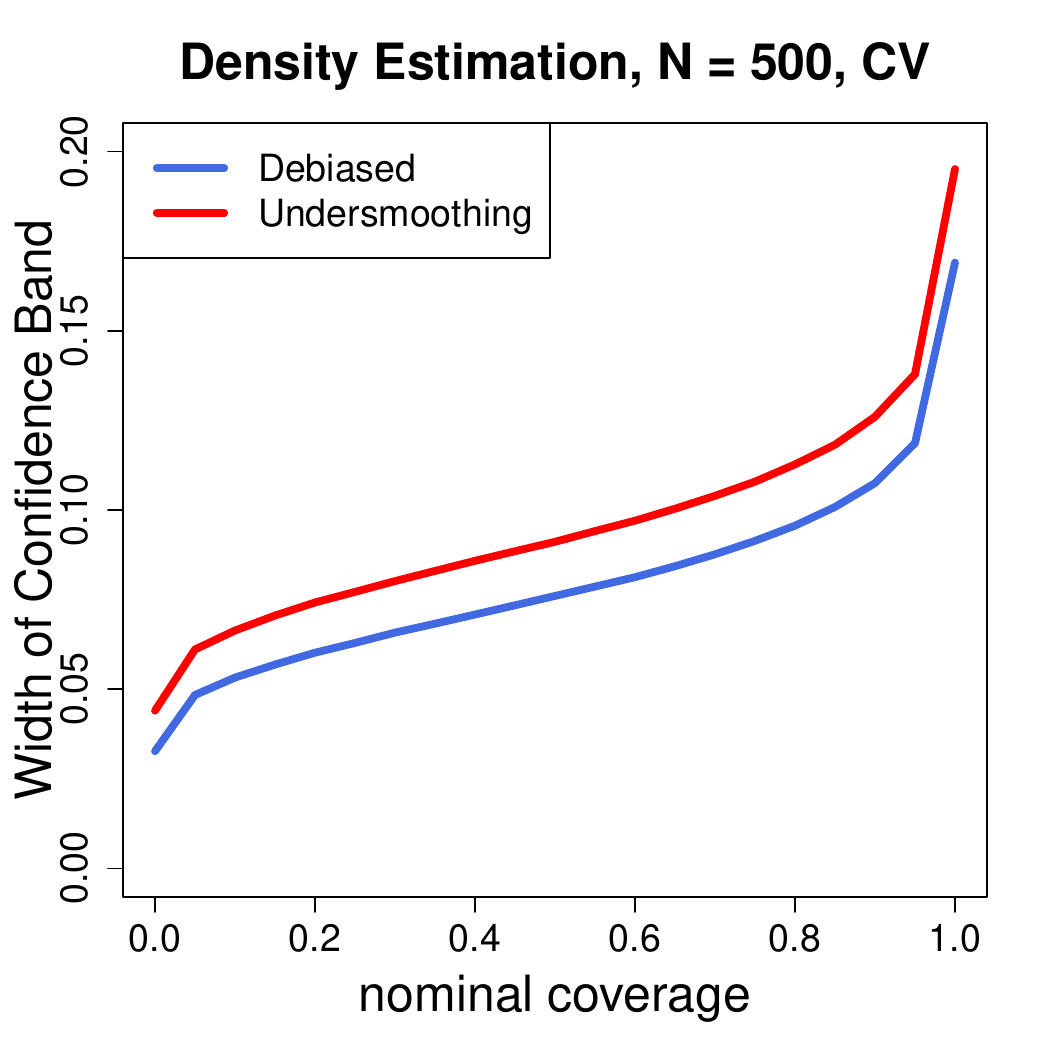} 
\includegraphics[height=1.5in]{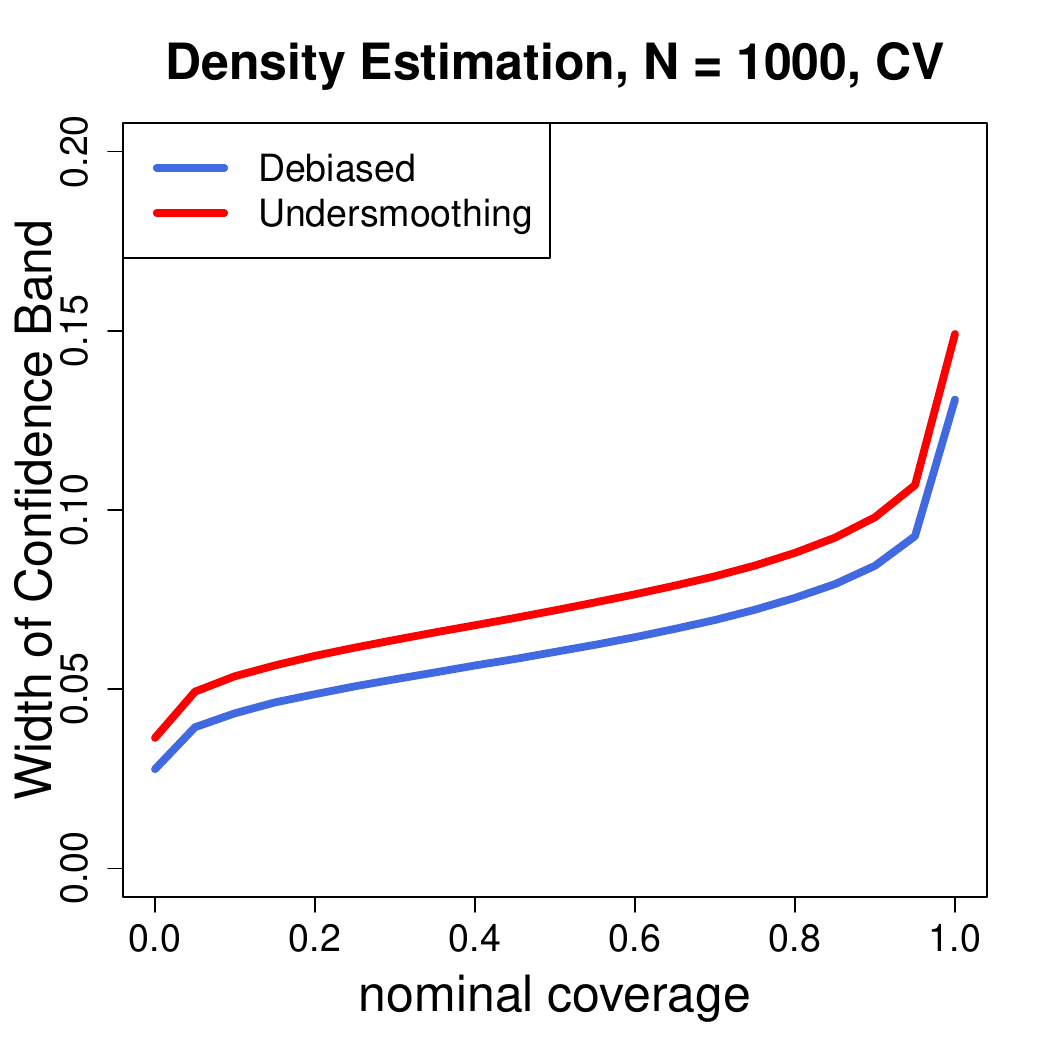} 
\includegraphics[height=1.5in]{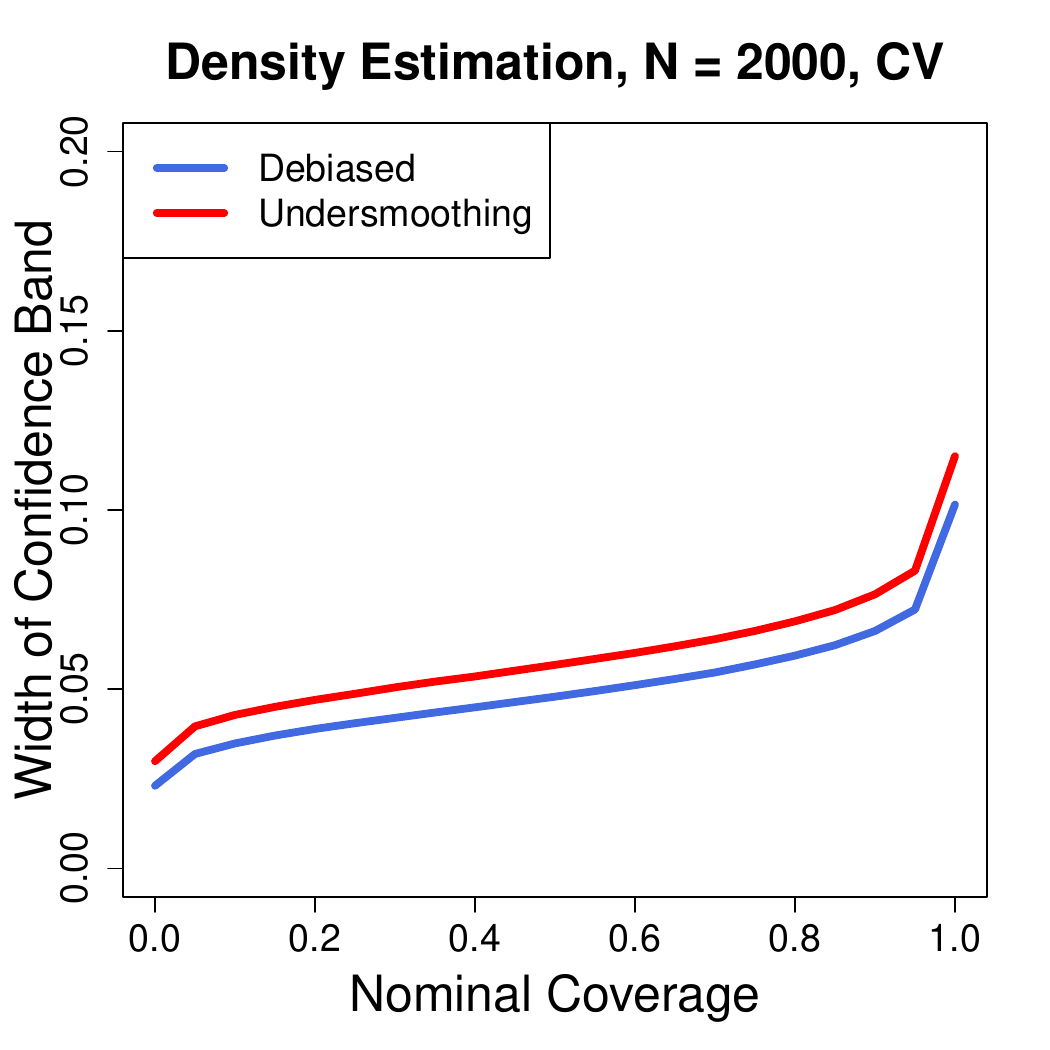} 
\caption{Comparison of width of confidence bands. 
The top row corresponds to the case where bandwidth $h$ is chosen by the Silverman's rule of thumb \citep{Silverman1986}. 
The bottom row corresponds to the case where bandwidth $h$ is chosen by 
the least squared cross validation method \citep{sheather2004density}.  
We compare the width of confidence bands using a debiased estimator and an undersmoothing bandwidth that has a smoothing bandwidth
being half of the chosen bandwidth.
The width of confidence band is computed using the median value over 1000 simulations.  
In every case, the width of confidence bands from the debiased method is always narrower than the ones from the undersmoothing method.}
\label{fig::kde_band_width}
\end{figure}

\begin{figure}[h]
\centering
\includegraphics[height=1.5in]{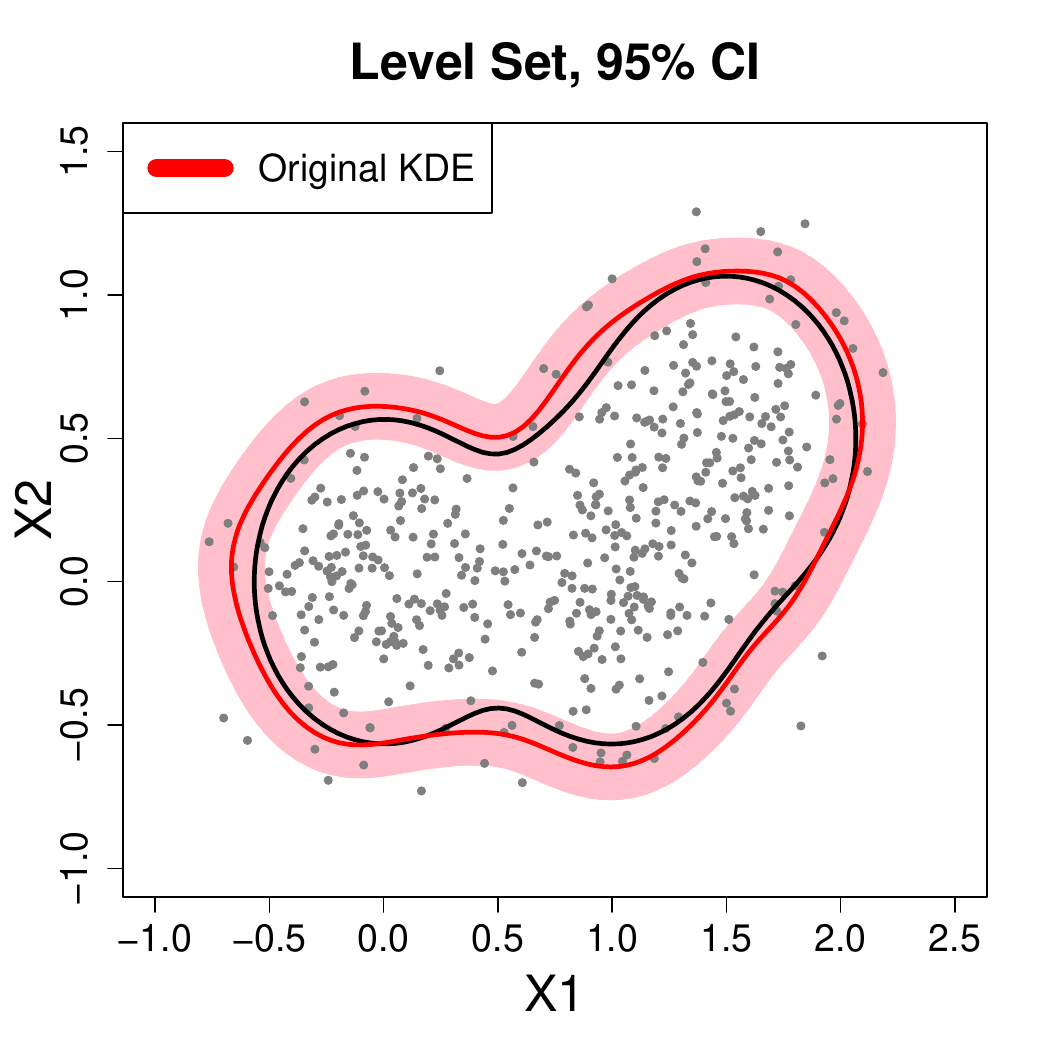}
\includegraphics[height=1.5in]{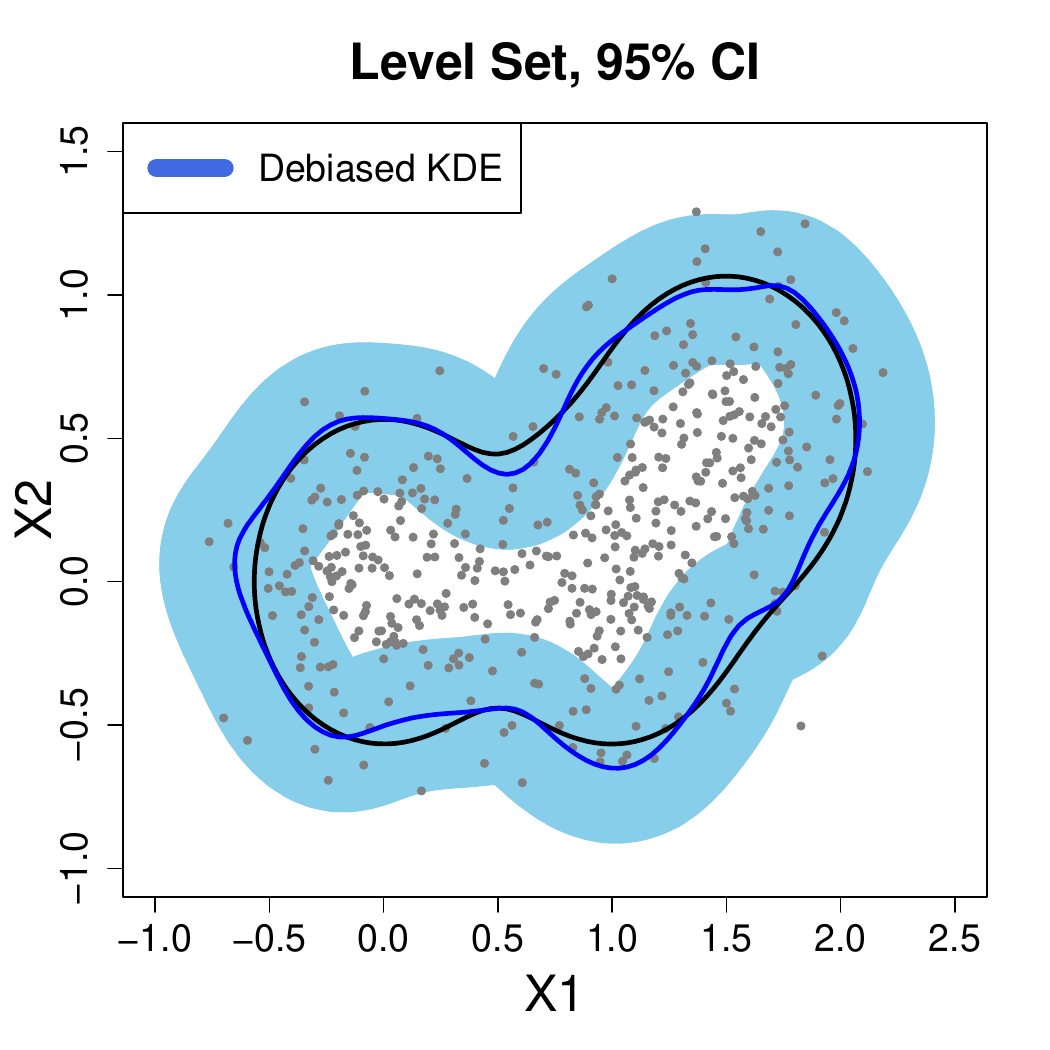}
\includegraphics[height=1.5in]{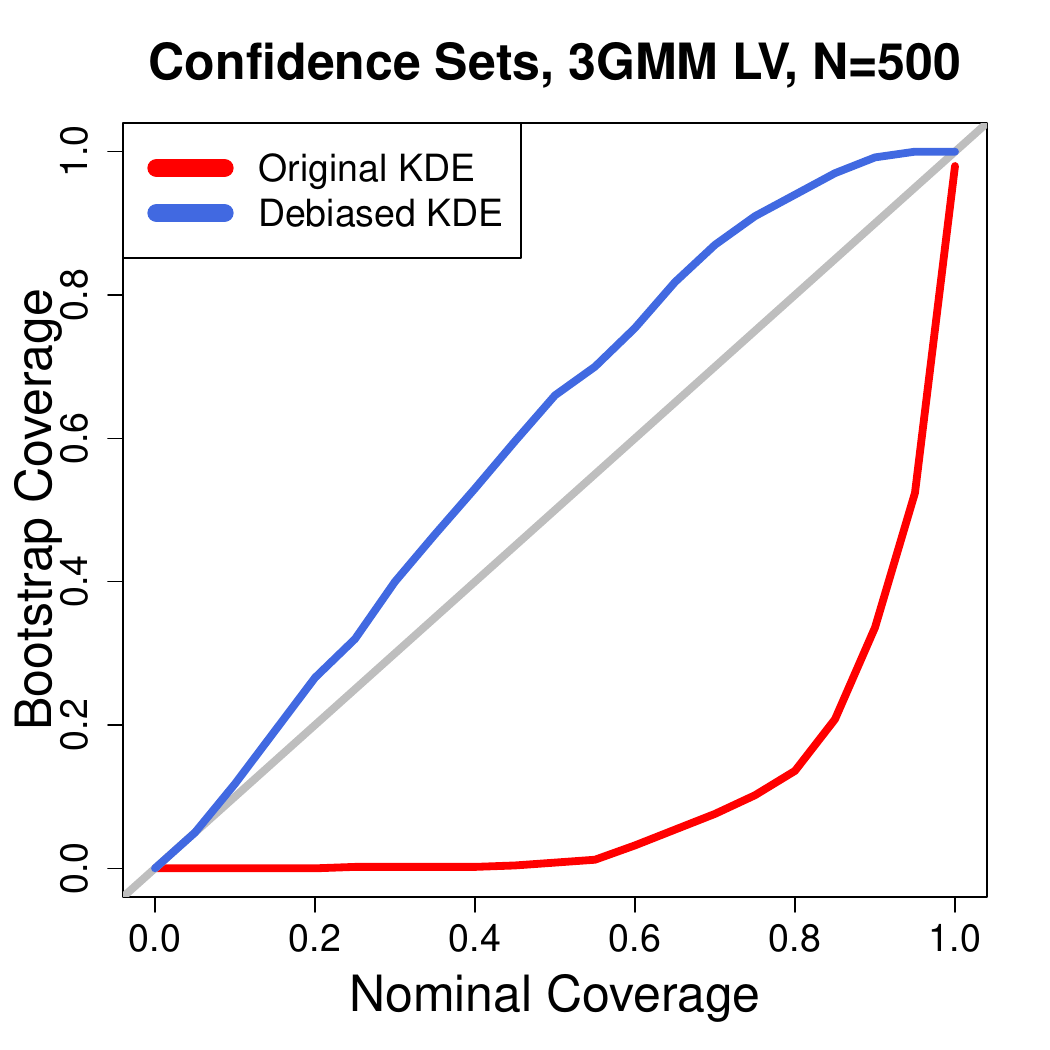}\\
\includegraphics[height=1.5in]{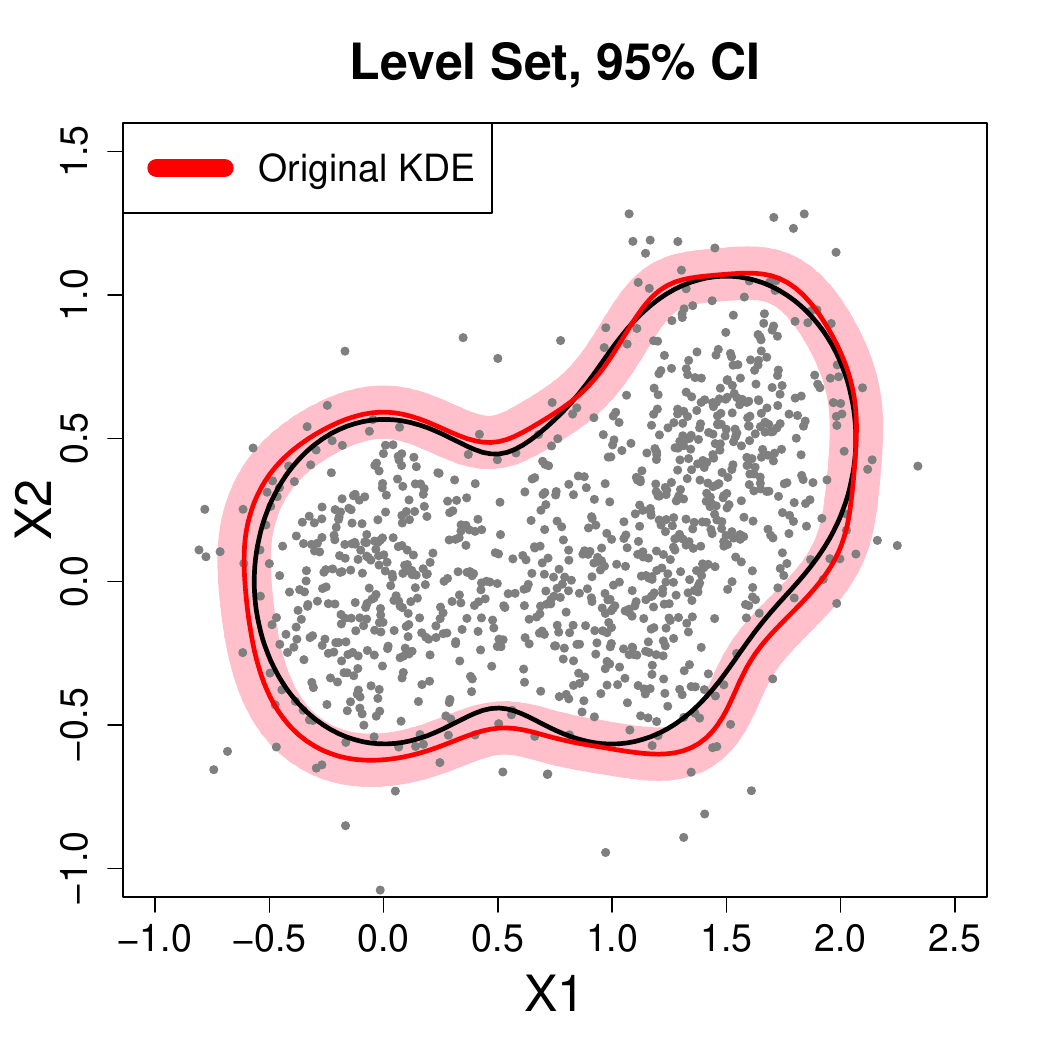}
\includegraphics[height=1.5in]{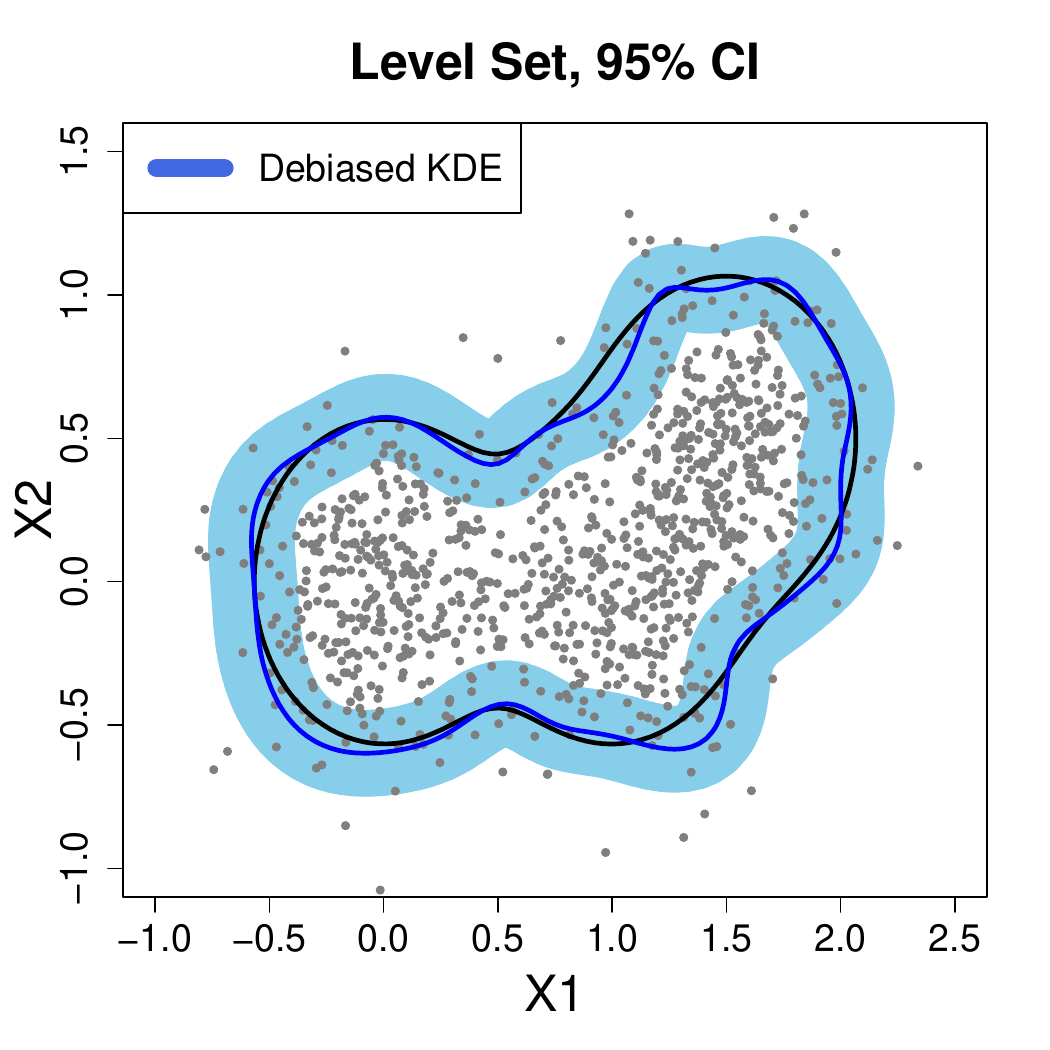}
\includegraphics[height=1.5in]{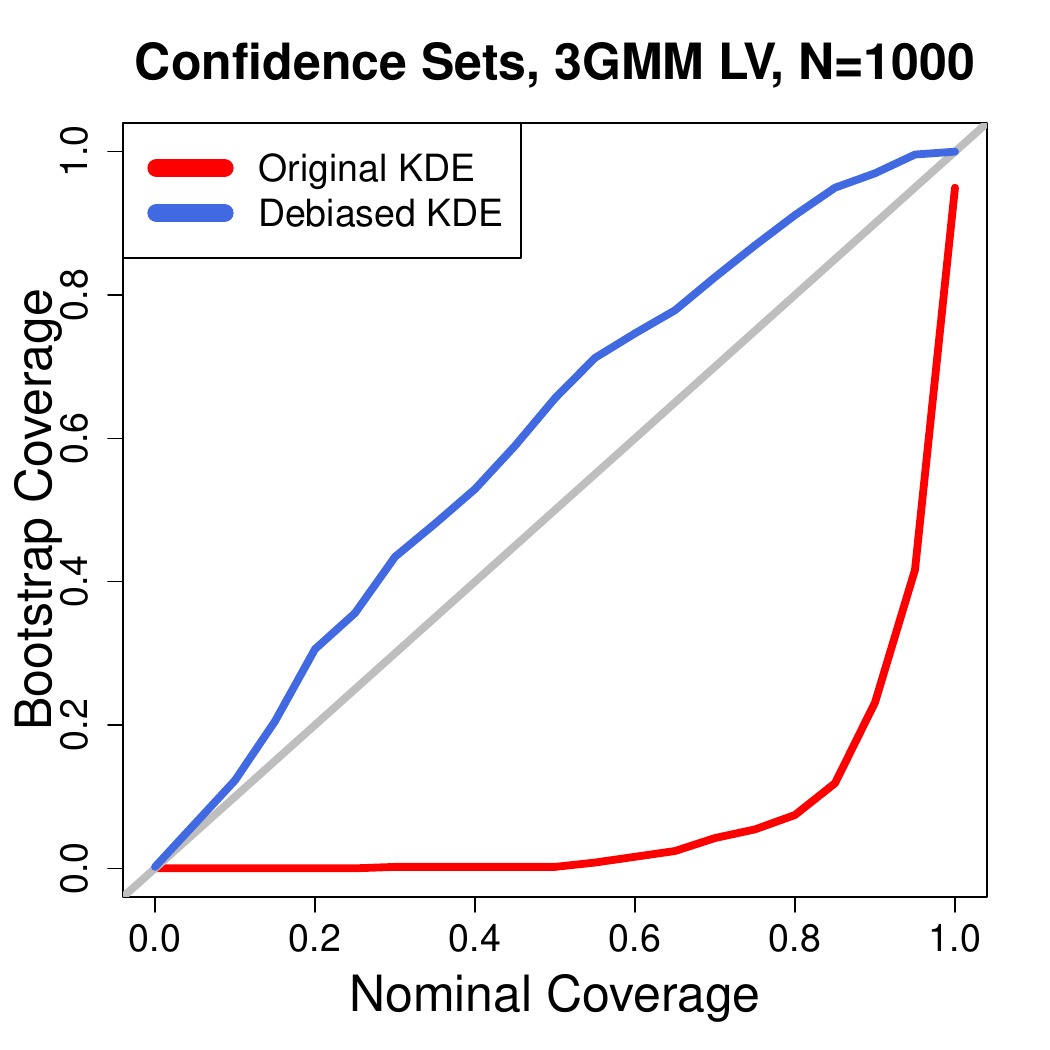}
\caption{Confidence sets of level sets.
In the first column, we display one instance of data points along with the true level contour (black curve),
the estimated level contour using the usual KDE (red curve), and the associated confidence set (red area).
The second column is similar to the first column, but we now use the level set estimator from the debiased KDE (blue curve)
and the blue band is the associated confidence set. 
The third column shows the coverage of the bootstrap confidence set and the nominal level.
The top row is the result of $n=500$ and the bottom row is the result of $n=1000$. 
Based on the third column, we see that bootstrapping the original KDE does not give
us a valid confidence set (we are under coverage) but bootstrapping the debiased KDE
does yield an asymptotically valid confidence set. }
\label{fig::LV01}
\end{figure}

{\bf Level sets.}
Next, we consider constructing the bootstrapped confidence sets of level sets. 
We generate the data from a Gaussian mixture model with three components: 
$$
N((0,0)^T, 0.3^2\cdot \mathbf{I}_2),\quad
N((1,0)^T, 0.3^2\cdot \mathbf{I}_2),\quad
N((1.5,0.5)^T, 0.3^2\cdot \mathbf{I}_2),
$$
where $\mathbf{I}_2$ is the $2\times 2$ identity matrix.
We have equal probability ($1/3$) to generate a new observation from each of the three Gaussians.
We use the level $\lambda=0.25$. 
This model has been used in \cite{chen2017density}.
The black contours in the left two columns of Figure~\ref{fig::LV01}
provide examples of the corresponding level set $D$.

We consider two sample sizes: $n=500$ and $1000$.
We choose the smoothing bandwidth by the rule of thumb \citep{Chacon2011,Silverman1986}
and apply the bootstrap 1000 times to construct the confidence set.
We repeat the entire procedure 1000 times to evaluate coverage, and the coverage plot is given in 
the right column of Figure~\ref{fig::LV01}. 
In both cases, the red curves are below the gray line (45 degree line).
This indicates that bootstrapping the usual level set does not give us a valid confidence set;
the bootstrap coverage is below nominal coverage.
On the other hand, the blue curves in both panels are close to the gray line,
showing that bootstrapping the debiased KDE does yield a valid confidence set.

\subsection{Simulation: Regression}	\label{sec::sim::reg}

Now we show that bootstrapping the debiased local linear smoothers yields
a valid confidence band/set of the regression function and inverse regression.

\begin{figure}[h]
\centering
\includegraphics[height=1.5in]{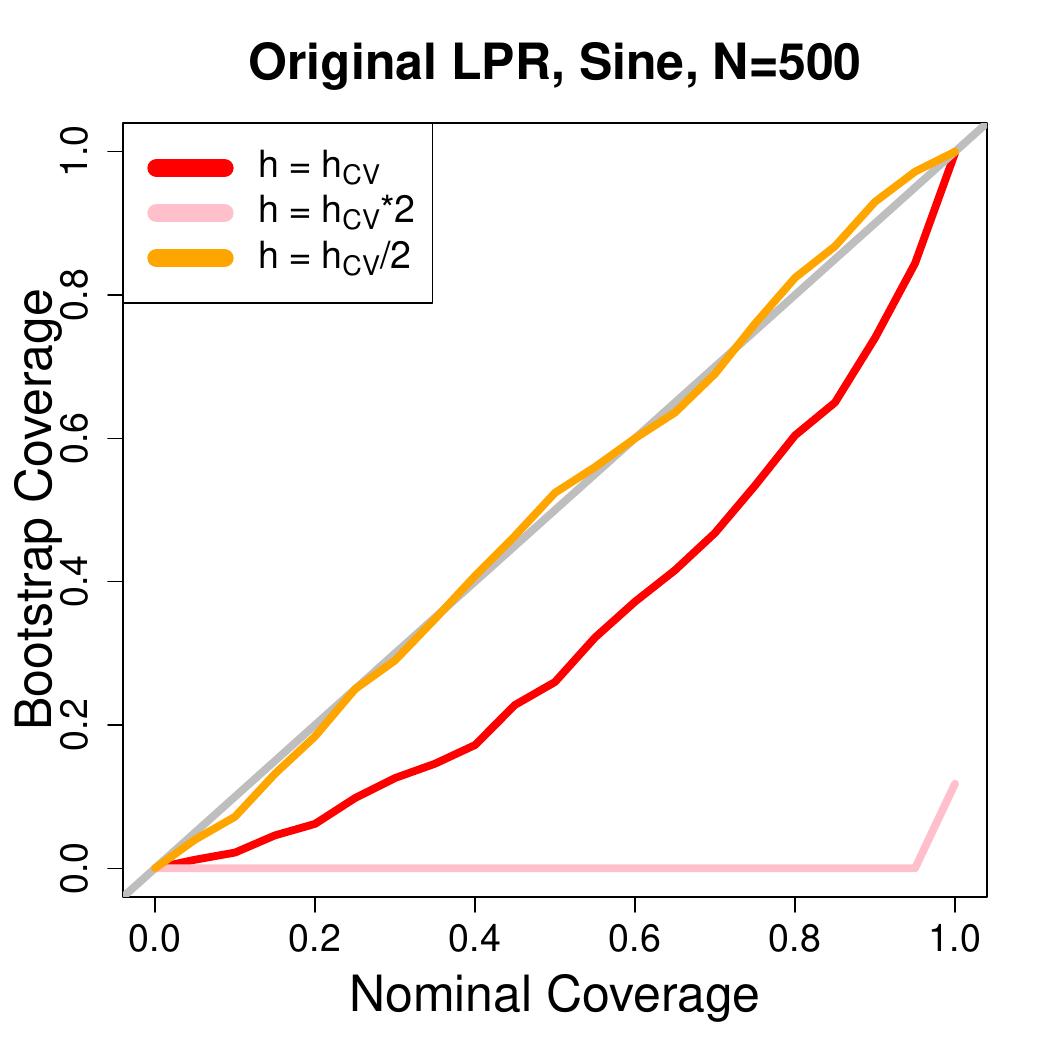}
\includegraphics[height=1.5in]{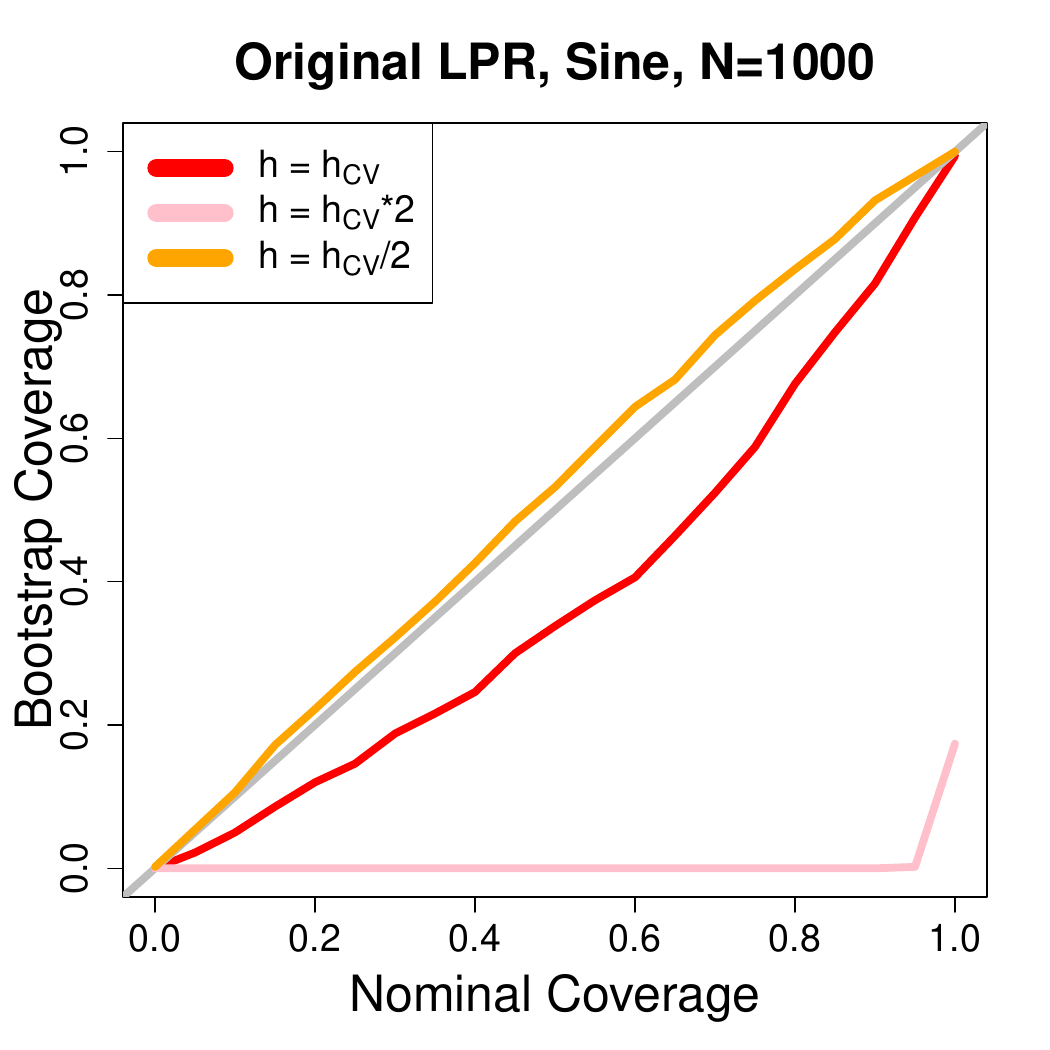}
\includegraphics[height=1.5in]{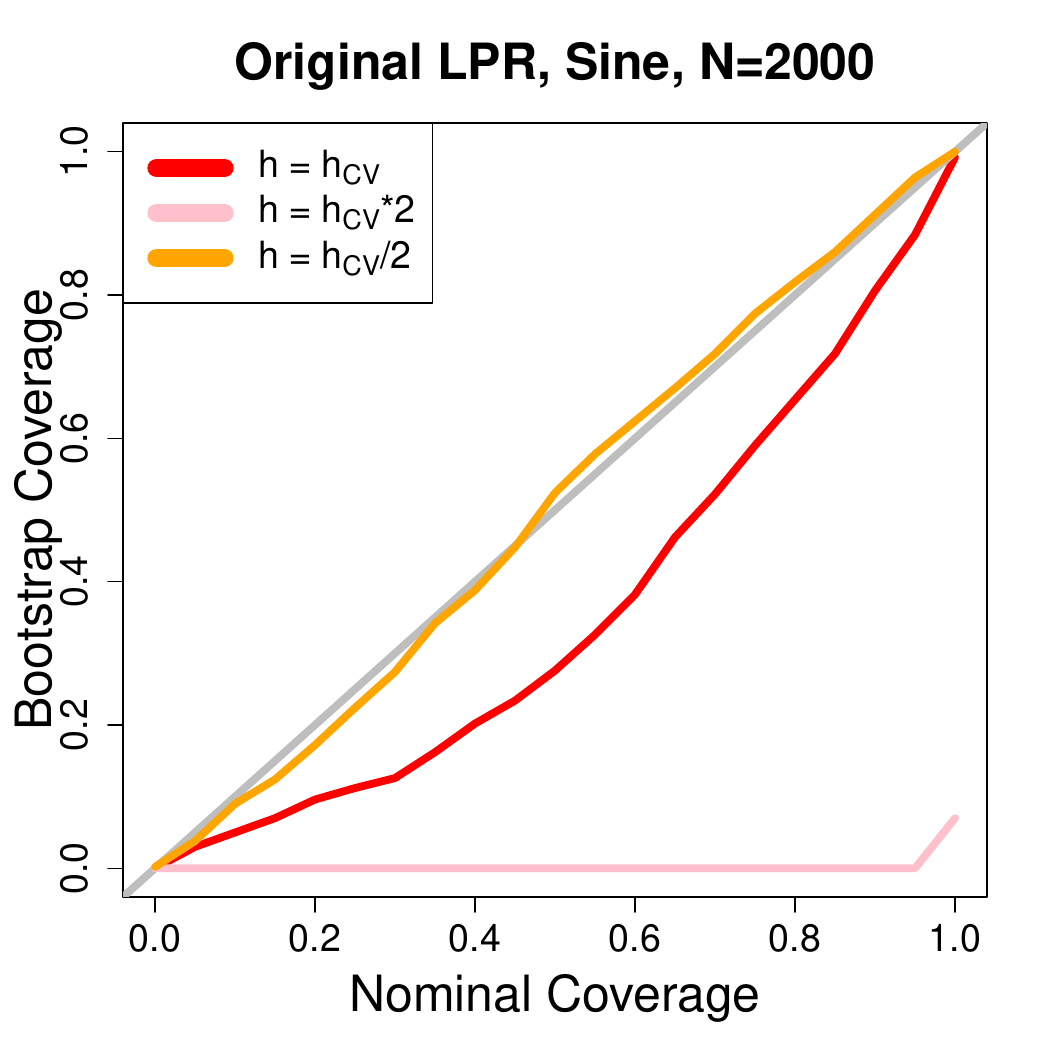}\\
\includegraphics[height=1.5in]{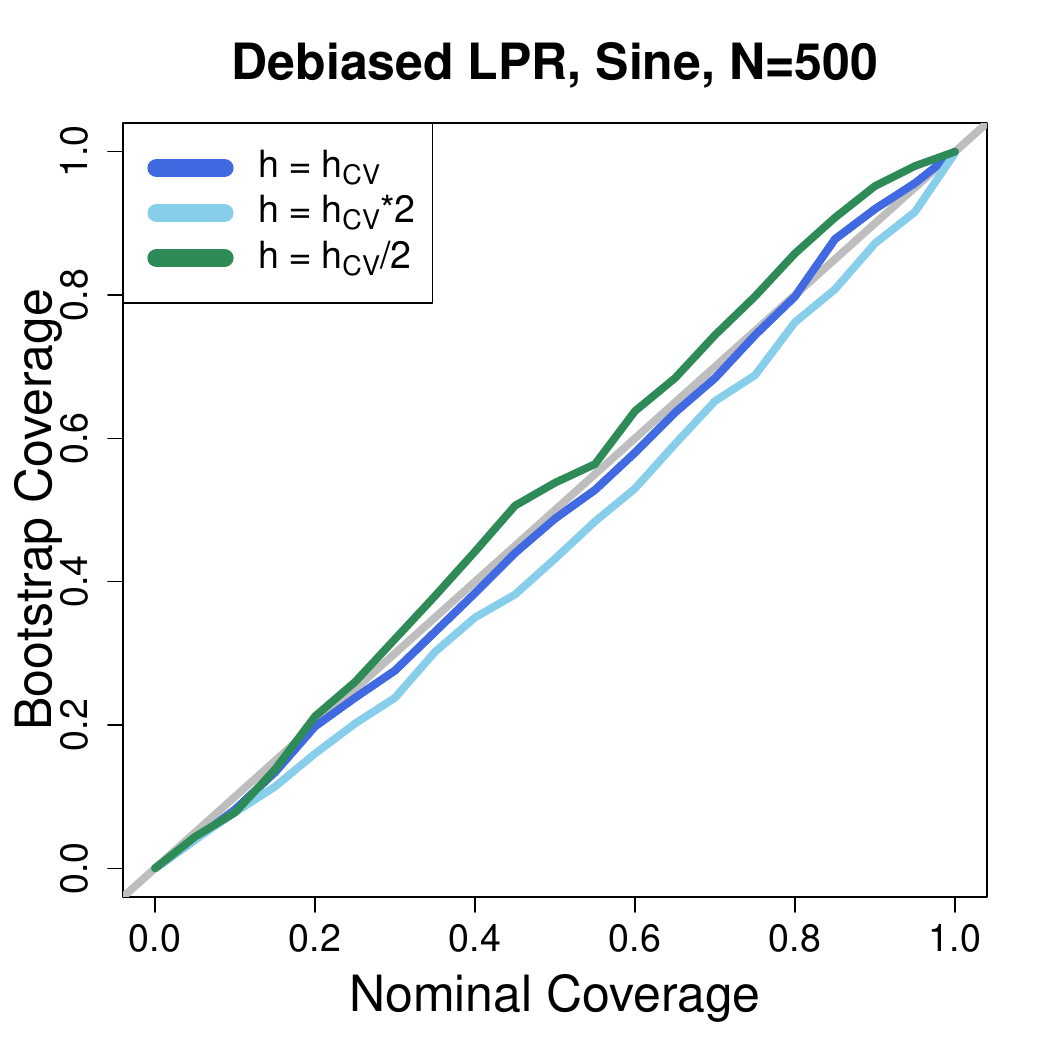}
\includegraphics[height=1.5in]{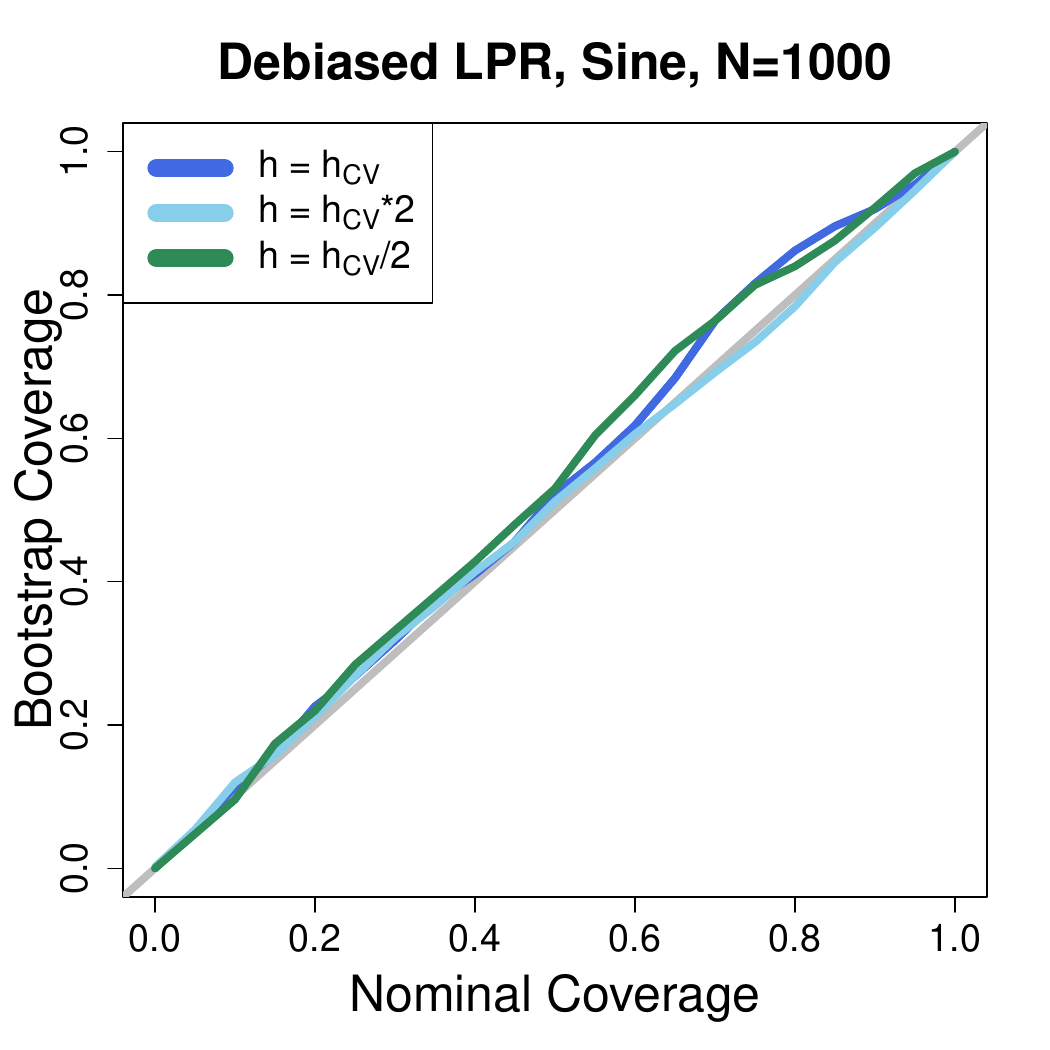}
\includegraphics[height=1.5in]{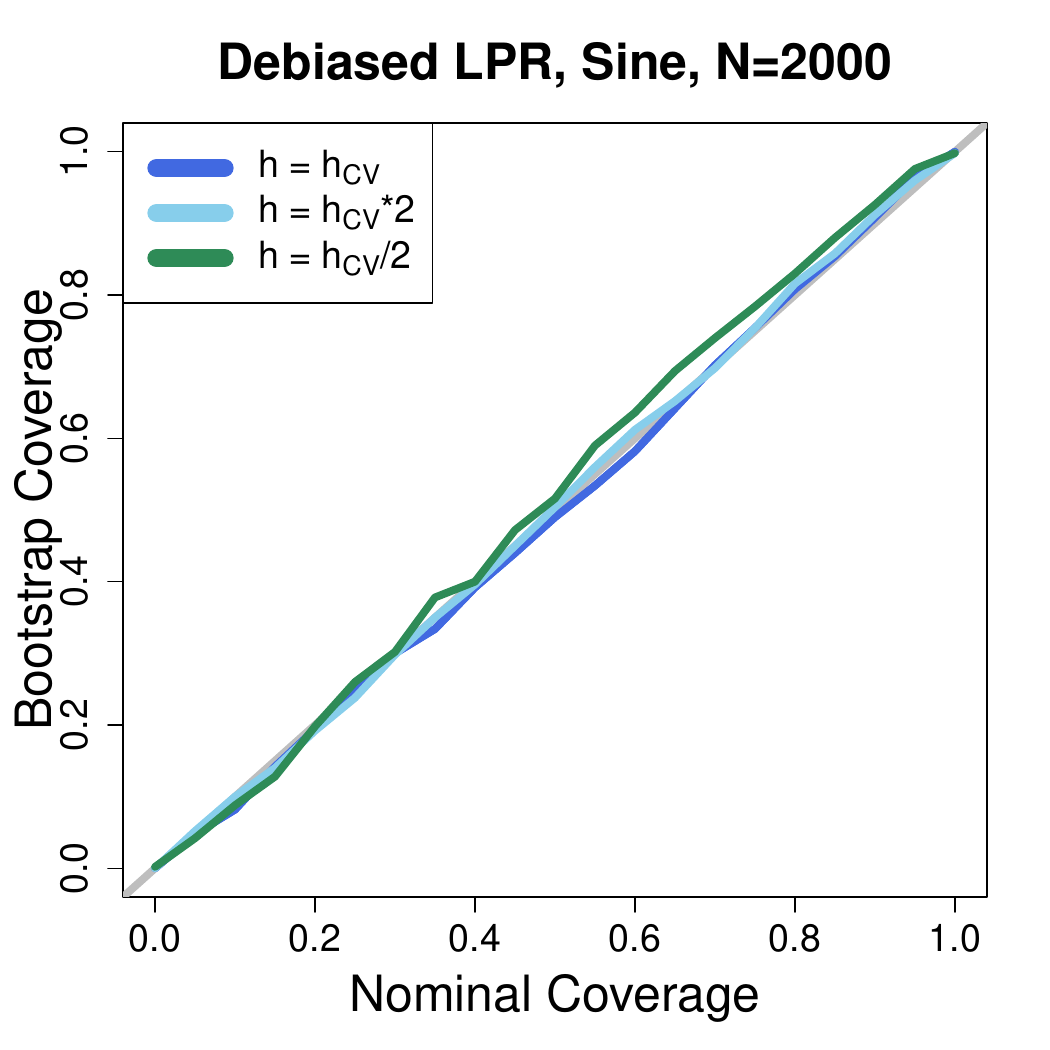}
\caption{Confidence band of regression. 
We use the same `sine' dataset as in Figure~\ref{fig::ex01} and consider three sample sizes: 
$n=500, 1000$, and $2000$.
And we consider $3$ different smoothing bandwidths: $h_{CV}$, $h_{CV}\times 2$, and $h_{CV}/2$,
where $h_{CV}$ is the bandwidth from $5$-fold cross-validation on the original local linear smoother.
The top row is the bootstrap coverage of the local linear smoother without debiasing.
The bottom row shows the bootstrap coverage of the debiased local linear smoother. 
We see a clear pattern that the debiased local linear smoother attains nominal coverage for
all three bandwidths. On the other hand, only in the undersmoothing case ($h_{CV}/2$) does the original local linear smoother have
nominal coverage.}
\label{fig::reg}
\end{figure}

{\bf Regression functions.}
To show the validity of confidence bands, we generate pairs of random variables $(X,Y)$
from
\begin{align*}
X&\sim {\sf Unif}[0,1],\\
Y&= \sin(\pi \cdot X) + \epsilon,\\
\epsilon&\sim N(0,0.1^2),
\end{align*}
where $X$ and $\epsilon$ are independent.
This is 
the same as the model used in the bottom row of Figure~\ref{fig::ex01}.
In the bottom left panel of Figure~\ref{fig::ex01}, we display
the true regression function (black curve),
the original local linear smoother (red curve), and the debiased local linear smoother (blue curve).
We consider three sample sizes: $n=500, 1000,$ and $2000$.
The smoothing bandwidth $h_{CV}$ is chosen using a $5$-fold cross-validation of the original local linear smoother. 
In addition to $h_{CV}$, we also consider $h_{CV}\times 2$ and $h_{CV}/2$
to show the robustness of the confidence bands against bandwidth selection. 
We then bootstrap both the original local linear smoother and the debiased local linear smoother
to construct confidence bands. 
Note that we restrict ourselves to the regions $[0.1,0.9]\subset [0,1]$, which is a subset of the support 
to avoid the problem of insufficient data points in the boundary.
The result is shown in Figure~\ref{fig::reg}. 
In the top panel, 
we present the coverage of bootstrapping the original local linear smoother. 
Only in the case of $h_{CV}/2$ (undersmoothing) do the confidence bands attain nominal coverage.
This makes sense because when we are undersmoothing the data, the bias vanishes faster than the stochastic 
variation so the bootstrap confidence bands are valid.
In the bottom panel, we present the coverage of bootstrapping the debiased local linear smoother.
It is clear that all curves are around the gray line, which means that 
the confidence bands attain nominal coverage in all the three smoothing bandwidths. 
Thus, this again shows the robustness of the confidence band from the debiased estimator against different bandwidths.  

To further investigate the property of confidence bands,
we apply the same analysis as in the KDE that we compare the width of confidence bands from
the debiased estimator (blue) and an undersmoothing estimator (red) in Figure~\ref{fig::lpr_band_width}.
In the top row, we choose the smoothing bandwidth by the rule of thumb
and in the bottom row, we choose the smoothing bandwidth by the $5$-fold cross-validation. 
The width is computed using the median width over $1000$ simulations. 
When the bandwidth is chosen by the rule of thumb, 
the two confidence bands have a very similar width.
However, when we use the $5$-fold cross-validation, the debiased estimator
has a confidence band with a narrower width. 

{We also compared our approaches to several other methods on constructing uniform confidence bands, including undersmoothing (US), off-the-shelf R package \texttt{locfit} \citep{loader2013locfit}, traditional bias correction (BC),  robust bias correction  in \cite{calonico2018effect} (Robust BC), and the simple bootstrap method of \citep{hall2013simple}(HH). 
The data are generated by the following model: 
\begin{align*}
X&\sim {\sf Unif}[-1,1],\\
Y&= sin(3\pi x/2)/ (1 + 18 x^2 [\text{sign}(x) + 1]) + \epsilon,\\
\epsilon&\sim N(0,0.1^2),
\end{align*}
This function was previously used by \cite{berry2002bayesian, hall2013simple, calonico2018effect} to construct pointwise confidence intervals.  We run simulations for 1000 times with each method.  For all but robust bias correction method, the smoothing bandwidth $h_0$ was chosen by the cross validation using \texttt{regCVBwSelC} method or rule of thumb using \texttt{thumbBw} with gaussian kernel both from the \texttt{locpol}\citep{cabrera2018locpol} package \footnote{In our experiments, we adjust the bandwidth for \texttt{locfit} by multiplying it by 2.5 since it looks like that \texttt{locfit} has some "automatic undersmoothing" effect when fitting a local linear smoother, we visually check the smoothness of resulting estimator and found that multiplying it by 2.5 gives similar result to other local linear packages}. For robust bias correction method, we use its own bandwidth selection algorithm.  Again we restrict the uniform confidence band to the regions $[-0.9, 0.9] \in [-1, 1]$. }

{Specifically, the undersmoothing method uses bandwidth $h_0 / 2$ to perform bootstrap with original local linear smother. For traditional bias correction method,  we use a second bandwidth for estimating the second order derivative with cross validation or rule of thumb for the third-order local polynomial regression\footnote{again using either \texttt{regCVBwSelC} or \texttt{thumbBw} from the \texttt{locpol} package}.  For both undersmoothing and traditional bias correction methods, we apply a similar bootstrap strategy as in Figure \ref{fig::alg::reg} and bootstrap 1000 times as in our debiased approach.  Further, we consider three cases with $n = 500, 1000, 2000$.  Notice that only undersmoothing, traditional bias correction and locfit \citep{sun1994simultaneous} are tailored for uniform confidence band, HH method \citep{hall2013simple} and robust bias correction\citep{calonico2018effect} are only for pointwise confidence intervals.  We do not report the results for HH method since it is especially bad for uniform coverage as there would be ``exptected 100$\xi$\% of points that are not covered"\citep{hall2013simple}. }
\begin{table}[H]
\caption{Empirical Coverage of 95\% simultaneous confidence band}
\begin{center}
\begin{tabular}{c c c  c c  c  c }
\hline
 & & \multicolumn{5}{c}{Empirical coverage} \\
 \hline
 n & BW Selection & US & locfit & BC &  Robust BC & Debiased \\
 \hline
 500 & CV & 0.993  & 0.848 & 0.959 & 0.074 & 0.976 \\
  & ROT & 0.968 & 0.313 & 0.946  &  - &  0.963 \\
1000 & CV & 0.99 & 0.872 & 0.965 & 0.052 & 0.976 \\
  & ROT & 0.971 & 0.28 & 0.935 & -&  0.961 \\
2000 & CV & 0.982 & 0.862 & 0.968 & 0.041 & 0.963 \\
 & ROT & 0.963 & 0.233 & 0.927 & -& 0.965 
\end{tabular}
\end{center}
\label{table::lpr_comparison_coverage}
\end{table}

\begin{table}[H]
\caption{Average width of 95\% simultaneous confidence band}
\begin{center}
\begin{tabular}{c c c  c c  c  c }
\hline
 & &  \multicolumn{5}{c}{Average Confidence Band width} \\
 \hline
 n & Bw Selection &  US & locfit & BC & Robust BC & Debiased \\
 \hline
 500 & CV & 0.122 & 0.061& 0.070 & 0.039 & 0.090 \\
  & ROT & 0.085  &  0.049 & 0.060 & - & 0.072 \\
1000 & CV & 0.081 & 0.047 & 0.051 & 0.030 & 0.066 \\
  & ROT & 0.060 & 0.037 & 0.043 & -& 0.052\\
2000 & CV & 0.057 & 0.035 & 0.038 & 0.023 & 0.049 \\
 & ROT & 0.044 & 0.028 & 0.031 & -& 0.038
\end{tabular}
\end{center}
\label{table::lpr_comparison_width}
\end{table}

{Table \ref{table::lpr_comparison_coverage} and \ref{table::lpr_comparison_width} display the empirical coverage and average confidence band width over 1000 replications. It appears that our debiased approach and undersmoothing approach always achieve the nominal coverage.  Traditional bias correction also works pretty well with cross validated bandwidth and undercovers only a bit with rule of thumb bandwidth. Our debiased approach has a narrower confidence band compared to the undersmoothing approach, but is wider than traditional bias correction. It is interesting that the traditional bias correction is working very well combined with bootstrap strategy. The consistent estimation of the bias seems to be helping with the confidence band in this case. 
Note that the only difference between the traditional bias correction approach and our approach is that our approach uses the same smoothing
bandwidth for both regression function estimation and bias correction whereas the bias correction approach uses two different
smoothing bandwidth.
Locfit and the pointwise robust bias correction interval always undercovers.  More simulations are provided in Appendix \ref{sec::REG::sim::app}. }


\begin{figure}
\centering
\includegraphics[height=1.5in]{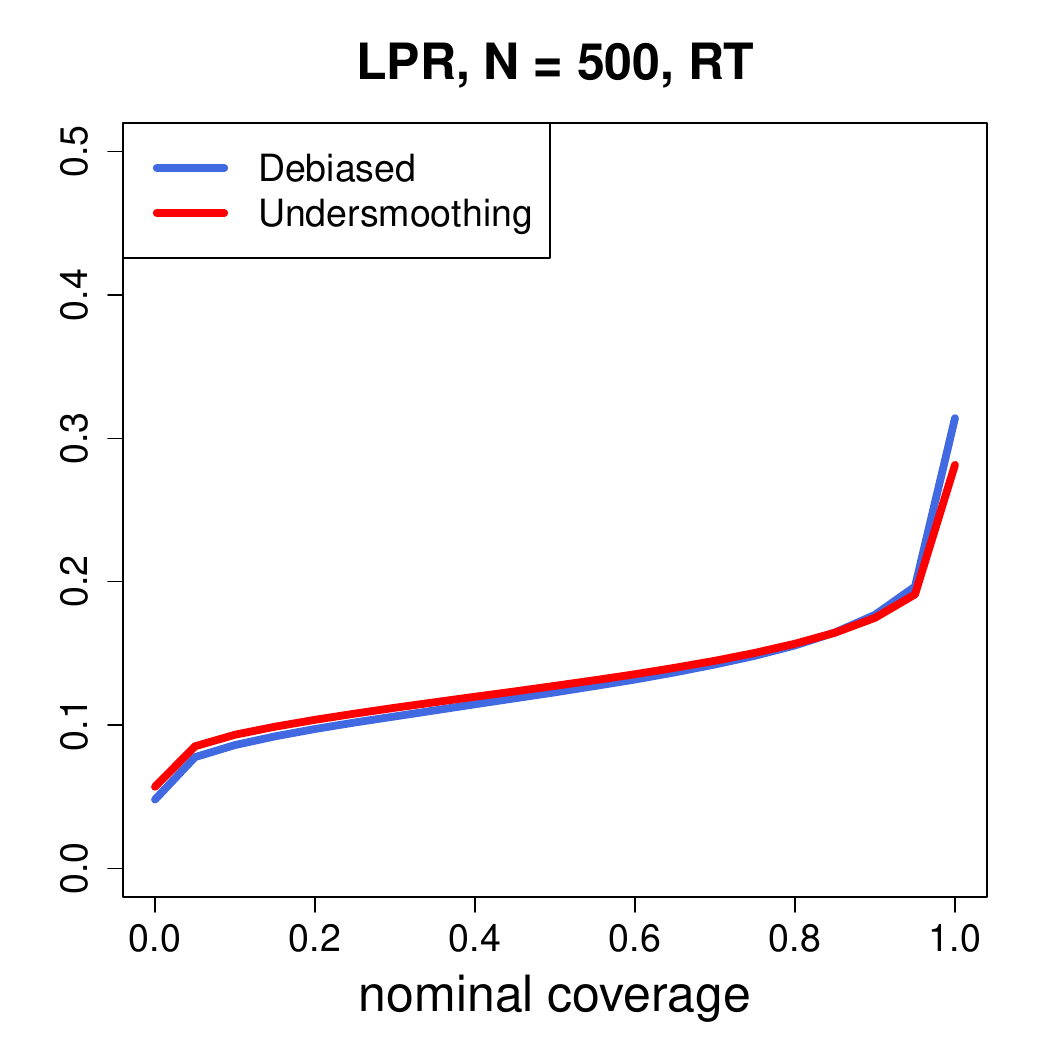}
\includegraphics[height=1.5in]{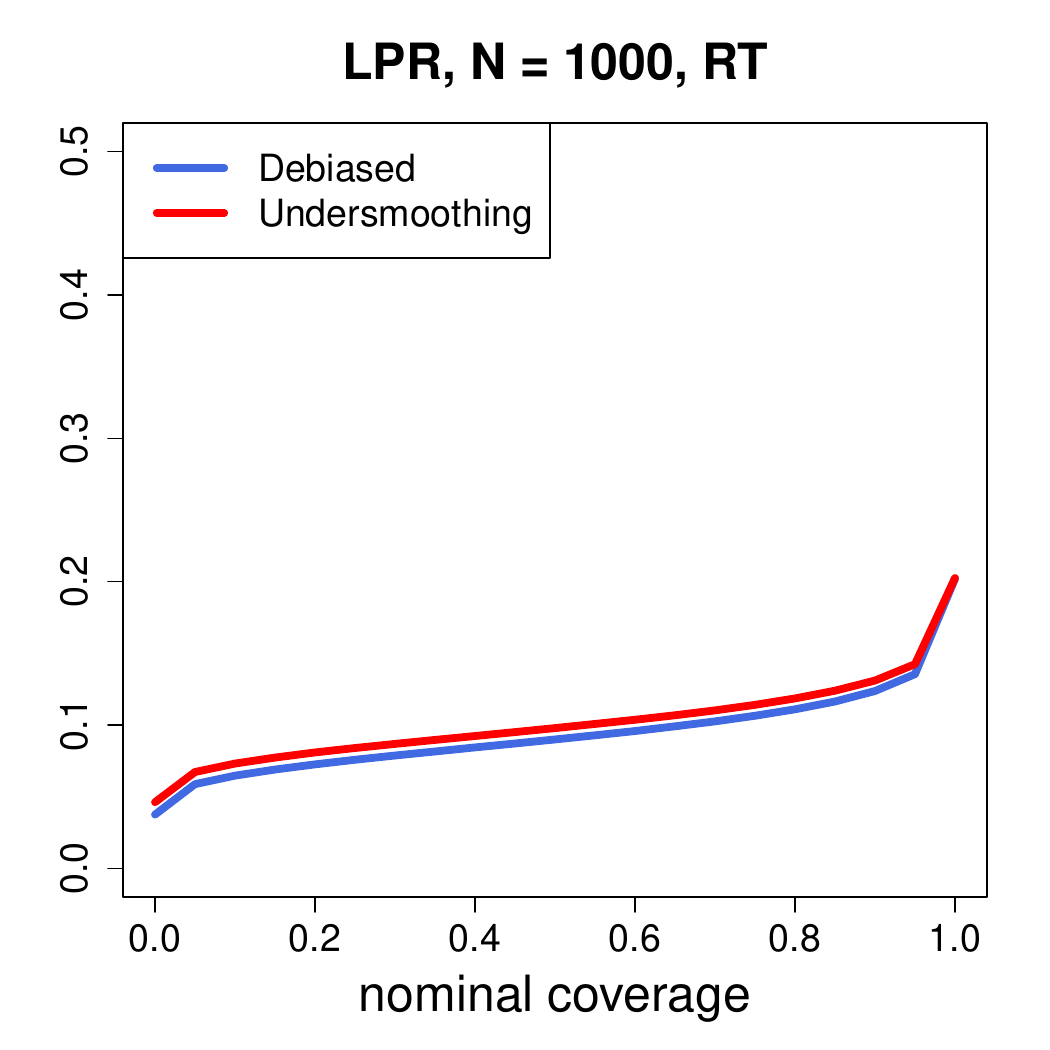}
\includegraphics[height=1.5in]{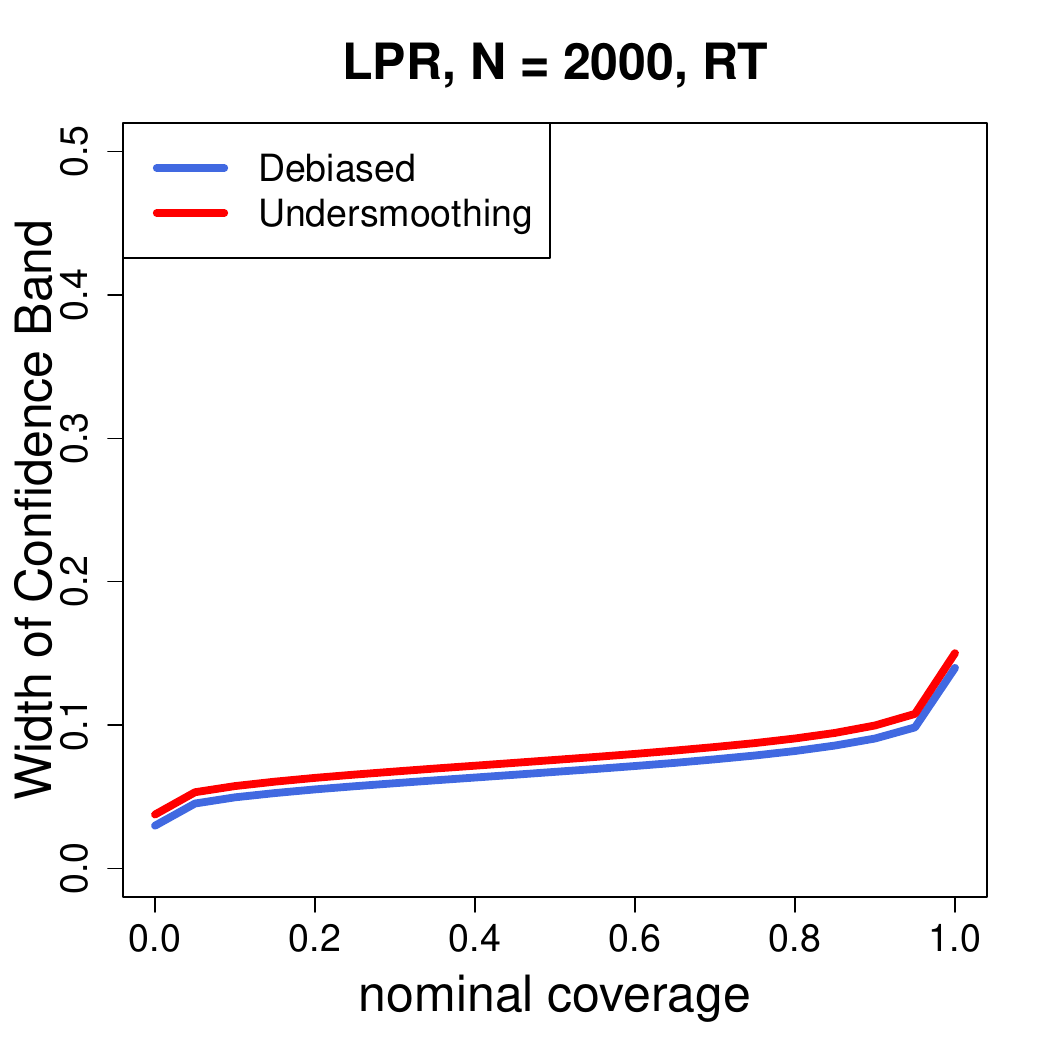}\
\includegraphics[height=1.5in]{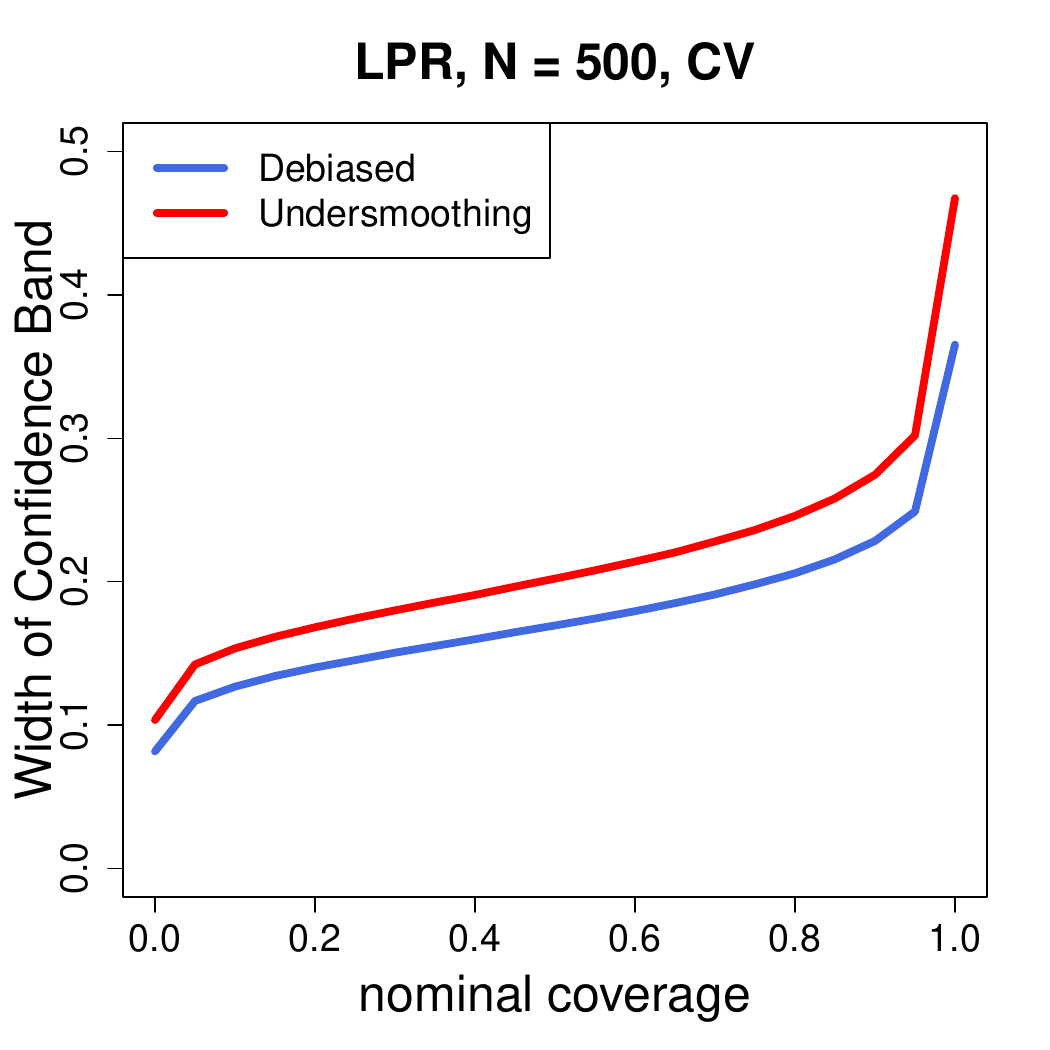}
\includegraphics[height=1.5in]{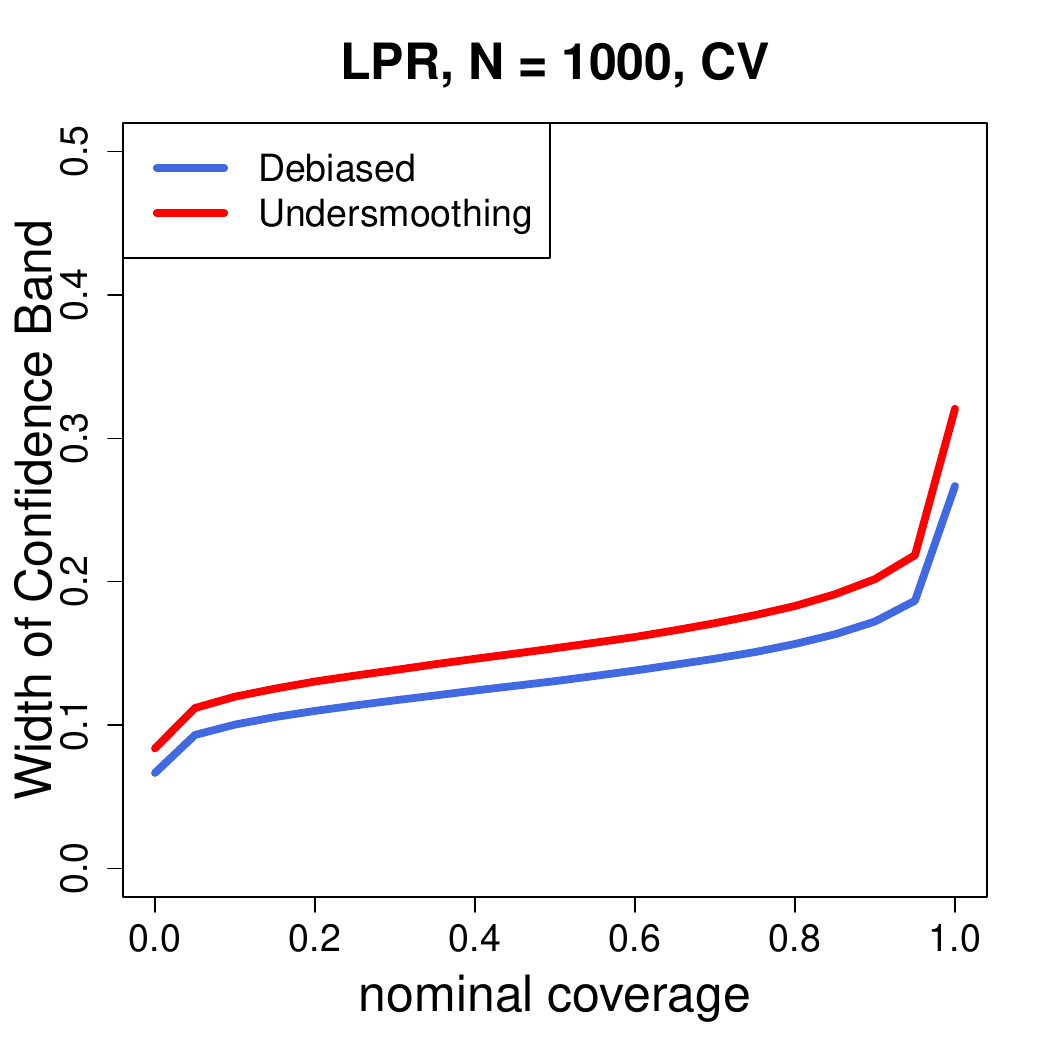}
\includegraphics[height=1.5in]{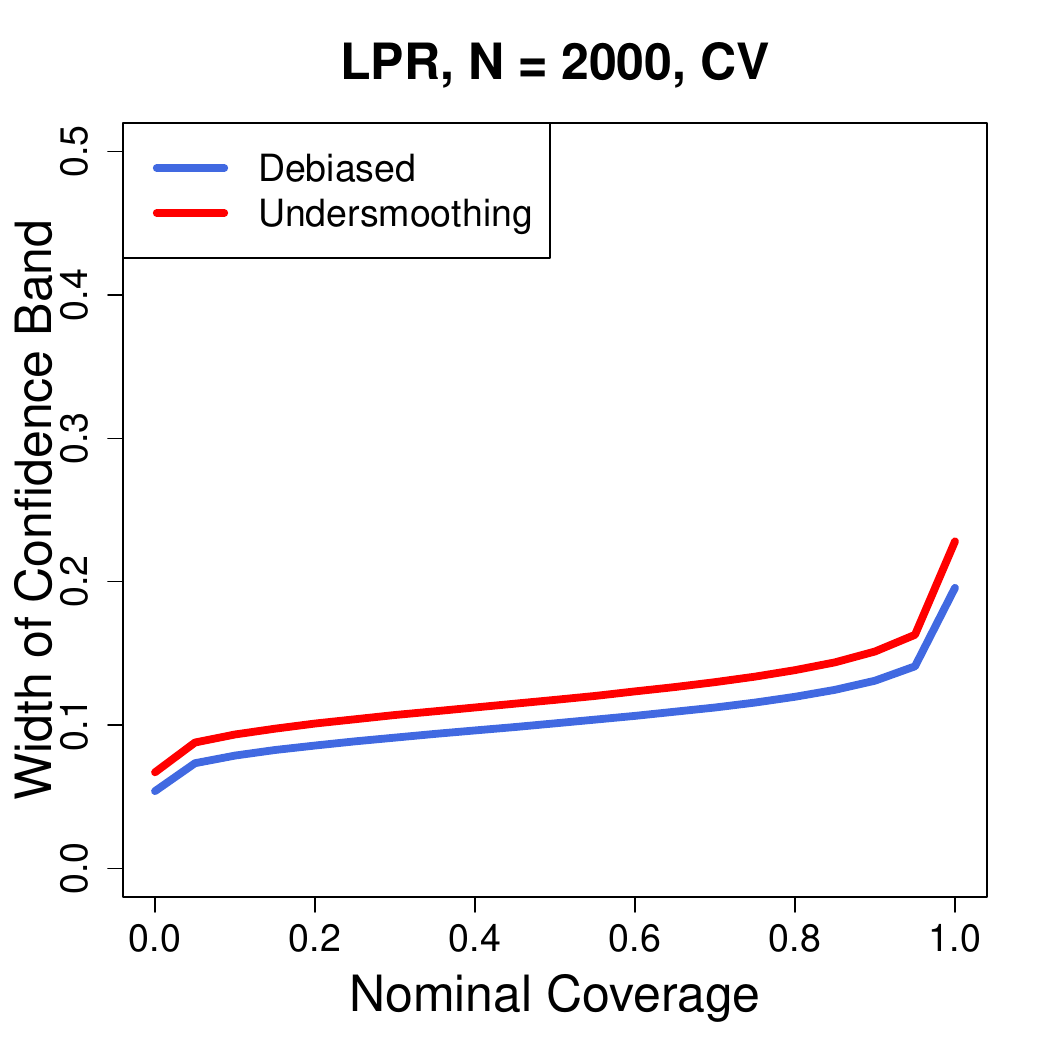}
\caption{Comparison of width of confidence bands. 
The top row corresponds to the case where bandwidth $h$ is chosen by the rule of thumb.
The bottom row corresponds to the case where bandwidth $h$ is chosen by the $5$-fold cross validation. 
We are comparing the width of confidence bands from the debiased estimator and from an undersmoothing estimator (in our case, $h/2$, the same idea as FIgure~\ref{fig::kde_band_width}).  
The width is computed using the median width of 1000 simulations.  
When we use the rule of thumb, the confidence band for both methods are very similar
but in the case of cross validation, the confidence band for debiased estimator is narrower than the undersmoothing method. 
}
\label{fig::lpr_band_width}
\end{figure}

\begin{figure}[h]
\centering
\includegraphics[height=1.5in]{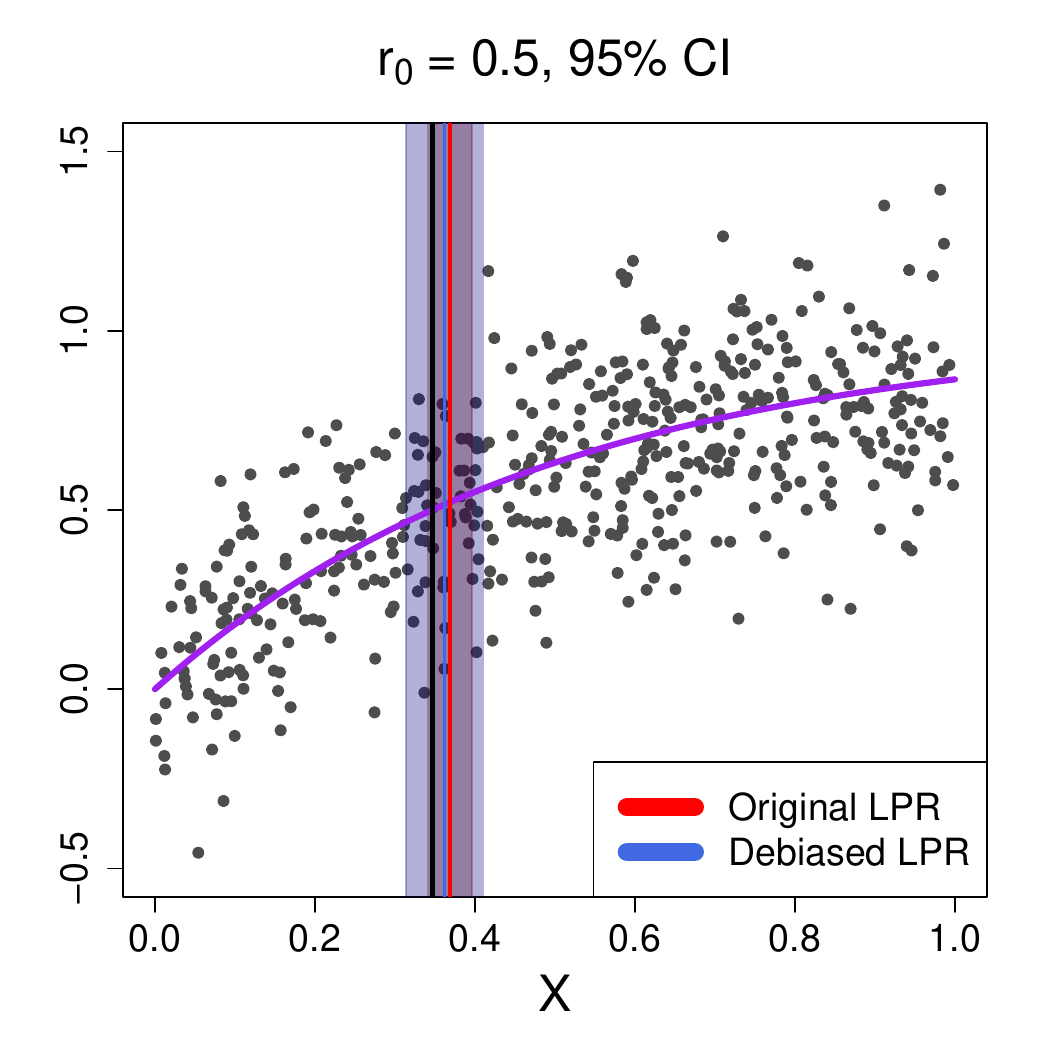}
\includegraphics[height=1.5in]{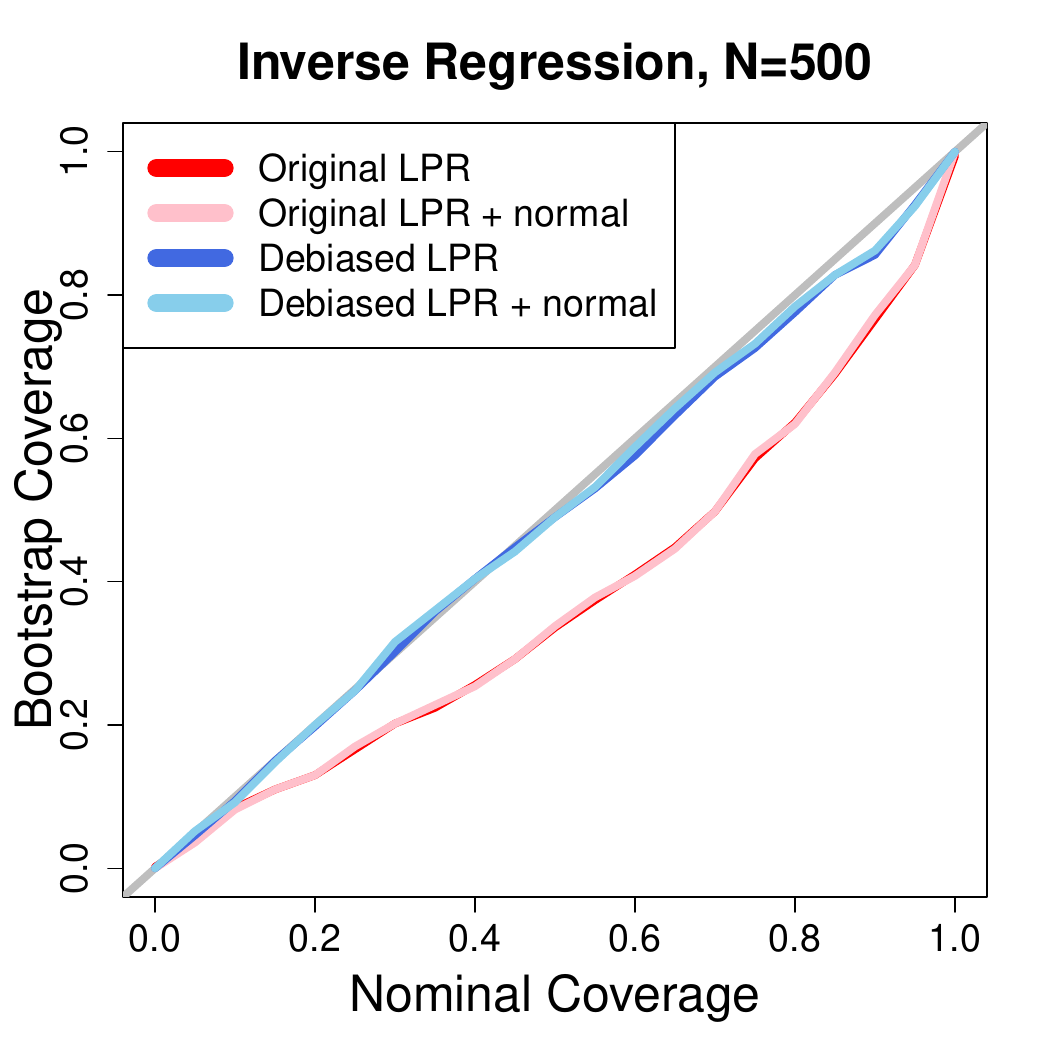}\\
\includegraphics[height=1.5in]{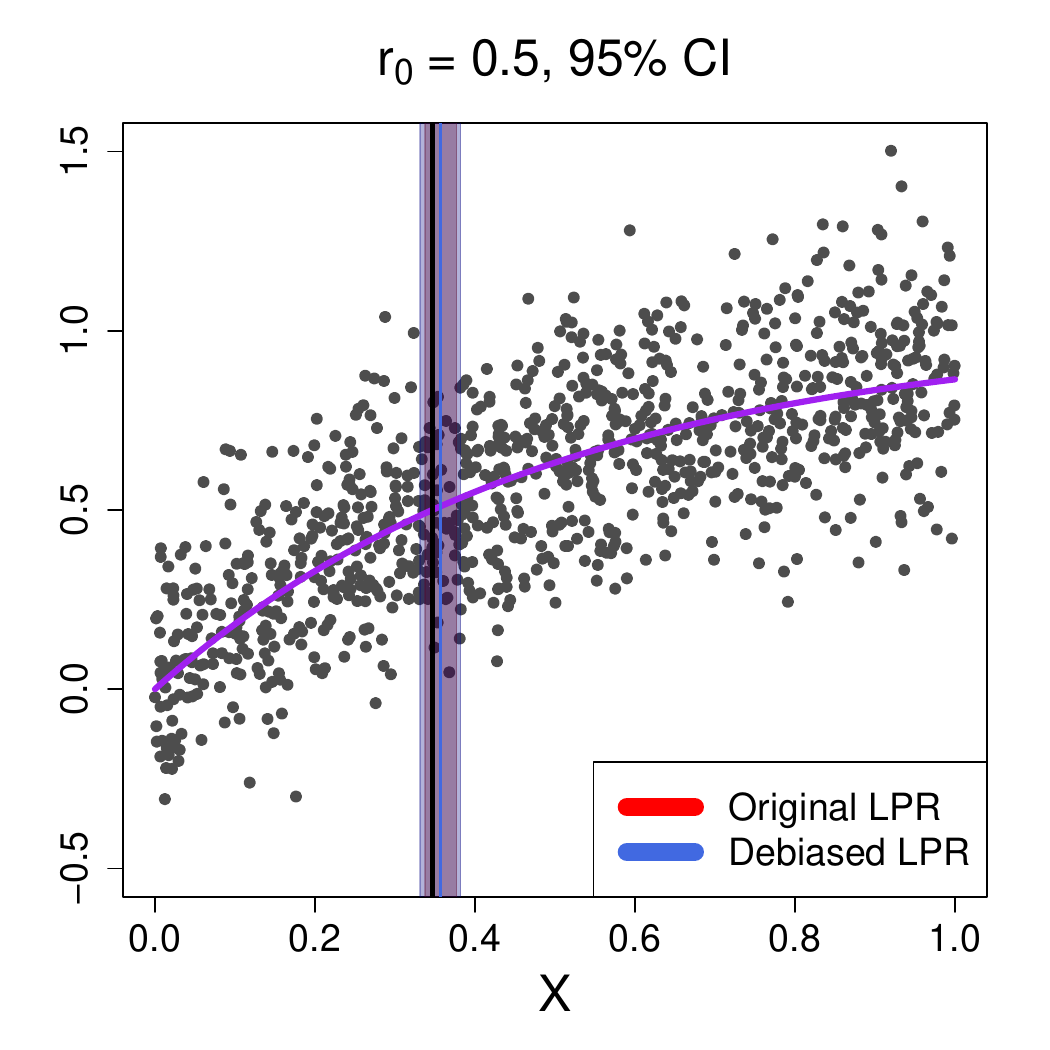}
\includegraphics[height=1.5in]{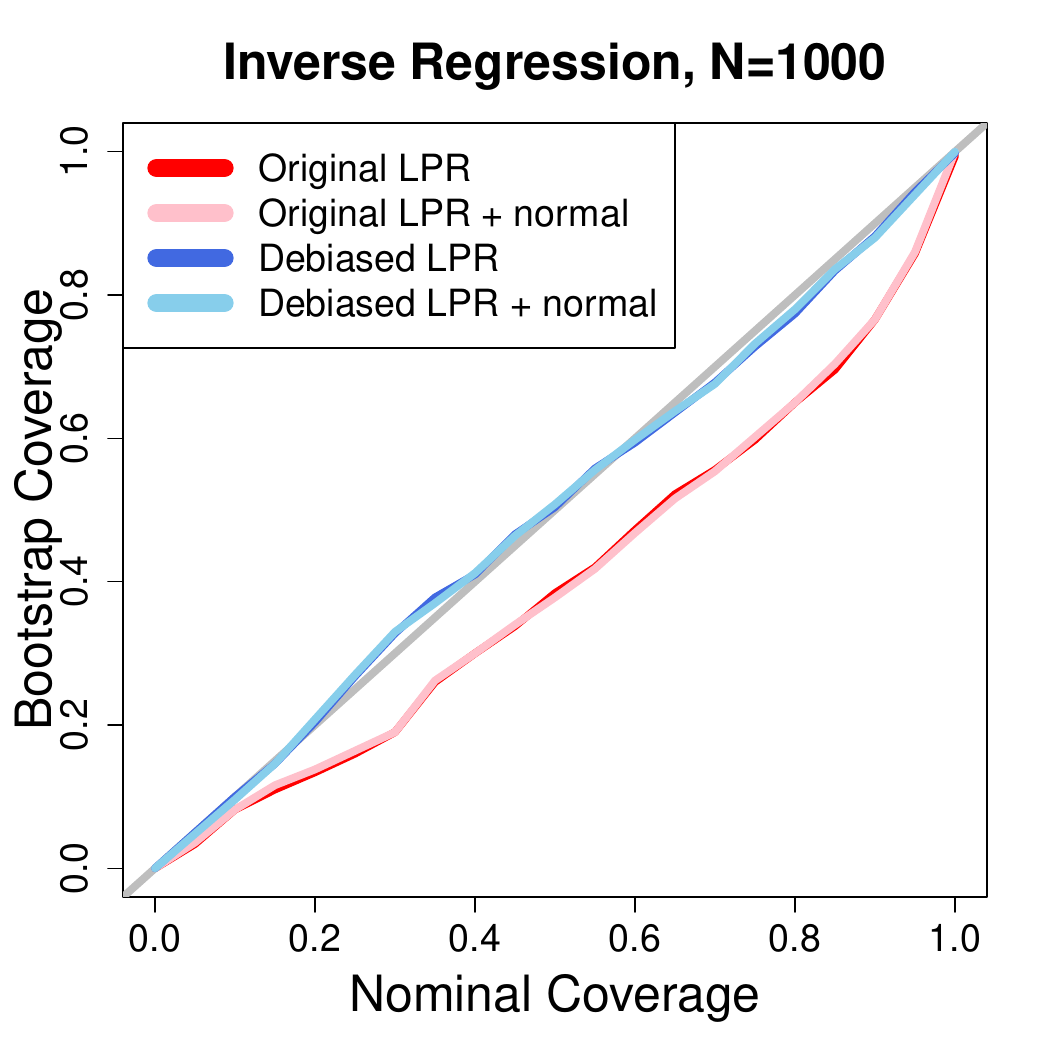}
\caption{Confidence sets of the inverse regression.
In the left column, we display one instance of the bootstrap confidence set 
using the local linear smoother (red region) and debiased local linear smoother (blue region).
The purple curve shows the actual regression line and the black vertical line shows the 
location of the actual inverse regression ($r_0=0.5$).
In the right column, we provide bootstrap coverage for both local linear smoother (red) and the debiased local linear smoother (blue). 
We also consider the confidence set using normality and bootstrap (in a lighter color).
The top row is the case of $n=500$ and the bottom row is the case of $n=1000$. 
}
\label{fig::inv01}
\end{figure}

{\bf Inverse regression.}
The last simulation involves inverse regression.
In particular, we consider the case where $\cR $ contains a unique point,
so we can construct a confidence set using both the bootstrap-only approach and normality with the bootstrap variance estimate. 
The data are generated by the following model:
\begin{align*}
X&\sim {\sf Unif}[0,1],\\
Y&= 1-e^{X} + \epsilon,\\
\epsilon&\sim N(0,0.2^2),
\end{align*}
where $X,\epsilon$ are independent.
Namely, the regression function $r(x) = E(Y|X=x) = 1-e^x$. 
We choose the level $r_0=0.5$, which corresponds to the location
$\cR = \{-\log (2)\}$.
We consider two sample sizes: $n=500$, and $1000$.
We choose the smoothing bandwidth using a $5$-fold cross-validation of the original local linear smoother.
The left column of Figure~\ref{fig::inv01} shows one example of the two sample sizes
where the black vertical line denotes the location of $\cR$, the red line and red band present
the estimator from the original local linear smoother and its confidence set, and
the blue line and blue band display the estimator and confidence set from the debiased local linear smoother. 
We construct the confidence sets by (i) completely bootstrapping (Section~\ref{sec::inv}),
and (ii) the normality with the bootstrap variance estimate.
The right column of Figure~\ref{fig::inv01} presents
the coverage of all four methods.
The reddish curves are the results of bootstrapping the original local linear smoother,
which do not attain nominal coverage.
The bluish curves are the results from bootstrapping the debiased local linear smoother,
which all attain nominal coverage.
Moreover, it seems that using normality does not change the coverage--
the light-color curves (using normality) are all close to the dark-color curves (without normality).

\begin{figure}
\centering
\includegraphics[height=1.5in]{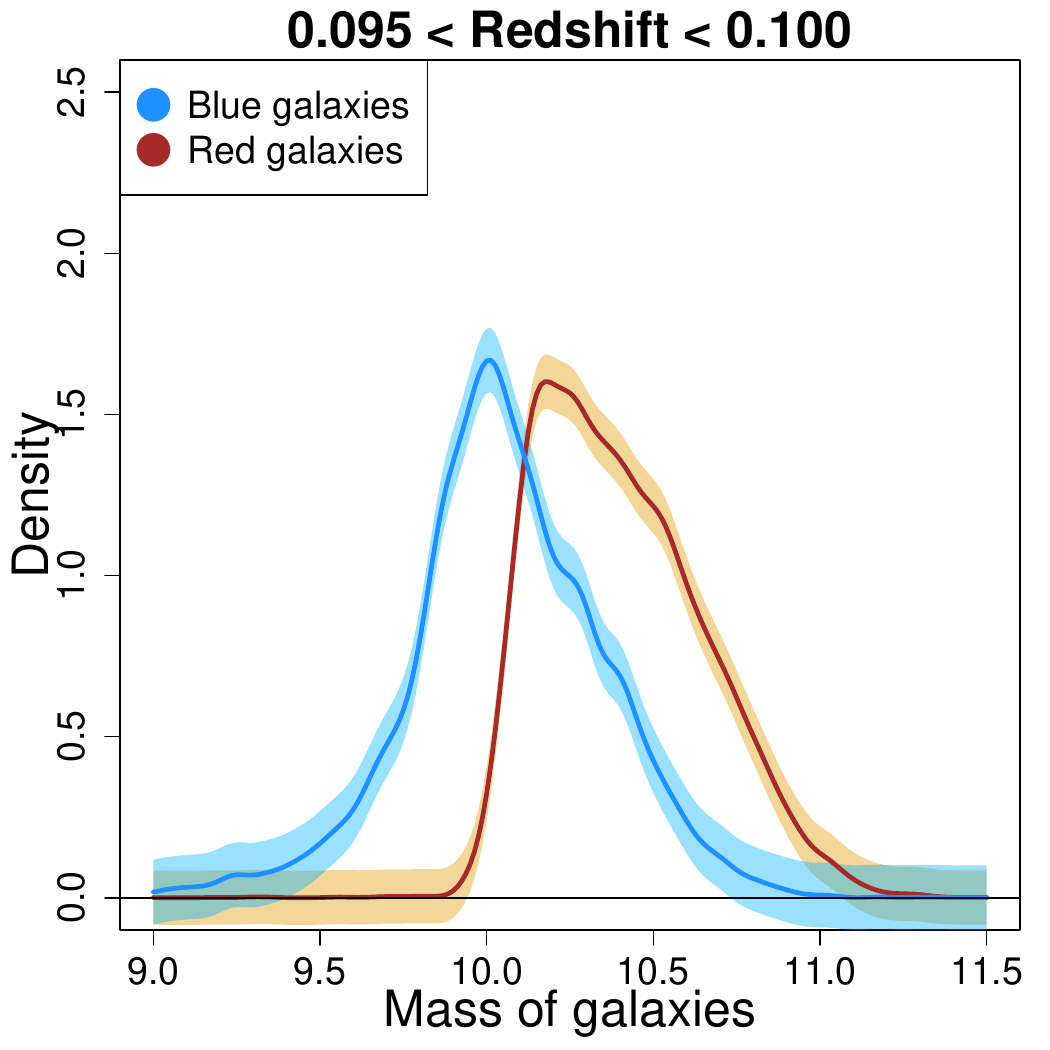}
\includegraphics[height=1.5in]{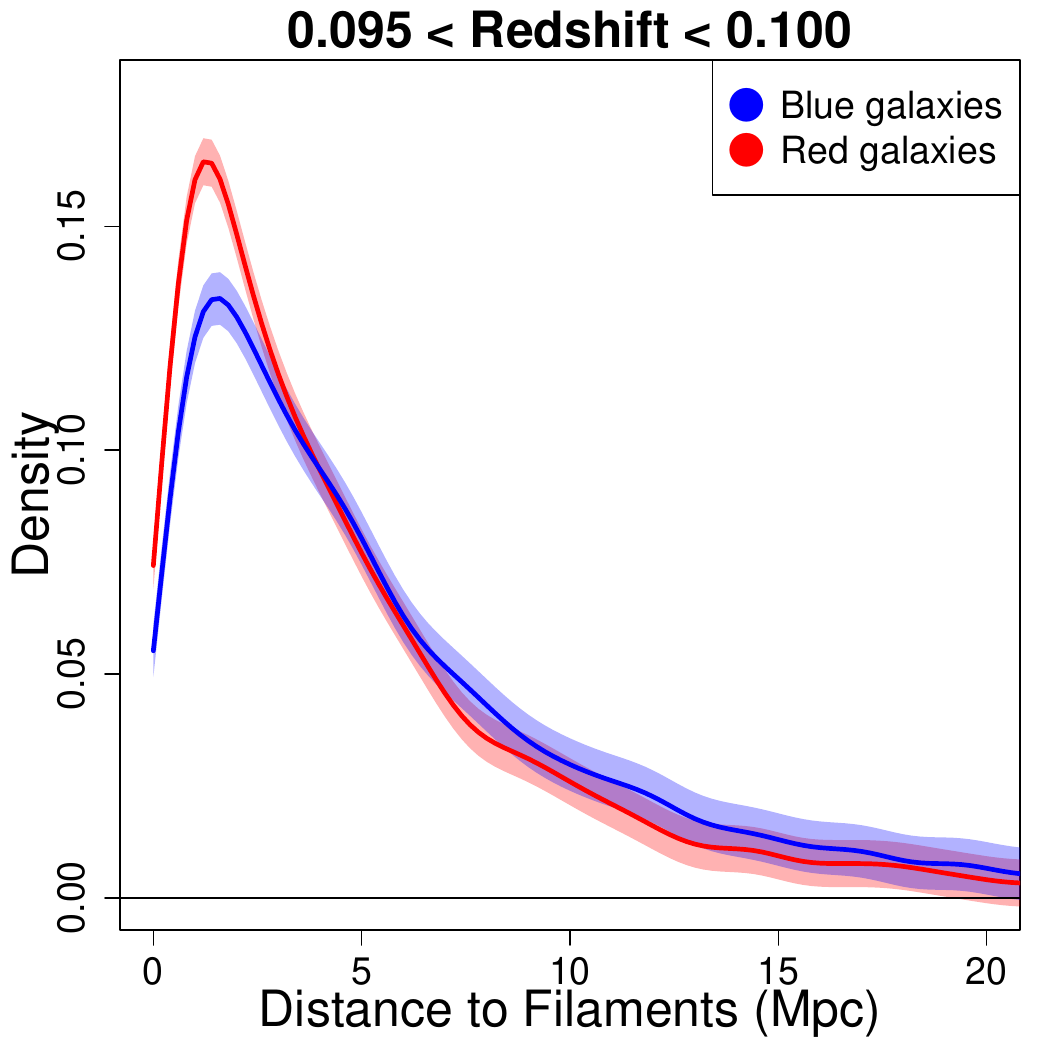}
\includegraphics[height=1.5in]{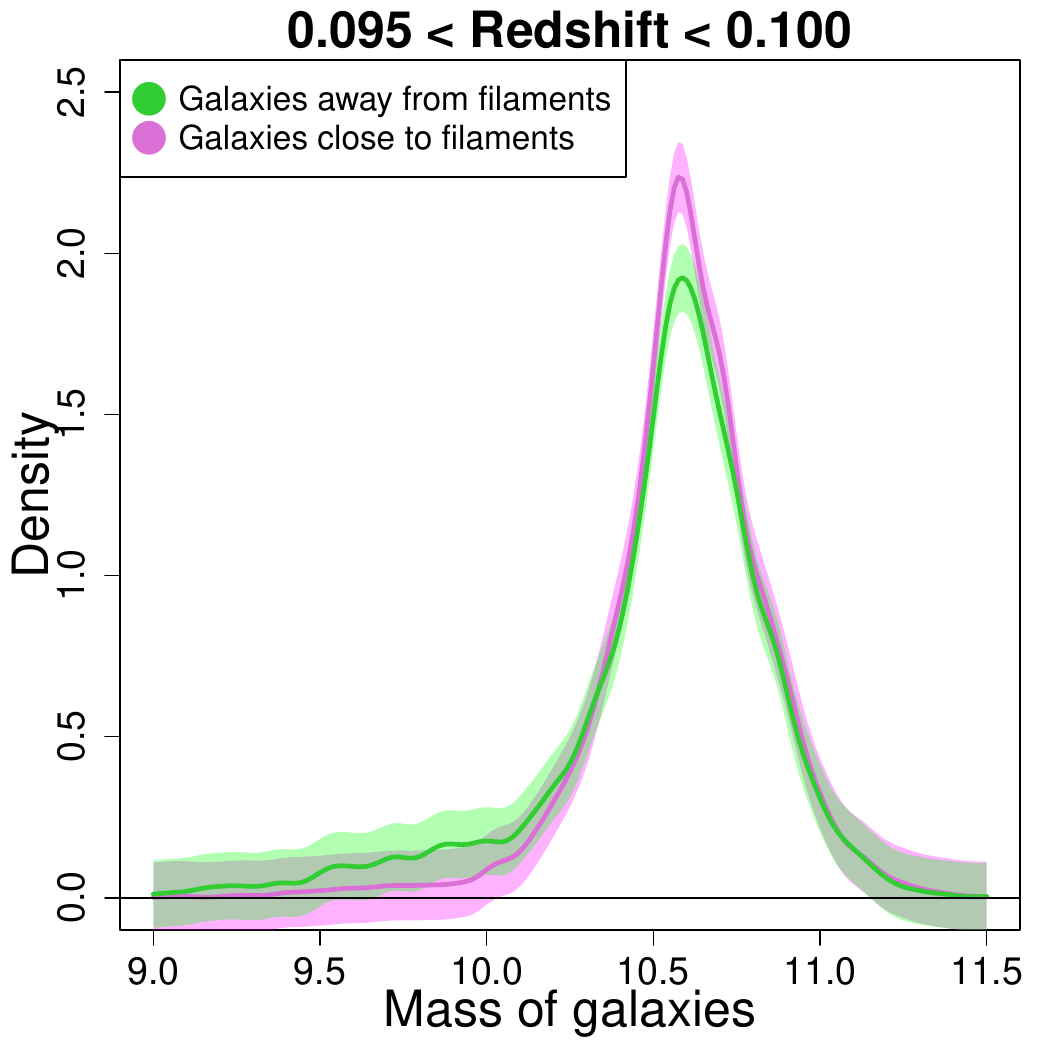}
\caption{Analyzing galaxies using the debiased KDE and the $95\%$ confidence band. 
We obtain galaxies from the Sloan Digital Sky Survey and focus on 
galaxies within a small region within our Universe ($0.095<$ {\sf redshift} $0.100$).
{\bf Left:} 
We separate galaxies by their colors and compare the densities of stellar mass distributions. 
We see a clear separation between the blue and the red galaxies.
{\bf Middle:}
Again we separate galaxies based on their color and compare their distance to the nearest filaments (curves characterizing
high density areas; \citealt{chen2015cosmic}). 
Red galaxies tend to concentrate more to regions around filaments than blue galaxies. 
{\bf Right:}
We separate galaxies baed on the median distance to the nearest filaments
and compare the stellar mass distribution in both groups. 
Although the difference is much smaller than other two panels,
we still observe a significant difference at the peak. 
Galaxies close to filaments tend to have a density that is highly concentrated around its peak
compared to those away from filaments.
}
\label{fig::astro}
\end{figure}

\subsection{Application in Astronomy}	\label{sec::astro}

To demonstrate the applicability, we apply our approach to the galaxy sample from
Sloan Digital Sky Survey (SDSS; \citealt{york2000sloan, eisenstein2011sdss}),
a well-known public Astronomy dataset that contains 1.2 million galaxies.
In particular, we obtain galaxies from the \emph{NYU VAGC}\footnote{\url{http://sdss.physics.nyu.edu/vagc/}} data
\citep{blanton2005new,padmanabhan2008improved,adelman2008sixth}, 
a galaxy catalog based on the SDSS sample.
We focus on five features of each galaxy:
{\sf RA} (right ascension--the longitude of the sky), {\sf dec} (declination--the latitude of the sky),
{\sf redshift} (distance to earth), {\sf mass} (stellar mass),
and {\sf color}\footnote{the color is based on the $(g-r)$ band magnitude; a galaxy is classified as a red galaxy 
if its $(g-r)$ band magnitude is greater than 0.8
otherwise it is classified as a blue galaxy \citep{chen2016detecting}. } (blue or red).

We select galaxies within a thin redshift slice $0.095<$ {\sf redshift} $<0.100$, 
which leads to a sample with size $n= 23,724$. 
We first examine the relationship between a galaxy's color versus its stellar mass.
Within this redshift region, there are $n_{red} = 13,910$ red galaxies
and $n_{blue}= 9,814$ blue galaxies.
We estimate the densities of the mass 
of both red and blue galaxies using the debiased KDE with the same smoothing bandwidth $h=0.046$ (chosen 
by the normal reference rule)
and apply the bootstrap $1000$ times to construct confidence bands.
The left panel of Figure~\ref{fig::astro} shows the two density estimators along with their $95\%$ confidence bands.
There is a clear separation between the stellar mass distribution of these two types of galaxies, which is affirmative to the 
literature in Astronomy \citep{tortora2010colour}.

Next we compare galaxies to cosmic filaments, curve-like structures characterizing high density regions 
of matter. 
We obtain filaments from the \emph{Cosmic Web Reconstruction} catalog\footnote{\url{https://sites.google.com/site/yenchicr/}},
a publicly available filament catalog \citep{chen2015cosmic, chen2016cosmic}.
Each filament is represented by a collection of spatial locations (in RA, dec, and redshift space).
Because we are using galaxies within a thin redshift slice, we select filaments within the same redshift region.
We then use the 2D spatial location ({\sf RA} and {\sf dec}) of galaxies to calculate their distance to the nearest filament
(we use the conventional unit of the distance in Astronomy: Mpc--megaparsec).
The distance to filament is a new variable of each galaxy.
Similar to the mass,
we estimate the densities of distance to the nearest filament
of both red and blue galaxies using the debiased KDE with $h=0.912$ (chosen 
by the normal reference rule)
and apply the bootstrap $1000$ times to construct confidence bands.
The center panel of Figure~\ref{fig::astro} displays the two density estimators and their $95\%$ confidence bands.
We see that most of the blue and red galaxies are within $10$ Mpc distance to filaments.
However, the density of red galaxies concentrates more at the low distance to filament region than the density of blue galaxies. 
The two confidence bands are separated, indicating that the difference is significant. 
This is also consistent with what is known in the Astronomy literature that
red galaxies tend to populate around high density areas (where most filaments live in)
compared to blue galaxies \citep{hogg2003overdensities, cowan2008environment}. 

Finally, we compare the mass distribution of galaxies at different distances to filaments. 
We separate galaxies into two groups, galaxies away from filaments and galaxies close to filaments,
using the median distance to the nearest filament.
We then estimate the densities of mass distribution
of both groups using the debiased KDE with $h=0.046$
and apply the bootstrap $1000$ times to construct confidence bands.
The right panel of Figure~\ref{fig::astro} displays the two density estimators and the $95\%$ confidence bands.
The two densities are close to each other but their are still significantly different--the mass distribution of galaxies
close to filaments concentrates more at its peak than the mass distribution of galaxies away from filaments. 
Judging from the shape of densities, galaxies close to filaments tend to be more massive
than those away from filaments. 
A similar pattern has been observed in the Astronomy literature as well \citep{grutzbauch2011does} and now we are using a different way of exploring
the difference between these two populations. 

%
%
%
%

\section{Discussion}	\label{sec::discuss}

In this paper, 
we propose to construct confidence bands/sets via bootstrapping the debiased estimators \citep{calonico2018effect}. 
We prove both theoretically and using simulations that our proposed confidence bands/sets
are asymptotically valid.
Moreover, our confidence bands/sets are compatible with many common bandwidth selectors,
such as the rule of thumb and cross-validation.

In what follows, we discuss some topics related to our methods.

\begin{itemize}
\item {\bf Higher-order kernels.} In this paper, we consider second-order kernels for simplicity.
Our methods can be generalized to higher-order kernel functions.
\cite{calonico2018effect} has already described the debiased estimator using higher-order kernel functions,
so to construct a confidence band, all we need to do is bootstrapping the $L_\infty$ error of the debiased estimator and take the quantile. 
Note that if we use a $\omega$-th order kernel function for the original KDE, 
then we can make inference for the functions in 
$$
\Sigma(\omega+\delta_0,L_0), \quad \delta_0>0 
$$
because the debiasing operation will kick the bias into the next order term. 
Thus, if we have some prior knowledge about the smoothness of functions we are interested in,
we can use a higher-order kernel function and bootstrap it to construct the confidence bands.

\item {\bf Detecting local difference of two functions. }
Our approaches can be used to detect local differences of two functions, which has been
used in Section~\ref{sec::astro}.
When the two functions being compared are densities, it is a problem for the local two sample test \citep{duong2009highest,duong2013local}.
When the two functions being compared are regression functions,
the comparison is related to the conditional average treatment effect curve
\citep{lee2009nonparametric, hsu2013consistent, ma2014treatment,abrevaya2015estimating}. 
In the local two sample test, we want to know
if two samples are from the same population or not and find out
the regions where the two densities differ. 
For the case of the conditional average treatment effect curve,
we compare the differences of two regression curves where
one curve is the regression curve from the control group and the other is the regression curve
from the treatment group.
The goal is to find out where we have strong evidence that the two curves differ.
In both cases, we can use the debiased estimators of the densities or regression functions,
and then bootstrap the difference to obtain an asymptotically valid confidence band.
\cite{chiang2017unified} has applied a similar idea to several local Wald estimators in econometrics.

\item {\bf Other geometric features.}
We can use the idea of bootstrapping the debiased estimator to make inferences
of other geometric features such as local modes \citep{romano1988bootstrapping}, ridges \citep{chen2015asymptotic}, and cluster trees \citep{jisu2016statistical}. 
\cite{romano1988bootstrapping} proved that naively bootstrapping the KDE
does not yield a valid confidence set 
unless we undersmooth the data.
However, bootstrapping the debiased KDE 
still works because
the optimal $h$ of the original KDE is 
an undersmoothed $h$ of the debiased KDE.
So our results are actually consistent with \cite{romano1988bootstrapping}.

\end{itemize}




\section*{Acknowledgement}	
YC is supported by NSF grant DMS 1807392. 
We thank Cun-Hui Zhang for useful and constructive suggestions about this paper.
{We thank two referees for very helpful comments on this paper. }

\appendix

\section{Proofs of the kernel density estimator}

\begin{proof}[ of Lemma~\ref{lem::KDE_BV}]

{\bf Bias.}
Recall from equation \eqref{eq::dKDE}
$$
\hat{p}_{\tau,h}(x) = \hat{p}_h(x) - \frac{1}{2}c_K\cdot h^2\cdot \hat{p}^{(2)}_b(x).
$$
Thus, by the standard derivation of the bias under assumption (P),
\begin{align*}
\E\left(\hat{p}_{\tau,h}(x)\right) &= \E\left( \hat{p}_h(x)\right) - \frac{1}{2}c_K\cdot h^2\cdot\E\left(\hat{p}^{(2)}_b(x)\right)\\
&= p (x) + \frac{1}{2}c_K\cdot h^2\cdot p^{(2)}(x) + O(h^{2+\delta_0}) -\frac{1}{2}c_K\cdot h^2\cdot (p^{(2)}(x) + O(b^{\delta_0}))\\
& = p(x) + O(h^{2+\delta_0}+h^2\cdot b^{\delta_0}).
\end{align*}
Because $\tau = h/b$ is fixed, we obtain the desired result for the bias.

{\bf Variance.}
To derive the variance, note that under (K1) and (P)
and $\tau = h/b$ is fixed,
\begin{align*}
{\sf Var}(\hat{p}_{\tau, h}(x)) & = \frac{1}{n h^{2d}} \left\{ \E\left[ M_{\tau}\left(\frac{x - X_i}{h}\right)^2\right] - \E\left[ M_{\tau} \left(\frac{x - X_i}{h} \right)\right]^2 \right\} \\
& = \frac{1}{n h^{2d}} \left\{ h^d \int p(x - th) M_{\tau}^2(t) dt - \left(h^d \int p(x - th) M_{\tau}(t) dt \right)^2 \right\} \\
& = \frac{1}{nh^d} \left( p(x) \int M_{\tau}^2(t) dt + O(h^2) + O(h^d) \right)
\end{align*}
Since we have that $\int t_i M_{\tau}^2(t) dt_i = 0$ for $i = 1, \ldots, d$. 
\end{proof}

\begin{proof}[ of Lemma~\ref{lem::KDE_infty}]
By Lemma~\ref{lem::KDE_BV},
\begin{equation}
\hat{p}_{\tau,h}(x) - p(x) = \hat{p}_{\tau,h}(x) - \E\left(\hat{p}_{\tau,h}(x)\right) + O(h^{2+\delta_0})
\label{eq::pf::kde_infty1}
\end{equation}
and when $\frac{ nh^{d+4}}{\log n}\rightarrow c_0 \geq 0$, the bias is negligible compared to $\hat{p}_{\tau,h}(x) - \E\left(\hat{p}_{\tau,h}(x)\right)$
so we only focus on the stochastic variation part.
To derive the rate of $\hat{p}_{\tau,h}(x) - \E\left(\hat{p}_{\tau,h}(x)\right)$,
note that 
$\hat{p}_{\tau,h}(x)$ is a KDE with the kernel function $M_\tau\left(\frac{x-y}{h}\right)$. 


Because assumption (K2) implies that for fixed $\tau$ and $\overline{h}$, 
$$
\cF_1 = \left\{g_x(y) = M_\tau\left(\frac{x-y}{h}\right): x\in \K, \overline{h}\geq h>0\right\}
$$
is a bounded VC class of functions.
Note that we can always find such a $\overline{h}$ because $h\rightarrow 0$ when $n\rightarrow \infty$.
Therefore, $\cF_1$ satisfies the $K_1$ condition of \cite{Gine2002},
which implies that 
$$
\sup_{x\in \K}\|\hat{p}_{\tau,h}(x)-\E\left(\hat{p}_{\tau,h}(x)\right) \| = O_P\left(\sqrt{\frac{\log n}{nh^d}}\right).
$$

Plugging this into equation \eqref{eq::pf::kde_infty1} and notice that the constant factor in $O(h^{2+\delta_0})$ is bounded by H\"older Class constant $L_0$. We obtain the desired result.
$$
\| \hat{p}_{\tau, h} - p\|_{\infty} = O(h^{2 + \delta_0}) + O_P\left(\sqrt{\frac{\log n}{n h^d}} \right)
$$
\end{proof}

\begin{proof}[ of Theorem~\ref{thm::KDE_gaussian}]
By Lemma~\ref{lem::KDE_infty},
when $\frac{ nh^{d+4}}{\log n}\rightarrow c_0 \geq 0$ bounded, the scaled difference 
\begin{equation}
\sqrt{nh^d} \|\hat{p}_{\tau,h}-p\|_\infty = \sqrt{nh^d} \|\hat{p}_{\tau,h}-\E\left(\hat{p}_{\tau,h}\right)\|_\infty + O(\sqrt{nh^{d+4 + 2\delta_0}}).
\label{eq::KDE_G1}
\end{equation}
By Corollary 2.2 and the derivation of Proposition 3.1 in \cite{chernozhukov2014gaussian},
there is a tight Gaussian process $\B_n$ as described in Theorem~\ref{thm::KDE_gaussian}
and constants $A_1, A_2>0$
such that for any $\gamma>0$,
\begin{equation*}
P\left(\left|\sqrt{nh^d} \|\hat{p}_{\tau,h}-\E\left(\hat{p}_{\tau,h}\right)\|_\infty- \sup_{f\in\cF_{\tau,h}}\|\B_n(f)\|\right|>\frac{A_1 \log^{2/3} n}{\gamma^{1/3}(nh^d)^{1/6}}\right)\leq
A_2\gamma
\end{equation*}
when $n$ is sufficiently large.
Using equation \eqref{eq::KDE_G1} and $\delta > 2/3$, we can revise the above inequality by
\begin{equation}
P\left(\left|\sqrt{nh^d} \|\hat{p}_{\tau,h}-p\|_\infty- \sup_{f\in\cF_{\tau,h}}\|\B_n(f)\|\right|>\frac{A_3 \log^{2/3} n}{\gamma^{1/3}(nh^d)^{1/6}}\right)\leq
A_2\gamma
\label{eq::KDE_G2}
\end{equation}
for some constants $A_3$. 

To convert the bound in equation \eqref{eq::KDE_G2} into a bound on the Kolmogorov distance,
we apply the anti-concentration inequality (Lemma~2.3 in \citealt{chernozhukov2014gaussian}; see also \citealt{chernozhukov2014anti}),
which implies that 
when $n$ is sufficiently large, there exists a constant $A_4>0$ such that 
\begin{equation}
\begin{aligned}
\sup_{t}\Bigg|P\Bigg(\sqrt{nh^d} \|\hat{p}_{\tau,h}-p\|_\infty <t\Bigg)&-P\Bigg(\sup_{f\in\cF_{\tau,h}}\|\B_n(f)\|<t\Bigg)\Bigg|\\
&\leq A_4 \cdot \E\left(\sup_{f\in\cF_{\tau,h}}\|\B_n(f)\|\right) \cdot \frac{A_3 \log^{2/3} n}{\gamma^{1/3}(nh^d)^{1/6}} + A_2\gamma.
\end{aligned}
\end{equation}
By DudleyÕs inequality of Gaussian processes \citep{Vaart1996}, 
$$
\E\left(\sup_{f\in\cF_{\tau,h}}\|\B_n(f)\|\right) = O(\sqrt{\log n}),
$$
so the optimal $\gamma = \left(\frac{\log^7 n}{nh^d}\right)^{1/8}$, which leads to the desired result:
\begin{equation*}
\begin{aligned}
\sup_{t}\Bigg|P\Bigg(\sqrt{nh^d} \|\hat{p}_{\tau,h}-p\|_\infty <t\Bigg)-P\Bigg(\sup_{f\in\cF_{\tau,h}}\|\B_n(f)\|<t\Bigg)\Bigg|
\leq  O\left( \left(\frac{\log^7 n}{nh^d}\right)^{1/8}\right).
\end{aligned}
\end{equation*}
\end{proof}

\begin{proof}[ of Theorem~\ref{thm::KDE_CI}]
The proof of Theorem~\ref{thm::KDE_CI} follows the same derivation as the proof of Theorem 4 of \cite{chen2017density}.
A similar derivation also appears in \cite{chernozhukov2014anti}.
Here we only give a high-level derivation. 

Let $t_{1-\alpha}$ be the $1-\alpha$ quantile of the CDF of $\|\hat{p}_{\tau,h}-p\|_\infty$.
By the property of the $L_\infty$ loss $\|\hat{p}_{\tau,h}-p\|_\infty$, it is easy to see that
$$
P(p(x) \in [\hat{p}_{\tau,h}(x) - t_{1-\alpha},\,\, \hat{p}_{\tau,h}(x) + t_{1-\alpha}]\,\, \forall x\in\K) = 1-\alpha.
$$
Thus, all we need to do is to prove that the bootstrap estimate $\hat{t}_{1-\alpha}$ approximates $t_{1-\alpha}$. 
We will prove this by showing that 
$\sqrt{nh^d} \|\hat{p}_{\tau,h}-p\|_\infty$ and the bootstrap $L_\infty$ metric $\sqrt{nh^d} \|\hat{p}^*_{\tau,h}-\hat{p}_{\tau,h}\|_\infty$
converges in the Kolmogorov distance (i.e., the Berry-Esseen bound).

By Theorem~\ref{thm::KDE_gaussian}, we known that there exists a Gaussian process $\B_n$ defined on $\cF_{\tau,h}$ such that
$$
\sqrt{nh^d} \|\hat{p}_{\tau,h}-p\|_\infty \approx \sup_{f\in\cF_{\tau,h}}|\B_n(f)|. 
$$
Conditioned on $\mathcal{X}_n=\{X_1,\cdots, X_n\}$, 
the bootstrap difference 
$$
\sqrt{nh^d}\|\hat{p}^*_{\tau,h}-\hat{p}_{\tau,h}\|_\infty = \sqrt{nh^d}\|\hat{p}^*_{\tau,h}-\E\left(\hat{p}^\star_{\tau,h}|\mathcal{X}_n\right)\|_\infty
$$
and similar to $\sqrt{nh^d} \|\hat{p}_{\tau,h}-p\|_\infty$, 
$\sqrt{nh^d}\|\hat{p}^*_{\tau,h}-\E\left(\hat{p}^*_{\tau,h}|\mathcal{X}_n\right)\|_\infty$
can be approximated by the maximum of an empirical bootstrap process \citep{chernozhukov2016empirical},
which, by the same derivation as the proof of Theorem~\ref{thm::KDE_gaussian}, leads to
the following conclusion
\begin{align*}
\sup_{t}\Bigg|P\Bigg(\sqrt{nh^d} \|\hat{p}^*_{\tau,h}-\hat{p}_{\tau,h}\|_\infty <t\Bigg|\mathcal{X}_n\Bigg)-  & P\Bigg(\sup_{f\in\cF_{\tau,h}}\|\tilde{\B}_n(f)\|<t\Bigg|\mathcal{X}_n\Bigg)\Bigg|  \\
& \leq  O_P\left( \left(\frac{\log^7 n}{nh^d}\right)^{1/8}\right),
\end{align*}
where $\tilde{\B}_n$ is a Gaussian process defined on $\cF_{\tau,h}$
such that for any $f_1,f_2\in\cF_{\tau,h}$
\begin{align*}
\E\left(\tilde{\B}_n(f_1)\tilde{\B}_n(f_2)|\mathcal{X}_n\right) &= \hat{{\sf Cov}}(f_1(X), f_2(X))\\
& = \frac{1}{n}\sum_{i=1}^nf_1(X_i)f_2(X_i) - \frac{1}{n}\sum_{i=1}^nf_1(X_i)\cdot\frac{1}{n}\sum_{i=1}^nf_2(X_i).
\end{align*}
Namely, $\E\left(\tilde{\B}_n(f_1)\tilde{\B}_n(f_2)|\mathcal{X}_n\right)$ follows
the sample covariance structure at $f_1$ and $f_2$.

Because 
$\tilde{\B}_n$ and $\B_n$ differ only in the covariance structure and the sample covariance 
converges to the population covariance,
by the Gaussian comparison Lemma (Theorem 2 in \citealt{chernozhukov2014comparison}), 
$\sup_{f\in\cF_{\tau,h}}\|\tilde{\B}_n(f)\|$ and $\sup_{f\in\cF_{\tau,h}}\|\B_n(f)\|$ converges
in the Kolmogorov distance (and the convergence rate is faster than the Gaussian approximation described in Theorem~\ref{thm::KDE_gaussian}
so we can ignore the error here). 

Thus, we have shown that 
$$
\sqrt{nh^d} \|\hat{p}_{\tau,h}-p\|_\infty \approx \sup_{f\in\cF_{\tau,h}}|\B_n(f)|\approx \sup_{f\in\cF_{\tau,h}}|\tilde{\B}_n(f)|\approx \sqrt{nh^d} \|\hat{p}^*_{\tau,h}-\hat{p}_{\tau,h}\|_\infty,
$$
which proves that 
\begin{align*}
\sup_{t}\Bigg|P\Bigg(\sqrt{nh^d} \|\hat{p}_{\tau,h}-p\|_\infty <t\Bigg)-& P\Bigg(\sqrt{nh^d} \|\hat{p}^*_{\tau,h}-\hat{p}_{\tau,h}\|_\infty <t\Bigg|\mathcal{X}_n\Bigg)\Bigg| \\
& \leq  O_P\left( \left(\frac{\log^7 n}{nh^d}\right)^{1/8}\right).
\end{align*}

Thus, the quantile of the distribution of $\|\hat{p}^*_{\tau,h}-\hat{p}_{\tau,h}\|_\infty $ approximates the quantile of 
the distribution of $\|\hat{p}_{\tau,h}-p\|_\infty$, which proves the desired result.
Note that although the rate is written in $O_P$,
the quantity inside the $O_P$ part has an expectation of the same rate
so we can convert it into an unconditional bound (the second CDF is not conditioned on $\mathcal{X}_n$) with the same rate.

\end{proof}

%

\section{Justification for remark \ref{rm::honest0}}	\label{sec::RM2}
{
In this section, we provide the justification for the argument in remark \ref{rm::honest0}.  
Let $\partial \mathbb{K} $ be the boundary of $\mathbb{K}$
and define the distance from a point $x$ to a set $A$ as $d(x,A) = \inf_{y\in A}\|x-y\|$.
Let $\mathbb{K}_{\delta} = \{x\in \mathbb{K}: d(x,\partial \mathbb{K}) \geq \delta\}$
be the region within $\mathbb{K}$ that are at least $\delta$ distance from the boundary. 
In lemma \ref{lem::KDE_BV}, we proved that 
\begin{align*}
{\sf Var}(\hat{p}_{\tau, h}(x)) & = \frac{1}{nh^d} \left( p(x) \int M_{\tau}^2(t) dt + O(h^2) + O(h^d) \right)
\end{align*}
which turn implies that $\sigma^2_{rbc} = (nh^d) {\sf Var}(\hat{p}_{\tau, h}(x)) = p(x) \int M_{\tau}^2(t) dt + O(h^2) + O(h^d)$.  For any $x \in \K$, 
\begin{align*}
\sigma_{rbc}(x) &= \sqrt{nh^d {\sf Var}(\hat{p}_{\tau, h}(x))} = \sqrt{p(x) \int M_{\tau}^2(t) dt + O(h^2) + O(h^d)} \\
& \geq \frac{1}{2}\sqrt{ p(x) \int M_{\tau}^2(t) dt }  \qquad \text{for $n$ sufficiently large}
\end{align*}
assuming that $p(x)$ bounded away from 0 for $x \in \K$.   Then we could easily get a constant lower bound on $\sigma_{rbc}(x)$. One thing to note here is that recall in our assumption (P),  we specifically assume that $p(x)$ and $\nabla p(x)$ is 0 on the boundary of $\K$. This is for the purpose of bias estimation and we could simply restrict the uniform confidence band to be covering a slightly smaller inner set for $\mathbb{K}$ like $\mathbb{K}_{\delta}$ defined above. Since $\sigma_{rbc}(x)$ is lower bounded on $\mathbb{K}_{\delta}$, we obtained that
$$
\left\| \frac{\hat{p}_{\tau, h} - p}{\sigma_{rbc}} \right\|_{\mathbb{K}_{\delta}} = O(h^{2 + \delta_0}) + O_P \left( \sqrt{\frac{\log n}{nh^d}} \right)
$$
where $\| f \|_{\mathbb{K}_{\delta}} = \sup_{x \in \mathbb{K}_{\delta}} | f(x)|$. The rescaled difference
can be written as 
$$
\sqrt{nh^d} \left( \frac{ \hat{p}_{\tau, h}(x) - p(x)}{\sigma_{rbc}(x)} \right) = \mathbb{G}_n (f_x^{rbc}) + O(\sqrt{nh^{d + 4 + 2\delta_0}}),
$$
where $f_x^{rbc}(y) \in {\cal F}^{rbc}_{\tau, h}$ and
$$
{\cal F}^{rbc}_{\tau,h} = \left\{f_x(y) = \frac{M_\tau\left(\frac{x-y}{h}\right)}{\sqrt{h^d} \sigma_{rbc}(x)}: x \in \K_\delta \right\}.
$$
As a result,
the Gaussian approximation method could be applied to the new function class ${\cal F}_{\tau, h}^{rbc}$.  
When we use the sample studentized quantity, it can be decomposed as
\begin{align*}
\sqrt{nh^d} \left(\frac{\hat{p}_{\tau, h}(x) - p(x)}{\hat{\sigma}_{rbc}(x)} \right) & = \sqrt{nh^d} \left(\frac{\hat{p}_{\tau, h}(x) - p(x)}{\sigma_{rbc}(x)} \right) + \sqrt{nh^d} \left( \frac{(\hat{p}_{\tau, h}(x) - p(x))(\sigma_{rbc}(x) - \hat{\sigma}_{rbc}(x))}{\hat{\sigma}_{rbc}(x) \sigma_{rbc}(x)} \right).
\end{align*}
Since the Gaussian approximation can be applied to the first term in the right hand side,
we only need to show that the second term is of the rate $o_P(1)$ (so it is negligible).
Noting that
\begin{align*}
\hat{\sigma}^2_{rbc} - \sigma^2_{rbc} & = \frac{1}{h^d}\left[\frac{1}{n}\sum_{i=1}^n M_{\tau}^2 \left( \frac{x - X_i}{h} \right) - \left(\frac{1}{n} \sum_{i=1}^n M_{\tau} \left(\frac{x - X_i}{h} \right)\right)^2 \right]  - \\
& \frac{1}{h^d} \left\{ \E\left[ M_{\tau}^2\left(\frac{x - X_i}{h}\right)\right] - \E\left[ M_{\tau} \left(\frac{x - X_i}{h} \right)\right]^2 \right\} \\
& = \left\{\frac{1}{n h^d} \sum_{i=1}^n M_{\tau}^2 \left(\frac{x - X_i}{h} \right) - \frac{1}{h^d} \E\left[M_{\tau}^2 \left(\frac{ x - X_i}{h} \right) \right] \right\} - \\
& h^d \left\{ \left(\frac{1}{nh^d} \sum_{i=1}^n M_{\tau} \left(\frac{x - X_i}{h} \right) \right)^2  - \E\left[ \frac{1}{h^d}M_{\tau} \left(\frac{ x - X_i}{h} \right) \right]^2 \right\}
\end{align*}
Since assumption (K2) implies that for fixed $\tau$ and $\overline{h}$, 
$$
\cF_2 = \left\{g_x(y) = M_\tau^2\left(\frac{x-y}{h}\right): x\in \K, \overline{h}\geq h>0\right\}
$$
is a bounded VC class of functions. The main result in \cite{Gine2002} implies that
$$
\|\hat{\sigma}^2_{rbc} - \sigma^2_{rbc}\|_\infty = O_P\left( \sqrt{\frac{\log n}{nh^d}}\right).
$$
This in turn implies that $\|\hat{\sigma}_{rbc}  - \sigma_{rbc}\|_\infty = O_P\left( \sqrt{\frac{\log n}{nh^d}}\right)$. As a result,
\begin{align*}
\sqrt{nh^d} \left\| \frac{(\hat{p}_{\tau, h} - p)(\sigma_{rbc} - \hat{\sigma}_{rbc})}{\hat{\sigma}_{rbc} \sigma_{rbc}} \right\|_\infty & = O_P\left(\sqrt{\log n}\cdot \|\hat{p}_{\tau, h} - p \|_\infty\right) = o_P(1),
\end{align*}
which shows that bootstrapping the studentized quantity also works.
}

\section{Proofs of the local polynomial regression}

To simplify our derivation, we denote
a random variable $Z_n = O_r(a_n)$ if $\E|Z_n|^r = O(a_n^r)$ for an integer $r$.   
Then it is obvious that 
$$
Z_n = \E Z_n + O_r( ( \E|Z_n - \E Z_n|^r)^{1/r}).
$$
Note that by Markov's inequality, $O_r(a_n)$ implies $O_P(a_n)$ for any sequence $a_n\rightarrow 0$.

Before we prove lemma ~\ref{lem::LS_BV}, we first derive the bias rate for $\hat{r}_h(x)$ and $\hat{r}^{(2)}_{h /\tau}(x)$.  The proof basically follows from \cite{fan1993local, fan1996local} and for completeness, we put it here. 
\begin{lem}[Rates for function and derivative estimation]
Assume (K3), (R1), and $\tau \in (0,\infty)$ is fixed,  then the bias of $\hat{r}_h$ and $\hat{r}^{(2)}_{h /\tau}$ for a given point $x_0$ is at rate
$$
\E(\hat{r}_h(x_0)) - r(x_0) = \frac{h^2}{2} c_K \cdot r^{(2)}(x_0) + O(h^{2 + \delta_0}) + h^2 O\left(\sqrt{\frac{1}{nh}}\right)
$$

$$
\E(\hat{r}^{(2)}_{h /\tau}(x_0)) - r^{(2)}(x_0) = O(h^{\delta_0}) + O\left(\sqrt{\frac{1}{nh}}\right)
$$

if $h = O(n^{-1/5})$, then 
$$
\E(\hat{r}_h(x_0)) - r(x_0) = \frac{h^2}{2} c_K \cdot r^{(2)}(x_0) + O(h^{2 + \delta_0}) 
$$
$$
\E(\hat{r}^{(2)}_{h /\tau}(x_0)) - r^{(2)}(x_0) = O(h^{\delta_0})
$$
\label{lem::Rate}
\end{lem}
\begin{proof}[ of Lemma ~\ref{lem::Rate}]

{\bf Bias of regression function.}
Using the notation from Section 2 and first conditioning on the covariates $X_1,\cdots,X_n$, the bias could be written as
$$
\E(\hat{r}_h(x_0)) - r(x_0)) = \E\left[\frac{\sum_{i=1}^n [r(X_i) - r(x_0)] w_{i,h}(x_0)}{\sum_{i=1}^n w_{i,h}(x_0)}\right],
$$
where
$$
\sum_{i=1}^n w_{i,h}(x_0) = S_{n,h,2}(x_0) S_{n,0,h}(x_0) - (S_{n,1,h}(x_0))^2.
$$

%
%
Given this notation, the following equation holds: 
\begin{align*}
\frac{1}{n h^{j+1}} S_{n,h,j}(x_0) & = \frac{1}{n h^{j+1}} \E S_{n,h,j}(x_0) + O_2\left(\sqrt{\frac{1}{nh}}\right) \\
& = p_X(x_0) s_j + s_{j+1} O(h) + s_{j+2} O(h^2) + O_2\left(\sqrt{\frac{1}{nh}}\right)
\end{align*}
for $j = 0,1,2$, where $s_j = \int^{\infty}_{-\infty} u^j K(u) du$ and 
$s_0 = 1$, $s_1 = 0$,
Moreover, 
\begin{equation}
\begin{aligned}
\sum_i w_{i,h}(x_0) & = S_{n,h,2}(x_0) S_{n,h,0}(x_0) - (S_{n,1,h}(x_0))^2 \\
& = n^2 h^4 s_2 p_X^2(x_0) \left(1 + O(h^2) + O_2\left(\sqrt{\frac{1}{nh}}\right)\right)
\end{aligned}
\label{eq::LPR::deno}
\end{equation}
Next, for the numerator,  let $R(X_i) = r(X_i) - r(x_0) - r'(x)(X_i - x_0)$ and using the property that $\sum_i w_{i,h}(x_0) (X_i - x_0) = 0$ given that this is a local linear smoother,
then we have
\begin{equation}
\begin{aligned}
\sum_{i=1}^n &[r(X_i) - r(x_0)] w_{i,h}(x_0) = \sum_{i=1}^n R(X_i)w_{i,h}(x_0) \\
& = \sum_i R(X_i) K\left(\frac{X_i - x_0}{h}\right) S_{n,h,2}(x_0) \\
& - \sum_i R(X_i) K\left(\frac{X_i - x_0}{h}\right) (X_i - x_0) S_{n,h,1}(x_0) 
\end{aligned}
\label{eq::LPR::nume}
\end{equation}
Here, for $j = 0, 1$, 
\begin{equation}
\begin{aligned}
\frac{1}{n h^{3+j}} & \sum_i R(X_i) (X_i - x_0)^j K\left(\frac{x_0 - X_i}{h}\right)  =\\
&  h^{- 3 - j} \E\left( [r(X) - r(x_0) - r'(x_0)(X - x_0)] (X - x_0)^j K\left(\frac{X - x_0}{h}\right)\right) + \\ & O_2\left(\sqrt{\frac{1}{nh}}\right).
\end{aligned}
\end{equation}
Note that by the mean value theorem,
\begin{align*}
& h^{- 3 - j} \E\left( [r(X) - r(x_0) - r'(x_0)(X - x_0)] (X - x_0)^j K\left(\frac{X - x_0}{h}\right)\right) \\
& = h^{-3 - j} \E\left( \frac{1}{2} r^{(2)}(X^\star) (X - x_0)^{j+2} K\left(\frac{X - x_0}{h}\right)\right) \\
\end{align*}
for some $X^{\star} \in [x_0-X, x_0+X]$.
%
Thus, by taking expectation with respect to $X$ (we denote it by $t$ and change the corresponding $X^{\star}$ as $t^{\star}$), 
\begin{align*}
& h^{-3 - j} \E\left(\frac{1}{2} r^{(2)}(X^{\star}) (X - x_0)^{j+2} K\left(\frac{X- x_0}{h}\right)\right)  \\
& = h^{-3 - j}\int \frac{1}{2}[ r^{(2)}(x_0) + r^{(2)}(t^{\star}) - r^{(2)}(x_0)] (t - x_0)^{j+2} K\left(\frac{t - x_0}{h}\right) p_X(t) dt \\
& = \frac{1}{2} r^{(2)}(x_0) \int u^{j + 2} K(u) p_X(x_0 - uh) du \\ 
& + h^{-3 - j} \int \frac{1}{2}[r^{(2)}(t^{\star}) - r^{(2)}(x_0)] (t - x_0)^{j + 2} K\left(\frac{t - x_0}{h}\right) p_X(t) dt \\
& = \frac{1}{2} r^{(2)}(x_0)s_{j+2}p_X(x_0)  + O(h^{2-j}) + O(h^{\delta_0}).
\end{align*}
The $O(h^{2-j})$ comes from the fact that $s_3 = 0$ and $s_4 \neq 0$ since $K(x)$ is a even function.
The $O(h^{\delta_0})$ comes from applying
a variable transformation for the second part $t - x_0 = uh$ and using the fact that $|t^{\star} - x_0 | \leq |t - x_0| = |uh|$ and the H\"older condition (Assumption (P))
to the second part.

Thus, pluging in the results above to \eqref{eq::LPR::nume}, 
\begin{align*}
& \sum_{i=1}^n [r(X_i) - r(x_0)] w_{i,h}(x_0)  \\
& = \sum_{i=1}^n R(X_i) w_{i,h}(x_0) \\
& = \sum_i R(X_i) K\left(\frac{X_i - x_0}{h}\right) S_{n,h,2}(x_0) - \sum_i R(X_i) K\left(\frac{X_i - x_0}{h}\right) (X_i - x_0) S_{n,h,1}(x_0) \\
& = n^2 h^6\left[\frac{1}{2} r^{(2)}(x_0)s_2 p_X(x_0) + O(h^{\delta_0}) + O(h^2) + O_2\left(\sqrt{\frac{1}{nh}}\right)\right] \\
& \left[p_X(x_0) s_2 + O(h^2) +  O_2\left(\sqrt{\frac{1}{nh}}\right)\right] \\
 & - n^2 h^6\left[O(h) + O(h^{\delta_0}) + O_2\left(\sqrt{\frac{1}{nh}}\right)\right]\left[O(h) + O_2\left(\sqrt{\frac{1}{nh}}\right)\right] \\
 & = n^2 h^6 \left[\frac{1}{2} r^{(2)}(x_0) s_2^2 p_X^2(x_0) + O(h^{\delta_0}) + O_2\left(\sqrt{\frac{1}{nh}}\right)\right]
\end{align*}
Then using this result combined with \eqref{eq::LPR::deno}
\begin{align*}
\E(\hat{r}_h(x_0)) - r(x_0)) & = \E\left[\frac{\sum_{i=1}^n [r(X_i) - r(x_0)] w_{i,h}(x_0)}{\sum_{i=1}^n w_{i,h}(x_0)}\right] \\
& = \frac{1}{2} s_2 h^2 r^{(2)}(x_0) + O(h^{2 + \delta_0}) + h^2 O\left(\sqrt{\frac{1}{nh}}\right)
\end{align*}
where given $h = O(n^{-1/5})$ and $\delta_0 \leq 2$,  the bias would be 
$$
\E(\hat{r}_h(x_0)) - r(x_0)) = \frac{1}{2} s_2 h^2 r^{(2)}(x_0) + O(h^{2 + \delta_0})
$$

{\bf Bias of derivative.}
For the rates of the bias of second order derivative,  here we use the notation defined before Lemma ~\ref{lem::LPR::inv}, 
\begin{equation}
\begin{aligned}
\E(\hat{r}^{(2)}_{h/\tau}(x_0)) - r^{(2)}(x_0) = 2\left[\E\left(e_3^T\left(X_{x_0}^T W_{x_0} X_{x_0}\right)^{-1} X_{x_0}^T W_{x_0} \Y - \frac{1}{2}r^{(2)}(x_0)\right)\right]
\end{aligned}
\label{eq::LPR::sd}
\end{equation}
and again we use the property of third order local polynomial regression such that 
$$
e_3^T\left(X_{x_0}^T W_{x_0} X_{x_0}\right)^{-1} X_{x_0}^T W_{x_0}\mathbbm{1}_n r(x_0) = 0 
$$
and 
$$
e_3^T\left(X_{x_0}^T W_{x_0} X_{x_0}\right)^{-1} X_{x_0}^T W_{x_0} \left(\begin{array}{c} X_1 - x_0 \\ X_2 - x_0 \\ \cdots \\ X_n - x_0 \end{array}\right) r'(x_0) = 0
$$
and 
$$
e_3^T\left(X_{x_0}^T W_{x_0} X_{x_0}\right)^{-1} X_{x_0}^T W_{x_0} \left(\begin{array}{c} (X_1 - x_0)^2 \\ (X_2 - x_0)^2 \\ \cdots \\ (X_n - x_0)^2 \end{array}\right) \frac{r^{(2)}(x_0)}{2} = \frac{1}{2} r^{(2)}(x_0)
$$

Let $R(X_i) = r(X_i) - r(x_0) - r'(x_0)(X_i - x_0) - \frac{1}{2}r^{(2)}(x_0)(X_i - x_0)^2$.
Then conditioning on $X_1,\ldots, X_n$, the right hand side of \eqref{eq::LPR::sd} becomes (up to a constant factor of 2)
\begin{align*}
\E\left(e_3^T(X_{x_0}^T W_{x_0} X_{x_0})^{-1} X_{x_0}^T W_{x_0} \left(\begin{array}{c} R(X_1) \\ R(X_2) \\ \cdots \\ R(X_n) \end{array}\right) \right)
\end{align*}
Using the notations in Lemma ~\ref{lem::LPR::inv},
the expectation can be rewritten as 
$$
\frac{1}{b^2}e_3^T (\frac{1}{nb}\X_{x_0, b}^T \W_{x_0} \X_{x_0,b})^{-1} \frac{1}{nb}\X_{x_0,b}^T \W_{x_0, b} \mathbf{R}
$$
where $\mathbf{R} = \left(R(X_1), \cdots, R(X_n)\right)^T$.  Now through a similar derivation, above equals to 
\begin{equation}
\begin{aligned}
e_3^T \left(\frac{1}{p_X(x)} \Omega_3^{-1} + O(b) + O_2\left(\sqrt{\frac{1}{nb}}\right)\right) \frac{1}{nb^3}\X_{x_0,b}^T \W_{x_0, b} \mathbf{R}
\end{aligned}
\label{eq::LPR::bias}
\end{equation}
where 
$$
\Omega^{-1}_3 = \left(\begin{array}{cccc} \cdots & \cdots & \cdots & \cdots \\ 
\cdots & \cdots & \cdots & \cdots \\ 
\frac{s_2 s_4^2 - s_2^2 s_6}{s_2^2 s_4^2 - s_4^3 - s_2^3 s_6 + s_2 s_4 s_6} & 0 & \frac{-s_4^2 + s_2 s_6}{s_2^2 s_4^2 - s_4^3 - s_2^3 s_6 + s_2 s_4 s_6} & 0 \\
\cdots & \cdots & \cdots & \cdots 
\end{array}\right).
$$
%
Here, we only show the third row in the $\Omega^{-1}$ since this is the only row that are related to our results here (we have a vector $e_3^T$
in front of it so other rows will be eliminated). 
Moreover, by a direct expansion,
$$
\frac{1}{nb^3}\X_{x_0,b}^T \W_{x_0, b} \mathbf{R} = \begin{pmatrix}
\frac{1}{nb^3}\sum_{i=1}^n R(X_i) K\left(\frac{X_i-x_0}{b}\right)\\
\frac{1}{nb^3}\sum_{i=1}^n R(X_i) \cdot \left(\frac{X_i-x_0}{b}\right)\cdot K\left(\frac{X_i-x_0}{b}\right)\\
\frac{1}{nb^3}\sum_{i=1}^n R(X_i) \cdot \left(\frac{X_i-x_0}{b}\right)^2\cdot K\left(\frac{X_i-x_0}{b}\right)\\
\frac{1}{nb^3}\sum_{i=1}^n R(X_i) \cdot \left(\frac{X_i-x_0}{b}\right)^3\cdot K\left(\frac{X_i-x_0}{b}\right)
\end{pmatrix}\\
$$
For $j = 0,1,2,3$, 
\begin{align*}
& \frac{1}{nb^3} \sum_{i=1}^n R(X_i)  \left(\frac{X_i - x_0}{b}\right)^j   K\left(\frac{X_i - x_0}{b}\right) \\
& = \frac{1}{b^3}\int R(t) \left(\frac{t - x_0}{b}\right)^j K\left(\frac{t - x_0}{n}\right)p_X(t) dt + O_2\left(\sqrt{\frac{1}{nb}}\right) \\
& = \frac{1}{b^2}\int R(x_0 + ub) u^j K(u) p_X(x_0 + ub) du + O_2\left(\sqrt{\frac{1}{nb}}\right)\\
& = \int \left(\frac{1}{2} r^{(2)}(x^{\star}) - \frac{1}{2} r^{(2)}(x_0)\right)u^{2+j} K(u) p_X(x_0 +ub) du + O_2\left(\sqrt{\frac{1}{nb}}\right) \\
& = O(b^{\delta_0}) + O_2\left(\sqrt{\frac{1}{nb}}\right)
\end{align*}
where $|x^{\star} - x_0| \leq ub$ is again from the mean value theorem.
Then the bias based on \eqref{eq::LPR::bias} is  
\begin{align*}
& \E[e_3^T \left(\frac{1}{p_X(x)} \Omega_3^{-1} + O(b) + O_2\left(\sqrt{\frac{1}{nb}}\right)\right) \frac{1}{nb^3}\X_{x_0,b}^T \W_{x_0, b} \mathbb{R}] = \\
& O(b^{\delta_0}) + O\left(\sqrt{\frac{1}{nb}}\right) = O(h^{\delta_0})+ O\left(\sqrt{\frac{1}{nh}}\right)
\end{align*}
given $\tau$ fixed,
which completes the proof.
\end{proof}

\begin{proof}[ of Lemma~\ref{lem::LS_BV}]
Recall from equation \eqref{eq::dLS} that the debiased local linear smoother is
$$
\hat{r}_{\tau,h}(x) = \hat{r}_h(x) - \frac{1}{2}\cdot c_K\cdot h^2\cdot \hat{r}^{(2)}_{h/\tau}(x).
$$
Under assumption (K3) and (R1) and by Lemma ~\ref{lem::Rate}, the bias and variance of $\hat{r}_h(x)$
is (by a similar derivation as the one described in Lemma~\ref{lem::KDE_BV})
\begin{equation*}
\begin{aligned}
\E\left( \hat{r}_h(x)\right) - r(x) &= \frac{h^2}{2}c_K\cdot r^{(2)}(x) + O(h^{2+\delta_0}) + h^2 O\left(\sqrt{\frac{1}{nh}}\right)\\
{\sf Var}( \hat{r}_h(x)) &= O_P\left(\frac{1}{nh}\right),
\end{aligned}
\end{equation*}
and the bias and variance of the second derivative estimator $ \hat{r}^{(2)}_{h/\tau}(x)$ is

\begin{equation*}
\begin{aligned}
\E\left( \hat{r}^{(2)}_{h/\tau}(x)\right) - r^{(2)}(x) &= O(h^{\delta_0}) + O\left(\sqrt{\frac{1}{nh}}\right)\\
{\sf Var}( \hat{r}^{(2)}_{h/\tau}(x)) &= O_P\left(\frac{1}{nh^5}\right).
\end{aligned}
\end{equation*}

Thus, given $\delta_0 \leq 2$, the bias of $\hat{r}_{\tau,h}(x)$ is 
\begin{align*}
\E\left(\hat{r}_{\tau,h}(x)\right) - r(x) & = \E\left(\hat{r}_{h}(x)\right) - r(x) - \frac{1}{2}\cdot c_K\cdot h^2\cdot \E\left(\hat{r}^{(2)}_{h/\tau}(x)\right)\\
& = \frac{h^2}{2}c_K\cdot r^{(2)}(x) + O(h^{2+\delta_0}) + h^2 O\left(\sqrt{\frac{1}{nh}}\right) \\
& - \frac{1}{2}\cdot c_K\cdot h^2\cdot \left( r^{(2)}(x) + O(h^{\delta_0}) + O\left(\sqrt{\frac{1}{nh}}\right)\right)\\
& = O(h^{2+\delta_0}) + h^2 O\left(\sqrt{\frac{1}{nh}}\right).
\end{align*}

The variance of $\hat{r}_{\tau,h}(x)$ is
\begin{align*}
{\sf Var}\left(\hat{r}_{\tau,h}(x)\right) &= {\sf Var}\left(\hat{r}_h(x) - \frac{1}{2}\cdot c_K\cdot h^2\cdot \hat{r}^{(2)}_{h/\tau}(x)\right)\\
&= {\sf Var}\left(\hat{r}_h(x)\right) + O\left(h^4\right)\cdot{\sf Var}\left(\hat{r}^{(2)}_{h/\tau}(x)\right)
- O\left(h^2\right) \cdot {\sf Cov}\left(\hat{r}_{\tau,h}(x), \hat{r}^{(2)}_{h/\tau}(x)\right)\\
&\leq O\left(\frac{1}{nh}\right) +O\left(\frac{h^4}{nh^5}\right) + O\left(h^2\cdot \sqrt{\frac{1}{nh}}\cdot \sqrt{\frac{1}{nh^5}}\right)\\
& = O\left(\frac{1}{nh}\right),
\end{align*}
which has proven the desired result.

\end{proof}

Before we move on to the proof of Lemma~\ref{lem::LS_empirical}, we first introduce some notations.
For a given point $x\in\D$, let 
\begin{align*}
\X_x &= \begin{pmatrix}
1 & (X_1-x) & (X_1-x)^2  &(X_1-x)^3 \\
1 & (X_2-x) & (X_2-x)^2  &(X_2-x)^3  \\
\vdots & \vdots &\vdots &\vdots\\
1 & (X_n-x) & (X_n-x)^2  &(X_n-x)^3 
\end{pmatrix}
\in \R^{n\times 4}\\
\X_{x,h} &= \begin{pmatrix}
1 & \left(\frac{X_1-x}{h}\right) & \left(\frac{X_1-x}{h}\right)^2&
\left(\frac{X_1-x}{h}\right)^3\\
1 & \left(\frac{X_2-x}{h}\right) & \left(\frac{X_2-x}{h}\right)^2&\left(\frac{X_2-x}{h}\right)^3\\
\vdots & \vdots &\vdots &\vdots\\
1 & \left(\frac{X_n-x}{h}\right) & \left(\frac{X_n-x}{h}\right)^2&\left(\frac{X_n-x}{h}\right)^3\\
\end{pmatrix}
\in \R^{n\times 4}\\
\W_x &= {\sf diag}\left(K\left(\frac{X_1-x}{h}\right),K\left(\frac{X_2-x}{h}\right),\cdots, K\left(\frac{X_n-x}{h}\right)\right)\in\R^{n\times n}\\
\Gamma_h & = {\sf diag}\left(1, h^{-1}, h^{-2}, h^{-3}\right)\in\R^{4\times 4}\\
\Y & = (Y_1,\cdots, Y_n)\in\R^n.
\end{align*}
Based on the above notations, 
the local polynomial estimator $\hat{r}^{(2)}_h(x)$
can be written as 
\begin{equation}
\begin{aligned}
\frac{1}{2!}\hat{r}^{(2)}_h(x) &= e_3^T (\X_x^T \W_x\X_x)^{-1}\X_x^T\W_x \Y\\
&= \frac{1}{h}e_3^T \Gamma_h \left(\frac{1}{nh}\Gamma_h\X_x^T \W_x\X_x\Gamma_h\right)^{-1}\frac{1}{n}\Gamma_h\X_x^T\W_x \Y\\
&= \frac{1}{h^3}e_3^T \left(\frac{1}{nh}\X_{x,h}^T \W_x\X_{x,h}\right)^{-1}\frac{1}{n}\X_{x,h}^T\W_x \Y
\end{aligned}
\label{eq::LPR::matrix}
\end{equation}
where $e_3^T = (0,0,1,0)$; see, e.g., \cite{fan1996local,wasserman2006all}.
Thus, a key element in the proof of Lemma~\ref{lem::LS_empirical} 
is deriving the asymptotic behavior of $(\frac{1}{nh}\X_{x,h}^T \W_x\X_{x,h})^{-1}$.

\begin{lem}
Assume (K1,K3) and (R1).
Then 
$$
\sup_{x\in\D}\left\|\frac{1}{nh}\X_{x,h}^T \W_x\X_{x,h} - p_X(x) \cdot \Omega_3\right\|_{\max} =O(h) + O_P\left(\sqrt{\frac{\log n}{nh}}\right).
$$
Thus, 
$$
\sup_{x\in\D}\left\|\left(\frac{1}{nh}\X_{x,h}^T \W_x\X_{x,h}\right)^{-1} - \frac{1}{p_X(x)} \Omega^{-1}_3\right\|_{\max} =O(h) + O_P\left(\sqrt{\frac{\log n}{nh}}\right).
$$
\label{lem::LPR::inv}
\end{lem}
\begin{proof}
We denote $\Xi_n(x,h) = \frac{1}{nh}\X_{x,h}^T \W_x\X_{x,h} $. 
Then direct computation shows that the $(j,\ell)$ element of the matrix $\Xi_n(x,h)$ is
$$
\Xi_n(x,h)_{j\ell} = \frac{1}{nh}\sum_{i=1}^n \left(\frac{X_i-x}{h}\right)^{j+\ell-2} K\left(\frac{X_i-x}{h}\right) 
$$
for $j,\ell = 1,2,3,4$.

Thus, the difference
\begin{align*}
\Xi_n(x,h)_{j\ell} - p_X(x) \Omega_{3,j\ell} = \underbrace{\Xi_n(x,h)_{j\ell} - \E\left(\Xi_n(x,h)_{j\ell}\right)}_{(I)}
+\underbrace{\E\left(\Xi_n(x,h)_{j\ell}\right) - p_X(x) \Omega_{3,j\ell}}_{(II)}.
\end{align*}
The first quantity (I) is about stochastic variation and the second quantity (II) is like bias in the KDE.

We first bound (II). 
By direct derivation,
\begin{align*}
\E\left(\Xi_n(x,h)_{j\ell}\right)  &= \frac{1}{h}\int\left(\frac{\omega-x}{h}\right)^{j+\ell-2} K\left(\frac{\omega-x}{h}\right) p_X(\omega) d\omega\\
& = \int u^{j+\ell-2} K\left(u\right) p_X(x+uh) du\\
& = p_X(x)\int u^{j+\ell-2} K\left(u\right)  + O(h)\\
& = p_X(x) \Omega_{3,j\ell} + O(h).
\end{align*}

Now we bound (I).
Let $\P_{X,n}$ denote the empirical measure and $\P_{X}$ denote the probability measure of the covariate $X$.
We rewrite the first quantity (I) as
$$
(I) = \frac{1}{h}\int\left(\frac{\omega-x}{h}\right)^{j+\ell-2} K\left(\frac{\omega-x}{h}\right) (d\P_{X,n}(\omega)-d\P_{X}(\omega)).
$$
This quantity can be uniformly bounded using the empirical process theory from \cite{Gine2002},
which proves that 
$$
(I) = O_P\left(\sqrt{\frac{\log n}{nh}}\right)
$$
uniformly for all $x\in\D$ under assumption (K3) and (R1).
Putting it altogether, we have proved
$$
\sup_{x\in\D}\left\|\Xi_n(x,h)_{j\ell} - p_X(x) \Omega_{3,j\ell}\right\| = O(h) + O_P\left(\sqrt{\frac{\log n}{nh}}\right).
$$
This works for all $j,\ell=1,2,3,4$, so we have proved this lemma.

\end{proof}

\begin{proof}[ of Lemma~\ref{lem::LS_empirical}]
{\bf Empirical approximation.}
Recall that 
$$
\hat{r}_{\tau,h}(x) = \hat{r}_h(x) - \frac{1}{2}\cdot c_K\cdot h^2\cdot \hat{r}^{(2)}_{h/\tau}(x),
$$
The goal is to prove that 
$\hat{r}_{\tau,h}(x)-r(x)$ can be uniformly approximated by
an empirical process.

By Lemma~\ref{lem::LS_BV},
the difference 
$
\hat{r}_{\tau,h}(x)-r(x)
$
is dominated by the stochastic variation $\hat{r}_{\tau,h}(x)-\E\left(\hat{r}_{\tau,h}(x)\right)$ when $nh^5\rightarrow c<\infty$. 
Thus, 
we only need to show that $\sqrt{nh}\left(\hat{r}_{\tau,h}(x)-\E\left(\hat{r}_{\tau,h}(x)\right)\right)\approx \sqrt{h}\G_n(\psi_x)$. 

Because 
$$
\hat{r}_{\tau,h}(x)-\E\left(\hat{r}_{\tau,h}(x)\right) = \hat{r}_h(x)-\E\left(\hat{r}_h(x)\right) - \frac{1}{2}\cdot c_K\cdot h^2\cdot \left(\hat{r}^{(2)}_{h/\tau}(x)-\E\left(\hat{r}^{(2)}_{h/\tau}(x)\right)\right).
$$
For simplicity, we only show the result of the second derivative part (the result of the first part 
can be proved in a similar way). Namely, we will prove
\begin{equation}
\frac{\sqrt{nh}}{2}\cdot c_K\cdot h^2\cdot \left(\hat{r}^{(2)}_{h/\tau}(x)-\E\left(\hat{r}^{(2)}_{h/\tau}(x)\right)\right) \approx \sqrt{h} \G_n(\psi_{x,2})
\label{eq::pf::LSe::1}
\end{equation}
uniformly for all $x\in \D$,
where $\psi_{x,2}(z_1,z_2) = \frac{1}{h p_X(x)}c_K\cdot\tau^3\cdot e_3^T \Omega_3^{-1}\Psi_{2,\tau x}(\tau z_1, z_2)$
is the second part of $\psi_x(z_1,z_2)$.

Recall from equation \eqref{eq::LPR::matrix} and apply Lemma~\ref{lem::LPR::inv},
\begin{align*}
\frac{1}{2} \hat{r}^{(2)}_{h/\tau}(x)
&= \frac{1}{b^2}e_3^T \left(\frac{1}{nb}\X_{x,b}^T \W_x\X_{x,b}\right)^{-1}\frac{1}{nb}\X_{x,h}^T\W_x \Y\\
&= \frac{1}{b^2}e_3^T\left(\frac{1}{p_X(x)} \Omega_3^{-1} + O(h) + O_P\left(\sqrt{\frac{\log n}{nh}}\right)\right)\frac{1}{nb}\X_{x,b}^T\W_x \Y
\end{align*}
where $b=h/\tau$.
The vector $\frac{1}{nb}\X_{x,b}^T\W_x \Y$ can be decomposed into
\begin{align*}
\frac{1}{nb}\X_{x,b}^T\W_x \Y &= \begin{pmatrix}
\frac{1}{nb}\sum_{i=1}^n Y_i K\left(\frac{X_i-x}{b}\right)\\
\frac{1}{nb}\sum_{i=1}^n Y_i\cdot \left(\frac{X_i-x}{b}\right)\cdot K\left(\frac{X_i-x}{b}\right)\\
\frac{1}{nb}\sum_{i=1}^n Y_i\cdot \left(\frac{X_i-x}{b}\right)^2\cdot K\left(\frac{X_i-x}{b}\right)\\
\frac{1}{nb}\sum_{i=1}^n Y_i\cdot \left(\frac{X_i-x}{b}\right)^3\cdot K\left(\frac{X_i-x}{b}\right)
\end{pmatrix}\\
& = \frac{\tau}{h}\int \Psi_{2,\tau x}(\tau z_1, z_2) d\P_{n}(z_1,z_2).
\end{align*}
Thus, by denoting $\P_n (e_3^T\Psi_{2,\tau x})= \int e_3^T\Psi_{2,\tau x}(\tau z_1, z_2) d\P_{n}(z_1,z_2)$,
\begin{align*}
\frac{1}{2}\hat{r}^{(2)}_{h/\tau}(x) &= \int\frac{\tau^3}{h^3\cdot p_X(x)}e_3^T \Omega_3^{-1} \Psi_{2,\tau x}(\tau z_1, z_2) d\P_{n}(z_1,z_2)\\
& + \frac{\tau^3}{h^3}e_3^T \Psi_{2,\tau x}(\tau z_1, z_2) d\P_n(z_1, z_2) \left(O(h) + O_P \left(\sqrt{\frac{\log n}{nh}}\right)\right) \\ 
& = \frac{1}{c_K\cdot h^2}\P_n(\psi_{x,2}) + \frac{2 \tau^3}{h^3} \P_n(e_3^T\Psi_{2,\tau x})\left(O(h) +O_P\left(\sqrt{\frac{\log n}{nh}}\right)\right).
\end{align*}
This also implies 
$$
\frac{1}{2}\E\left(\hat{r}^{(2)}_{h/\tau}(x)\right) = \frac{1}{c_K\cdot h^2}\P(\psi_{x,2}) + \frac{2 \tau^3}{h^3} \P(e_3^T\Psi_{2,\tau x})\left(O(h) +O\left(\sqrt{\frac{\log n}{nh}}\right)\right),
$$
where $\P (e_3^T\Psi_{2,\tau x})= \int e_3^T\Psi_{2,\tau x}(\tau z_1, z_2) d\P(z_1,z_2)$.
Based on the above derivations, the scaled second derivative
\begin{align*}
& \frac{\sqrt{nh}}{2}\cdot c_K\cdot h^2\cdot \left(\hat{r}^{(2)}_{h/\tau}(x)-\E\left(\hat{r}^{(2)}_{h/\tau}(x)\right)\right)
= \sqrt{nh} \cdot\left(\P_n(\psi_{x,2})-\P(\psi_{x,2})\right) + \\
&  \frac{\sqrt{n} \cdot c_K \cdot \tau^3}{\sqrt{h}} ( \P_n(e_3^T \Psi_{2, \tau x} - \P(e_3^T \Psi_{2, \tau x})) \left(O(h) + O_P \left(\sqrt{\frac{\log n}{nh}}\right) \right) \\
& = \sqrt{h} \G_n(\psi_{x,2}) + \frac{c_K \cdot \tau^3}{\sqrt{h}} \G_n(e_3^T \Psi_{2, \tau x}) \left(O(h) + O_P \left(\sqrt{\frac{\log n}{nh}}\right) \right).
\end{align*}
By the similar derivation, we obtain
$$
\sqrt{nh}(\hat{r}_h(x) - \E(\hat{r}_h(x)) = \sqrt{h} \G_n(\psi_{x,0}) + \frac{1}{\sqrt{h}} \G_n(e_1^T \Psi_{x,0}) \left(O(h) + O_P\left(\sqrt{\frac{ \log n}{nh}}\right)\right)
$$
where $\psi_{x,0}(z) = \frac{1}{p_X(x) h} e_1^T \Omega_1^{-1} \Psi_{0,x}(z)$. Then combined together,
\begin{align*}
& \sqrt{nh}(\hat{r}_{\tau, h}(x) - \E (\hat{r}_{\tau, h}(x)))  = \sqrt{h} \G_n(\psi_x) + \\
& \frac{1}{\sqrt{h}} \G_n(e_1^T \Psi_{x,0}- c_K \cdot \tau^3 e_3^T \Psi_{2, \tau x}) \left(O(h) + O_P\left(\sqrt{\frac{\log n}{nh}}\right)\right)
\end{align*}

Because this $O_P$ term is uniformly for all $x$ (Lemma~\ref{lem::LPR::inv}), we conclude that 
$$
\sup_{x\in\D}\left\| \frac{\sqrt{nh}(\hat{r}_{\tau, h}(x) - \E(\hat{r}_{\tau, h}(x)) - \sqrt{h} \G_n(\psi_x)}{\frac{1}{\sqrt{h}} \G_n(e_1^T \Psi_{x,0} - c_K \cdot \tau^3 e_3^T \Psi_{2,\tau x})}\right\|
= O(h) +O_P\left(\sqrt{\frac{\log n}{nh}}\right).
$$

{\bf Uniform bound.}
In the first assertion, we have shown that 
$$
\sqrt{nh}(\hat{r}_{\tau,h}(x)- \E(\hat{r}_{\tau, h}(x)) \approx \sqrt{h}\G_n(\psi_x).
$$
In the definition of $\psi_x(z_1,z_2)$ (equation \eqref{eq::z1z2}),
the dependency on $z_2$ can is linear, i.e., $\psi_x(z_1,z_2) = z_2 \cdot \psi_x'(z_1)$  
for some function $\psi_x'(z_1)$.
Thus, by defining $z_2\cdot \Psi'_{0,x}(z_1) = \Psi'_{0,x}(z_1,z_2)$ and $z_2\cdot \Psi'_{2,x}(\tau z_1) = \Psi'_{2,x}(\tau z_1,z_2)$,
we can rewrite the above expectation as
\begin{align*}
& \sqrt{h} \G_n(\psi_x)  = \sqrt{\frac{h}{n}} \sum_{i=1}^n \left(\psi_x(X_i, Y_i) - \E \psi_x(X_i, Y_i)\right)\\
& = \frac{1}{\sqrt{nh}} \sum_{i=1}^n \left( \frac{1}{p_X(x)} \left(e_1^T \Omega_1^{-1} \Psi_{0,x}(X_i, Y_i) -  c_K \cdot \tau^3 e_3^T \Omega_3^{-1} \Psi_{2, \tau x}(\tau X_i, Y_i) \right) - \E(\psi_x(X_i, Y_i))\right) \\
& = \frac{1}{\sqrt{nh}} \sum_{i=1}^n Y_i \left( \frac{1}{p_X(x)} \left(e_1^T \Omega_1^{-1} \Psi'_{0,x}(X_i) -  c_K \cdot \tau^3 e_3^T \Omega_3^{-1} \Psi'_{2, \tau x}(\tau X_i) \right) - \E(\psi_x'(X_i))\right).
\end{align*}

Then since $\Psi'_{0,x}$ and $\Psi'_{2,\tau x}$ are simply linear combinations of functions from 
$$
{\cal M}_3^{\dagger} = \left\{ y \mapsto \left(\frac{y - x}{h} \right)^\gamma K \left(\frac{y - x}{h} \right):  x \in \D, \gamma = 0, \ldots, 3, h > 0\right\}
$$
that is bounded if $K$ is gaussian kernel or any other compact supported kernel function, so we can apply the empirical process theory in \cite{Einmahl2005}
which proves that 
\begin{align*}
\sup_{x \in \D} \Bigg\|\sum_{i=1}^n Y_i \Big( \frac{1}{p_X(x)} &\bigg(e_1^T \Omega_1^{-1} \Psi'_{0,x}(X_i)-\\ 
 & c_K \cdot \tau^3 e_3^T \Omega_3^{-1} \Psi'_{2, \tau x}(\tau X_i) \bigg) - \E(\psi_x'(X_i) \Big)\Bigg\| \\
&\qquad =O_P(\sqrt{nh \log n})
\end{align*}
Thus, 
$$
\sup_{{x \in \D}}\|\sqrt{h} \G_n(\psi_x)\| =  O_P(\sqrt{\log n})
$$
and similarly, 
$$
\left\|\frac{1}{\sqrt{h}} \G_n(e_1^T \Psi_{x,0} - c_K \cdot \tau^3 e_3^T \Psi_{2, \tau x})\right\|_\infty = O_P(\sqrt{\log n})
$$
Therefore,
\begin{align*}
 \hat{r}_{\tau, h}(x) &- \E(\hat{r}_{\tau, h}(x))=\\ 
& \frac{\sqrt{h} \G_n(\psi_x) + \frac{1}{\sqrt{h}} \G_n(e_1^T \Psi_{x,0} - c_K \cdot \tau^3 e_3^T \Psi_{2, \tau x}) \left(O(h) + O_P\left(\sqrt{\frac{\log n}{nh}}\right)\right)}{\sqrt{nh}}
\end{align*}
Then 
\begin{align*}
 \|\hat{r}_{\tau, h}(x) - \E(\hat{r}_{\tau, h}(x))\|_\infty &\leq\left\|\frac{\sqrt{h}\G_n(\psi_x)}{\sqrt{nh}}\right\|_\infty + \\
&\quad  \left\|\frac{\frac{1}{\sqrt{h}} \G_n(e_1^T \Psi_{x,0} - c_K \cdot \tau^3 e_3^T \Psi_{2,\tau x})}{\sqrt{nh}}\right\|_\infty \left(O(h) + O_P\left(\sqrt{\frac{\log n}{nh}}\right)\right) \\
&  = O_P\left(\sqrt{\frac{\log n}{nh}}\right) + O_P\left(\sqrt{\frac{\log n}{nh}}\right) \left(O(h) + O_P\left(\sqrt{\frac{\log n}{nh}}\right)\right) \\
&= O_P\left(\sqrt{\frac{\log n}{nh}}\right)
\end{align*}
This, together with the bias in Lemma~\ref{lem::LS_BV},
implies 
$$
\|\hat{r}_{\tau,h}(x)-r(x)\|_\infty  = O(h^{2+\delta_0}) + h^2 O\left(\sqrt{\frac{1}{nh}}\right) + O_P\left(\sqrt{\frac{\log n}{nh}}\right)
$$
which completes the proof.
\end{proof}
{
\section{More Simulation Results for Density Case}	\label{sec::DE::sim::app}
In this section, we further evaluate the performance of our debiased approach by comparing it with undersmoothing (US), traditional bias correction (BC), and the variable-width debiased approach (Variable DE) proposed in remark \ref{rm::honest0} on a couple different density functions. 

We consider the following 5 scenarios that 
have been used in \cite{marron1992exact, calonico2018effect}:
\begin{align*}
& \text{Model 1}: x \sim \frac{1}{5} \cN(0, 1) + \frac{1}{5} \cN\left(\frac{1}{2}, \left(\frac{2}{3} \right)^2 \right) + \frac{3}{5} \cN\left(\frac{13}{12}, \left( \frac{5}{9} \right)^2 \right) \\
& \text{Model 2}: x \sim \frac{1}{2} \cN\left(-1, \left(\frac{2}{3} \right)^2 \right) + \frac{1}{2} \cN\left(1, \left(\frac{2}{3} \right)^2 \right) \\
& \text{Model 3}: x \sim \frac{3}{4} \cN(0, 1) + \frac{1}{4} \cN\left(\frac{3}{2}, \left(\frac{1}{3}\right)^2 \right) \\
& \text{Model 4}: x \sim \frac{3}{5} \cN \left(-1, \left( \frac{1}{4} \right)^2 \right) + \frac{2}{5} \cN \left(1, 1 \right) \\
& \text{Model 5}: x \sim \frac{3}{5} \cN \left(1.5, 1 \right) + \frac{2}{5} \cN \left(-1.5, 1\right)
\end{align*}
For each model, we consider $n=500,1000,2000$ and 
construct uniform confidence band within range of $[-2, 2]$ for each model.  
The bandwidth was selected by either cross validation with package \texttt{ks} \citep{duong2007ks} or rule of thumb (ROT).  
As mentioned previously, undersmoothing is performed with half the bandwidth of debiased approach.  Traditional bias correction involves a second bandwidth for consistently estimating the second derivative. }

{
Table \ref{table::appendix_kde_comparison_coverage} and \ref{table::appendix_kde_comparison_width} reports the simulation results for the above 5 models.  Overall our debiased approach achieves very accurate coverages in almost all the cases except when bandwidth are chosen by the rule of thumb and the scenarios are Model 3 and 4, which are two of the hardest cases among the 5 models that we tested (one can see that
	almost all methods with ROT fail in this case).  
	With these observations, we recommend to use cross validation for bandwidth selection over rule of thumb, which in a way ``adapts" to the specific smoothness of each density function.  It is also worth mentioning that undersmoothing combined with bootstrap also works pretty well in our simulations.  In all the cases we considered, the width of confidence bands is wider for undersmoothing than the debiased approach, which makes our approach more favorable than undersmoothing. Traditional bias correction always undercovers in our simulations, which may be caused by the fact that the requirement of an additional bandwidth makes the problem more complicated. Finally, the variable-width debiased approach also works well in most cases, especially with bandwidth selected by cross validation.  Again for Model 4, it seems to be undercovering a bit.  This is due to the fact that the density value is very close to 0 when $x$ is close to -2 under Model 4.  If we restrict the range for uniform confidence band to be $[-1.5, 2]$, then the empirical coverage for $n = 500, 1000, 2000$ would be $[0.972, 0.956, 0.948]$ with bandwidth selected by cross validation.   It seems that the original debiased approach is more robust to the case where the density value is close to 0. Another interesting observation is that the variable-width confidence band from the debiased estimator has a similar but slightly smaller averaged width compared to the fixed width method.
\begin{table}[H]
\caption{Empirical Coverage of 95\% simultaneous confidence band}
\begin{tabular}{c c c  c  c  c  c  }
\hline
 & & & \multicolumn{4}{c}{Empirical coverage}\\
 \hline
 Model & n & Bw Selection & US & BC & Debiased & Variable DE  \\
 \hline
 1 & 500 & CV & 0.953 &0.913 & 0.954 & 0.949\\
    & & ROT & 0.949 &0.937 &0.951 & 0.96\\
    & 1000 & CV & 0.949 & 0.906&0.950 &0.977\\
    & & ROT & 0.941 & 0.936&0.948  &0.979\\
    & 2000 & CV & 0.949 & 0.913&0.953 &0.97\\
    & & ROT & 0.947 & 0.929 & 0.948 &0.962\\ 
 \hline
  2 & 500 & CV & 0.957 & 0.914 & 0.949 & 0.953\\
    & & ROT & 0.953 & 0.833 &0.956 & 0.956\\
    & 1000 & CV & 0.961 & 0.912 &0.958 & 0.949\\
    & & ROT & 0.947 & 0.83 &0.949  & 0.955\\
    & 2000 & CV & 0.95 & 0.916 &0.95 & 0.958\\
    & & ROT & 0.937 & 0.84 & 0.953 & 0.954\\ 
 \hline
   3 & 500 & CV & 0.933 & 0.863 & 0.924 & 0.939\\
    & & ROT & 0.9 & 0.324 &0.847 &0.89 \\
    & 1000 & CV & 0.951 & 0.883&0.947 &0.958\\
    & & ROT & 0.892 & 0.196&0.874  & 0.911\\
    & 2000 & CV & 0.931 & 0.897&0.932 & 0.94\\
    & & ROT & 0.868 & 0.171 & 0.877 & 0.911\\ 
 \hline
   4 & 500 & CV & 0.942 &0.88 & 0.951 & 0.781\\
    & & ROT & 0.01 & 0&0 & 0 \\
    & 1000 & CV & 0.951 & 0.898&0.942 & 0.787 \\
    & & ROT & 0.008 & 0&0  &0\\
    & 2000 & CV & 0.941 & 0.897&0.95 & 0.827\\
    & & ROT & 0.008 & 0& 0 &0 \\ 
 \hline
   5 & 500 & CV & 0.956 &0.907 & 0.956 & 0.939\\
    & & ROT & 0.944 & 0.84&0.95 & 0.937\\
    & 1000 & CV & 0.955 &0.911 &0.954 &0.955\\
    & & ROT & 0.946 &0.844 &0.95  &0.946\\
    & 2000 & CV & 0.947 & 0.911&0.944 &0.944\\
    & & ROT & 0.934 & 0.853 & 0.942 & 0.946\\ 
\end{tabular}
\label{table::appendix_kde_comparison_coverage}
\end{table}
\begin{table}[H]
\caption{Average width of 95\% simultaneous confidence band}
\begin{tabular}{c c c  c  c  c  c }
\hline
 & & & \multicolumn{4}{c}{Average Confidence Band width} \\
 \hline
 Model & n & Bw Selection & US  & BC & Debiased & Variable DE \\
 \hline
 1 & 500 & CV & 0.133 & 0.189 & 0.113 & 0.134 \\
    &  & ROT & 0.118 & 0.078 & 0.100 & 0.106  \\
    & 1000 & CV  & 0.103 & 0.114 & 0.089 & 0.082 \\
    &  & ROT & 0.092 & 0.061 & 0.079 & 0.071\\
    & 2000 & CV  & 0.081 & 0.078 &0.070 & 0.060\\
    &  & ROT & 0.072 & 0.047 & 0.062 & 0.053\\ 
 \hline
  2 & 500 & CV & 0.096 & 0.148& 0.083 & 0.086\\
    &  & ROT & 0.077 & 0.052 & 0.066 &  0.067\\
    & 1000 & CV & 0.076 & 0.097 & 0.066 & 0.066\\
    &  & ROT & 0.060 & 0.040& 0.052 & 0.052\\
    & 2000 & CV & 0.059 & 0.065 &0.052 & 0.052\\
    &  & ROT  & 0.047 & 0.031 & 0.040 & 0.040\\ 
 \hline
   3 & 500 & CV & 0.110 & 0.177 & 0.095 & 0.099\\
    &  & ROT & 0.086 & 0.058& 0.073 & 0.074 \\
    & 1000 & CV & 0.088 & 0.111& 0.077 &0.077 \\
    &  & ROT & 0.067 & 0.045& 0.057 & 0.057\\
    & 2000 & CV & 0.071 & 0.077 &0.062 & 0.061\\
    &  & ROT & 0.052 & 0.035 & 0.045 & 0.044\\ 
 \hline
   4 & 500 & CV  & 0.219 &0.244 & 0.186 & 0.183\\
    & & ROT & 0.090 & 0.053 & 0.071 & 0.064 \\
    & 1000 & CV & 0.179 & 0.159 & 0.154 & 0.128\\
    & & ROT & 0.072 & 0.042  & 0.057 & 0.049 \\
    & 2000 & CV & 0.143 & 0.116 &0.125 & 0.093\\
    & & ROT & 0.057 & 0.034 & 0.046 & 0.038\\ 
 \hline
   5 & 500 & CV & 0.065 &0.097 & 0.056 & 0.055 \\
    & & ROT & 0.052 & 0.035 & 0.044 & 0.043 \\
    & 1000 & CV & 0.051 & 0.066 & 0.044 & 0.043\\
    & & ROT & 0.040 &0.027  & 0.035 & 0.034 \\
    & 2000 & CV & 0.040 & 0.044&0.035 & 0.033  \\
    & & ROT & 0.031 & 0.021 & 0.027 &  0.026 \\ 
\end{tabular}
\label{table::appendix_kde_comparison_width}
\end{table}

\begin{figure}[H]
\centering
\includegraphics[height=1.5in]{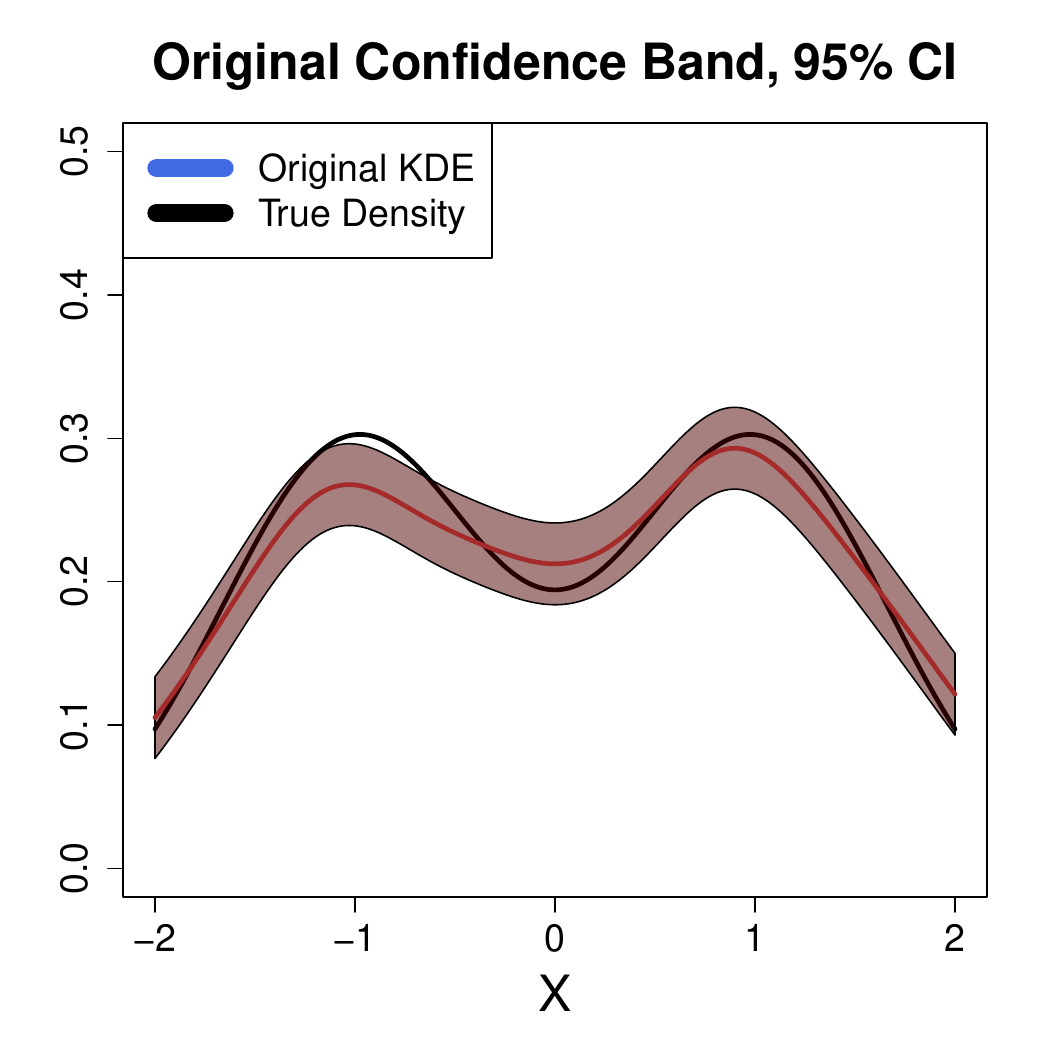}
\includegraphics[height=1.5in]{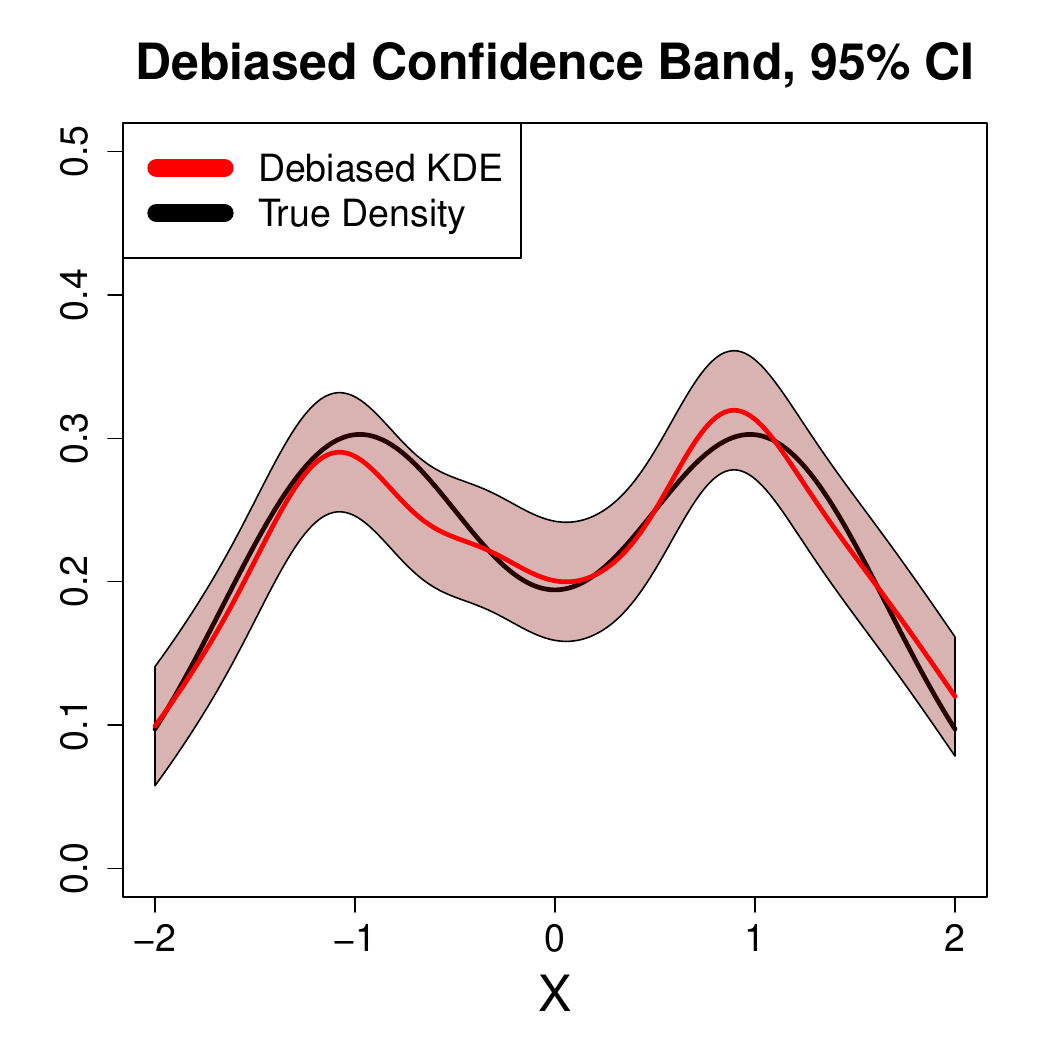}
\includegraphics[height=1.5in]{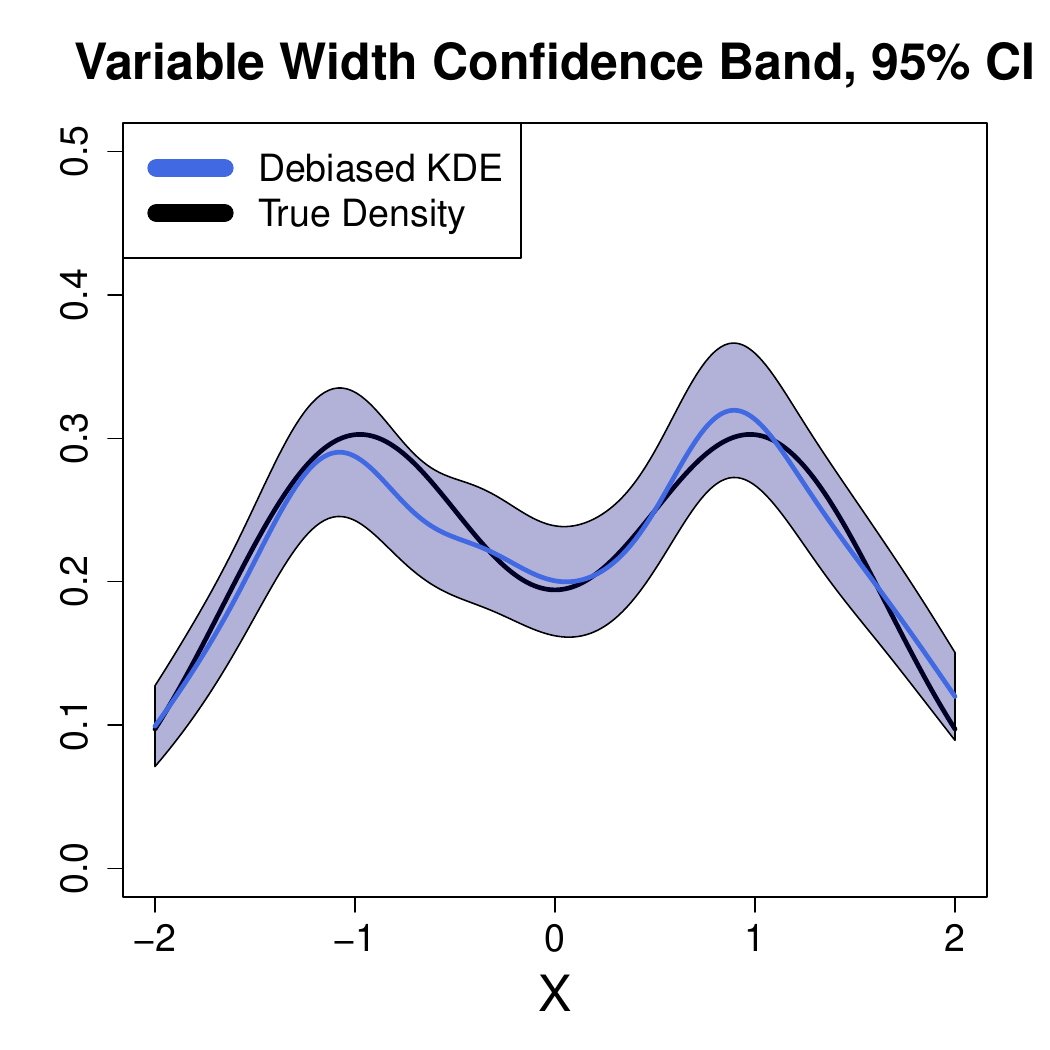}\\
\caption{From left to right, simultaneous confidence bands by bootstrapping original kernel density estimator, debiased density estimator and the variable-width confidence band.}
\label{fig::variable_length_confidence_band}
\end{figure}

As a case study, 
we plot simultaneous confidence bands for one instance of Model 2 with three approaches: 1. bootstrapping the original KDE, 2. our debiased estimator, and 3. the variable length debiased estimator 
in Figure \ref{fig::variable_length_confidence_band}. 
The confidence band of bootstrapping the original kernel density estimator is narrowest. The confidence band with from bootstrapping the original KDE is around 0.028 in this case, but it does not cover the whole density function within range [-2, 2].  Our debiased confidence band has width 0.042 and indeed cover the whole density function.  The variable-width confidence band has width that roughly increases as the density value increases (which is expectd from the theory).  The far left side has width 0.028, which achieves the width as bootstrapping the original kernel density. The first bump has width about 0.045, which is wider than the debiased confidence band.  The width at the valley (around value 0) is about 0.038 and the width at the second bump is around 0.046. Finally, the width at the far right side is about 0.031, which is again pretty close to the width of bootstrapping the original kernel estimator.
}

{
\section{More Simulation Results for Regression Case}	\label{sec::REG::sim::app}
In this section,  we reports more simulation results for the regression case.  We evaluate our debiased approach by comparing it with undersmoothing (US), traditional bias correction (BC),  locfit on several functions.  
We consider the scenarios from the supplement to \citep{calonico2018effect} such that
\begin{align*}
Y = m(x) + \epsilon;  \quad x \sim {\cal U}[-1, 1] \quad \epsilon \sim \cN(0, 1)
\end{align*}
with $m(x)$ being as follow:
\begin{align*}
& \text{Model 1}: m(x) = \sin(4x) + 2\exp(-64x^2) \\
& \text{Model 2}: m(x) = 2x + 2 \exp(-64x^2) \\
& \text{Model 3}: m(x) = 0.3 \exp(-4(2x + 1)^2) + 0.7\exp(-16(2x-1)^2) \\
& \text{Model 4}: m(x) = x + 5 \frac{1}{\sqrt{2\pi}} \exp(- (10x)^2 / 2) \\
& \text{Model 5}: m(x) = \frac{\sin(\pi x / 2)}{1 + 2x^2 [\text{sign}(x) + 1]}
\end{align*}
We vary the sample size from $n=500,1000$ to $2000$.
The confidence band is constructed on $[-0.9, 0.9] \subset [-1, 1]$. 
For the traditional bias correction, we use cross validation to choose the bandwidth for both function estimation and derivative estimation.  To be more specific, we use cross validation to choose the bandwidth for the local linear smoother to estimate the regression function, and again use cross validation to choose the bandwidth for the third order local polynomial regression to estimate the second order derivative.  Then we apply our bootstrap strategy to the debiased estimator to obtain the uniform confidence band. }

{
Table \ref{table::lpr_comparison_complete_coverage} and \ref{table::lpr_comparison_complete_width} show the final simulation result.  Overall, the undersmoothing, bias correction and debiased approaches all perform well with the cross validated bandwidth. They almost always achieve the nominal coverage except Model 5. With Model 5, the rule of thumb bandwidth is slightly better than cross validation bandwidth and traditional bias correction undercovers a bit.  The above analysis suggests that we should be using cross validation to choose bandwidth in general case. Locfit does not seem to be a good approach because it always suffers from undercoverage.  The debiased approach also performs much better with rule of thumb bandwidth than the other two approaches though it still undercovers.  This suggests that our debiased approach is more robust to bandwidth selection than other approaches. }

{
In terms of the width of confidence band,  again the traditional bias correction achieves the narrowest band with the bootstrap strategy, our debiased approach comes the next, and the undersmoothing gives the widest confidence band.  The consistent estimation of the bias seems to be improving the width of confidence bands.
\begin{table}[H]
\caption{Empirical Coverage and average width of 95\% simultaneous confidence band}
\begin{center}
\begin{tabular}{c c c c  c  c  c }
\hline
 & & & \multicolumn{4}{c}{Empirical coverage}\\
 \hline
 Model & n & Bw Selection & US & locfit & BC & Debiased\\
 \hline
 1 & 500 & CV & 1  & 0.77 & 0.985 & 0.993\\
  & & ROT & 0.498 & 0   &  0 &  0.805 \\
& 1000 & CV & 0.997 & 0.812 & 0.981 & 0.988 \\
  & & ROT & 0.34 & 0 &  0 &  0.827 \\
& 2000 & CV & 0.994 & 0.818 & 0.977 & 0.98 \\
 & & ROT & 0.263 & 0 & 0 & 0.841 \\
 \hline
  2 & 500 & CV & 1  & 0.76 & 0.985 & 0.993 \\
  & & ROT & 0.118 & 0   &  0 &  0.602 \\
& 1000 & CV & 0.997 & 0.807 & 0.98 & 0.989 \\
  & & ROT & 0.038 & 0 &  0 &  0.593 \\
& 2000 & CV & 0.994 & 0.803 & 0.977 & 0.979 \\
 & & ROT & 0.021 & 0 & 0 & 0.611 \\
  \hline
  3 & 500 & CV & 0.995  & 0.74 & 0.958 & 0.979 \\
  & & ROT & 0.949 & 0   &  0.451 &  0.932 \\
& 1000 & CV & 0.994 & 0.786 & 0.961 & 0.979 \\
  & & ROT & 0.927 & 0 &  0.342 &  0.928\\
& 2000 & CV & 0.982 & 0.789 & 0.958 & 0.971\\
 & & ROT & 0.908 & 0 & 0.27 & 0.943\\
 \hline
   4 & 500 & CV & 1  & 0.791 & 0.983 & 0.992 \\
  & & ROT & 0.448 & 0   &  0.029 &  0.805 \\
& 1000 & CV & 0.996 & 0.812 & 0.981 & 0.988\\
  & & ROT & 0.314 & 0 &  0.005 &  0.827\\
& 2000 & CV & 0.992 & 0.82 & 0.974 & 0.977  \\
 & & ROT & 0.24 & 0 & 0 & 0.833 \\
  \hline
   5 & 500 & CV & 0.974  & 0.852 & 0.955& 0.956\\
  & & ROT & 0.973 & 0.938   &  0.969&  0.967  \\
& 1000 & CV & 0.967 & 0.874 & 0.942 & 0.959 \\
  & & ROT & 0.976 & 0.927 &  0.965 &  0.97\\
& 2000 & CV & 0.969 & 0.858 & 0.942 & 0.948 \\
 & & ROT & 0.966 & 0.927 & 0.953 & 0.963 \\
\end{tabular}
\end{center}
\label{table::lpr_comparison_complete_coverage}
\end{table}
\begin{table}[H]
\caption{Average width of 95\% simultaneous confidence band}
\begin{center}
\begin{tabular}{c c c c  c  c  c }
\hline
 & & & \multicolumn{4}{c}{Average Confidence Band width} \\
 \hline
 Model & n & Bw Selection & US & locfit & BC & Debiased \\
 \hline
 1 & 500 & CV & 0.301 & 0.090& 0.113 & 0.144 \\
  & & ROT & 0.091  &  0.058 & 0.092 & 0.096 \\
& 1000 & CV & 0.135 & 0.067& 0.077 & 0.099 \\
  & & ROT& 0.063 & 0.042 & 0.058 & 0.063\\
& 2000 & CV & 0.088 & 0.050 & 0.055 & 0.072 \\
 & & ROT & 0.046 & 0.030 & 0.038 & 0.043 \\
 \hline
  2 & 500 & CV & 0.297 & 0.089& 0.112 & 0.143 \\
  & & ROT & 0.084  &  0.057 & 0.099 & 0.102 \\
& 1000 & CV & 0.134 & 0.066& 0.077 & 0.099 \\
  & & ROT & 0.058 & 0.040 & 0.063 & 0.065\\
& 2000 & CV & 0.088 & 0.050 & 0.055 & 0.072 \\
 & & ROT & 0.042 & 0.029 & 0.040 & 0.043 \\
  \hline
  3 & 500 & CV & 0.144 & 0.067 & 0.080 & 0.101 \\
  & & ROT & 0.084  &  0.048 & 0.060 & 0.071 \\
& 1000 & CV & 0.091 & 0.051 & 0.058 & 0.073 \\
  & & ROT & 0.058 & 0.036 & 0.042 & 0.050\\
& 2000 & CV & 0.064 & 0.039 & 0.042 & 0.054 \\
 & & ROT & 0.043 & 0.027 & 0.031 & 0.037 \\
 \hline
   4 & 500 & CV & 0.255 & 0.085 & 0.105 & 0.134 \\
  & & ROT & 0.084  &  0.054 & 0.085 & 0.089 \\
& 1000 & CV  & 0.125 & 0.063 & 0.072 & 0.094 \\
  & & ROT & 0.058 & 0.039 & 0.054 & 0.058\\
& 2000 & CV & 0.083 & 0.048 & 0.052 & 0.068 \\
 & & ROT & 0.043 & 0.029 & 0.036 & 0.043 \\
  \hline
   5 & 500 & CV & 0.065 & 0.039 & 0.045 & 0.071 \\
  & & ROT & 0.076  &  0.044 & 0.050 & 0.068 \\
& 1000 & CV & 0.048 & 0.030 & 0.033 & 0.047 \\
  & & ROT & 0.058 & 0.034 & 0.037 & 0.048\\
& 2000 & CV & 0.035 & 0.023 & 0.025 & 0.033 \\
 & & ROT & 0.041 & 0.026 & 0.027 & 0.036 \\
\end{tabular}
\end{center}
\label{table::lpr_comparison_complete_width}
\end{table}
}
\bibliographystyle{abbrvnat}
\bibliography{DKDE1.bib}

\end{document}